\documentclass[a4paper,UKenglish,cleveref, autoref, thm-restate]{lipics-v2021}

\title{Optimal strategies in concurrent reachability games}

\titlerunning{Concurrent reachability games}

\author{Benjamin Bordais, Patricia Bouyer and Stéphane Le
  Roux}{Université Paris-Saclay, CNRS, ENS Paris-Saclay, LMF, 91190
  Gif-sur-Yvette, France}{}{}{}

\authorrunning{B. Bordais, P. Bouyer and S. Le Roux}

\Copyright{Benjamin Bordais, Patricia Bouyer and Stéphane Le Roux}



\category{} 

\relatedversion{} 
\nolinenumbers

\usepackage{stmaryrd}
\renewcommand{\restriction}[2]{\left.#1\right|_{#2}}
\newcommand{\norm}[2]{\left\lVert#1 - #2\right\rVert}

\newtheorem*{proofs}{Proof Sketch}

\newcommand{\lfp}{\mathsf{m}}
\newcommand{\formNF}{\mathcal{F}}
\newcommand{\gameNF}[2]{#1^{#2}}

\newcommand{\outComeNF}{\mathsf{O}}
\newcommand{\outCNF}{\varrho}

\newcommand{\limval}[1]{\ensuremath{\tilde{#1}}}

\newcommand{\A}{\mathsf{A} }
\newcommand{\B}{\mathsf{B} }

\newcommand{\setA}{A}
\newcommand{\setB}{B}

\newcommand{\dNorm}{\ensuremath{\tau}}

\newcommand{\s}{\ensuremath{\mathsf{s}}}
\newcommand{\St}{\ensuremath{\mathsf{St}}}

\newcommand{\OS}{\ensuremath{\mathsf{MaxQ}}}
\newcommand{\SubOS}{\ensuremath{\mathsf{SubMaxQ}}}
\newcommand{\Opt}{\ensuremath{\mathsf{Opt}}}
\newcommand{\Gex}{\ensuremath{\mathsf{Gd}}}
\newcommand{\Bex}{\ensuremath{\mathsf{Bd}}}

\newcommand{\Bad}{\ensuremath{\mathsf{Bad}}}
\newcommand{\Prg}{\ensuremath{\mathsf{Prog}}}
\newcommand{\Rsk}{\ensuremath{\mathsf{Risk}}}
\newcommand{\Eff}{\ensuremath{\mathsf{Eff}}}
\newcommand{\Scr}{\ensuremath{\mathsf{Sec}}}

\newcommand{\Dist}{\mathcal{D} }

\newcommand{\Supp}{\mathsf{Supp} }
\newcommand{\distribSet}{\mathsf{D} }
\newcommand{\distribFunc}{\mathsf{dist} }
\newcommand{\prob}[2]{\mathbb{P}^{#1}_{#2} }
\newcommand{\probTrans}[1]{p^{#1} }

\newcommand{\val}[2]{\chi^{#1}_{#2} }

\newcommand{\Aconc}{\mathcal{C} }

\newcommand{\Games}[2]{\langle #1,#2 \rangle }
\newcommand{\AConc}{\ensuremath{\langle 
		\setA,\setB,Q,\distribSet,\delta,\distribFunc \rangle}}

\newcommand{\SetStrat}[2]{\ensuremath{\mathsf{S}_{#1}^{#2} }}
\newcommand{\SetPosStrat}[2]{\ensuremath{\mathsf{PS}_{#1}^{#2} }}

\newcommand{\head}{\ensuremath{\mathsf{lt}}}

\newcommand{\outM}{\ensuremath{\mathsf{out}}}
\newcommand{\va}{\ensuremath{\mathsf{val}}}
\newcommand{\N}{\ensuremath{\mathbb{N}}}

\newcommand{\R}{\ensuremath{\mathbb{R}}}

\begin{document}

\maketitle

\begin{abstract}
	We study two-player reachability games on finite graphs. At each state the interaction between the players is concurrent and there is a stochastic Nature. Players also play stochastically. The literature tells us that 1) Player $\B$, who wants to avoid the target state, has a positional strategy that maximizes the probability to win (uniformly from every state) and 2) from every state, for every $\varepsilon > 0$, Player $\A$ has a strategy that maximizes up to $\varepsilon$ the probability to win. Our work is two-fold.
	
	First, we present a double-fixed-point procedure that says from which state Player $\A$ has a strategy that maximizes (exactly) the probability to win. This is computable if Nature's probability distributions are rational. We call these states \emph{maximizable}. Moreover, we show that for every $\varepsilon > 0$, Player $\A$ has a positional strategy that maximizes the probability to win, exactly from maximizable states and up to $\varepsilon$ from sub-maximizable states.
	
	Second, we consider three-state games with one main state, one target, and one bin. We characterize the \emph{local interactions} at the main state that guarantee the existence of an optimal Player $\A$ strategy. In this case there is a positional one. It turns out that in many-state games, these local interactions also guarantee the existence of a uniform optimal Player $\A$ strategy. In a way, these games are well-behaved by design of their elementary bricks, the local interactions. It is decidable whether a local interaction has this desirable property.
\end{abstract}

	\section{Introduction}
	\label{sec:introduction}

\subparagraph*{Stochastic concurrent games.}

Games on graphs are an intensively studied mathematical tool, with
wide applicability in verification and in particular for the
controller synthesis problem, see for
instance~\cite{thomas02,BCJ18}. We consider two-player stochastic
concurrent games played on finite graphs. For simplicity (but this is
with no restriction), such a game is played over a finite bipartite
graph called an arena: some states belong to Nature while others
belong to the players. Nature is stochastic, and therefore assigns a
probabilistic distribution over the players' states. In each players'
state, a local interaction between the two players (called Player
$\A$ and Player $\B$) happens, specified 
by a two-dimensional table. Such an interaction is resolved as
follows: Player $\A$ selects a probability
distribution over the rows 
of the table while Player $\B$ selects a probability distribution over the columns
of the table; this results into a distribution over the cells of the
table, each one pointing to a Nature state of the graph. An example of
game arena is given in Figure~\ref{fig:arbitrarilyClose}:
circle states are players' while square states are Nature's; note that dashed
arrows assign only probability $1$ to a next state in this example
(but in general could give probabilities to several states).

\begin{figure}
	\begin{minipage}[t]{0.6\linewidth}
		\includegraphics[scale=1]{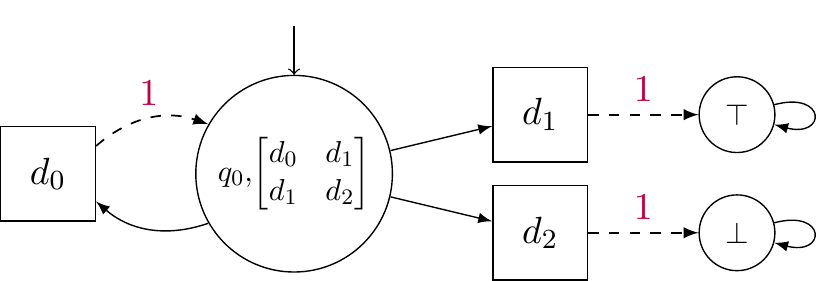}
		\caption{The game starts in $q_0$ with two actions available for each player. Player $\A$ wins if the state $\top$ is reached.}
		\label{fig:arbitrarilyClose}
	\end{minipage}
        \hfill
	\begin{minipage}[t]{0.38\linewidth}
		\centering
		\includegraphics[scale=1.5]{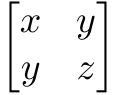}
		\caption{The local interaction at $q_0$ up to a renaming of the outcomes.}
		\label{fig:local_interaction_qzero}
	\end{minipage} \hfill
\end{figure}

Globally, the game proceeds as follows: starting at an initial state
$q_0$, the two players play the local interaction of the current
state, and the joint choice determines (stochastically) the next
Nature state of the game, itself moving randomly to players' states;
the game then proceeds subsequently from the new players' state. The
way players make choices is given by strategies, which, given the
sequence of states visited so far (the so-called history), assign
local strategies for the local interaction of the state the game is
in.
For application in controller synthesis, strategies will correspond to
controllers, hence it is desirable to have strategies simple to
implement. We will be in particular interested in strategies which are
\emph{positional}, that is, strategies which only depend on the
current state of the game, not on the whole history.
%
When each player has fixed a strategy (say $\s_{\A}$ for Player $\A$
and $\s_{\B}$ for Player $\B$), this defines a probability
distribution $\prob{q_0}{\s_\A,\s_\B}$ over infinite sequences of
states of the game.
The objectives of the two players are opposite (we assume a zero-sum
setting): together with the game, a measurable set $W$ of infinite
sequences of states is fixed; the objective of Player $A$ is then to
maximize the probability of $W$ while the objective of Player $B$ is
to minimize this probability.

Back to the example of Figure~\ref{fig:arbitrarilyClose}, assume
Player $\A$ (resp. $\B$) plays the first row (resp. column) with
probability $p_\A$ (resp. $p_\B$), then the probability 
to move to $\top$ is $p_\A + p_\B -2 p_\A p_\B$.
If Player $\A$ repeatedly plays the same strategy at $q_0$ with
$p_\A<1$, then the probability to reach $\top$ will lie between $p_\A$
and $1$, depending on Player $\B$; however, if she plays $p_\A=1$,
then by playing $p_\B=1$, Player $\B$ enforces staying in $q_0$, hence
reaching $\top$ with probability $0$.

\subparagraph*{Values and (almost-)optimal strategies.}

As mentioned above, Player $\A$ wants to maximize the probability of
$W$, while Player $\B$ wants to minimize this probability. Formally,
given a strategy $\s_{\A}$ for Player $\A$, its value is measured by
$\inf_{\s_{\B}} \prob{q_0}{\s_\A,\s_\B}(W)$, and Player $\A$ wants to
maximize that value. Dually, given a strategy $\s_{\B}$ for Player
$\B$, its value is measured by
$\sup_{\s_{\A}} \prob{q_0}{\s_\A,\s_\B}(W)$, and Player $\B$ wants to
  minimize that value. Following Martin's determinacy theorem for
  Blackwell games~\cite{martin98}, it actually holds that when $W$ is
  Borel, then the game has a \emph{value} given by
\[
  \val{}{q_0} = \sup_{\s_{\A}} \ \inf_{\s_{\B}}
  \prob{q_0}{\s_\A,\s_\B}(W) = \inf_{\s_{\B}} \ \sup_{\s_{\A}}
  \prob{q_0}{\s_\A,\s_\B}(W)
\]
While this ensures the existence of almost-optimal strategies (that
is, $\varepsilon$-optimal strategies for every $\varepsilon>0$) for
both players, it says nothing about the existence of optimal
strategies, which are strategies achieving $\val{}{q_0}$. In general,
as already mentioned in~\cite{everett57}, optimal strategies may not
exist. Indeed assuming a reachability objective with target $\top$,
the game in Figure~\ref{fig:arbitrarilyClose} is such that
$\val{}{q_0}=1$, however Player $\A$ can only achieve $1-\varepsilon$
for every $\varepsilon>0$ by playing repeatedly at $q_0$ the first
row of the table with probability $1-\varepsilon$ and the second row
with probability $\varepsilon$, but Player $\A$ cannot achieve
$1$. 

\subparagraph*{Our setting.}

In this paper we focus on reachability games, that is, $W$ is a
reachability condition. 
They are a special case of recursive games (where targets are assigned
payoffs), as studied in~\cite{everett57}. As such, they enjoy several
nice properties: (i) Player $\A$ has positional almost-optimal
strategies; (ii) Player $\B$ has 
positional optimal strategies~\cite{AM04b}.
These properties are specific to reachability games (or slight
generalizations thereof), and this is for instance not the case of
B\"uchi games, see~\cite[Thm. 2]{AM04b}.

Our goal is to study \emph{maximizable} and \emph{sub-maximizable} states in
(reachability) games: maximizable (resp. sub-maximizable) states are states
from which optimal strategies exist (resp. no optimal strategies
exist). Our contributions are then mostly twofolds:
\begin{enumerate}
\item We characterize via a double-fixed-point procedure maximizable and
  sub-maximizable states. This characterization cautiously analyzes when
  and why no optimal strategies will exist. Back to the example of
  Figure~\ref{fig:arbitrarilyClose}, we realize that no optimal
  strategy exists 
  since at the limit of $\varepsilon$-optimal
    strategies, i.e. when Player $\A$ plays the first row
    almost-surely, Player $\B$ can enforce cycling back to $q_0$,
    hence disabling state $\top$.
  This simple analysis close to the target has to be
  propagated carefully in the game, in which some strategies which are
  designated as risky (since they ultimately lead to such a situation)
  have to be avoided.

  As a byproduct of our construction, we have
  Theorem~\ref{thm:positional_optimal}, which establishes that one can
  build almost-optimal positional strategies, which are actually
  optimal where they can be. This refines the result
  of~\cite{everett57} which did not ensure optimality where it could.

  A consequence of that construction is that maximizable and sub-maximizable
  states can be computed 
  under slight assumptions
  , and that
  witness positional strategies can be computed as well. For these
  results we rely on Tarski's decidability result of the theory of the
  reals \cite{RE89}.

  We also show that our result cannot be extended to games with
  countably many states by exhibiting such a game in which an optimal
  strategy exists, but there is no optimal positional strategy.
 
\item Local interactions played by the players are abstracted into
  game forms, where cells of the matrix are now seen as variables
  (some of them being equal). For instance, the game form associated
  with state $q_0$ in the running example has three outcomes: $x$, $y$
  and $z$, and it is given in
  Figure~\ref{fig:local_interaction_qzero}.
  Game forms can be seen as elementary bricks that can be used to
  build games on graphs. We can embed such a brick into various
  three-states games with one main state, one target, and one bin (as
  is done in Figure~\ref{fig:arbitrarilyClose} for the interaction of
  Figure~\ref{fig:local_interaction_qzero}).  We characterize the
  \emph{local interactions} at the main state that guarantee the
  existence of an optimal Player $\A$ strategy. In this case there is
  a positional one. It turns out that in many-state games, these local
  interactions also guarantee the existence of a uniform optimal
  Player $A$ strategy. In a way, these games are well-behaved by
  design of their elementary bricks, the local interactions. It is
  decidable whether a local interaction has this desirable property.
  


  Importantly we exhibit a simple condition on game forms which
  ensures the above: determined game forms as studied
  in~\cite{FromLocalToGlobal} do satisfy the condition. The latter
  game forms generalize turn-based local interactions (where each
  players' state is controlled by a unique player -- that is, the
  matrix defining the local interaction has a single row or a single
  column). We therefore recover the fact that stochastic turn-based
  reachability games admit optimal positional strategies, which was
  shown in~\cite{MM02,CJH04,zielonka04}. 
\end{enumerate}





\subparagraph*{Related work.}  
In~\cite{AHK07}, the authors characterize using fixed points as well
states with value $1$: sure-winning states (all generated plays
satisfy the reachability condition -- as if no probabilities were
involved), almost-sure winning states (that is, maximizable states with
value $1$) and limit-sure winning states (that is, sub-maximizable states
with value $1$). Our work generalizes this result with states with
arbitrary values.

There are many works dedicated to the study of stochastic turn-based
games. These games enjoy more properties. Indeed, in parity stochastic
turn-based games, Player $\A$ always has an optimal pure positional
strategy~\cite{MM02,CJH04,zielonka04}. These results do not extend in
general to infinite (turn-based) arenas (even when they are
finitely-branching): optimal strategies may not exist, and 
when they exist, they may require infinite memory~\cite{kucera11}.

	\section{Preliminaries}
	\label{sec:preliminaries}

Consider a non-empty set $Q$. The \emph{support} $\Supp(\mu)$ of a function $\mu: Q \rightarrow [0,1]$ corresponds to set of non-0s of the function: $\Supp(\mu) = \{ q \in Q \mid \mu(q) \in \; ]0,1] \}$. A \emph{discrete probabilistic distribution} over a non-empty set $Q$ is a function $\mu: Q \rightarrow [0,1]$ such that its support $\Supp(\mu)$ is countable and $\sum_{x \in Q} \mu(x) = 1$. The set of all distributions over the set $Q$ is denoted $\Dist(Q)$. 
We also consider the product order on vectors $\mathord\preceq: \R^n \times \R^n$  defined for any $n \in \N$ by, for all $v,v' \in \R^n$, we have $v \preceq v' \Leftrightarrow \forall i \in \llbracket 1,n \rrbracket,\; v(i) \leq v'(i)$. For $v \in \R^n$ and $x \in \R$, the notation $v + x$ refers to the vector $v' \in \R^n$ such that, for all $i \in \llbracket 1,n \rrbracket$, we have $v'(i) = v(i) + x$.
	
	\section{Game Forms}
	\label{sec:gameForms}
We recall the definition of game forms which informally are 2-dim. tables with variables.
\begin{definition}[Game form and game in normal form]
	\label{def:arena_game_nf}
	A \emph{game form} is a tuple
	$\formNF = \langle \St_\A,\St_\B,\outComeNF,\outCNF \rangle$ where
	$\St_\A$ (resp. $\St_\B$) is the non-empty set of (pure) strategies available to Player $\A$ (resp. $\B$), $\outComeNF$ is a non-empty set of possible outcomes, and
	$\outCNF: \St_\A \times \St_\B \rightarrow \outComeNF$ is a function
	that associates an outcome to each pair of strategies. When the set of
	outcomes $\outComeNF$ is equal to $[0,1]$, we say that $\formNF$ is a \emph{game in normal form}. For a valuation $v \in [0,1]^\outComeNF$ of the outcomes, the notation $\gameNF{\formNF}{v}$ refers to the 
game in normal form $\langle \St_\A,\St_\B,[0,1],v \circ \outCNF \rangle$. A game form $\formNF = \langle \St_\A,\St_\B,\outComeNF,\outCNF \rangle$ is \emph{finite} if the set of pure strategies $\St_\A \cup \St_\B$ is finite. 
\end{definition}
In the following, the game form $\formNF$ will always refer to the tuple $\langle \St_\A,\St_\B,\outComeNF,\outCNF \rangle$ unless otherwise stated. Furthermore, we will be interested in valuations of the outcomes in the interval $[0,1]$. Informally, Player $\A$ (the rows) tries to maximize the outcome, whereas Player $\B$ (the columns) tries to minimize it. 
\begin{definition}[Outcome of a game in normal form]
	\label{def:outcome_game_form}
	Consider a game in normal form $\formNF = \langle \St_\A,\St_\B,[0,1],\outCNF \rangle$. The set $\Dist(\St_\A)$ corresponds to the set of mixed strategies available to Player $\A$, and analogously for Player $\B$. For a pair of mixed strategies $(\sigma_\A,\sigma_\B) \in \Dist(\St_\A) \times \Dist(\St_\B)$, the outcome $\outM_\formNF(\sigma_\A,\sigma_\B)$ in $\formNF$ of the strategies $(\sigma_\A,\sigma_\B)$ is defined as:
	$\outM_\formNF(\sigma_\A,\sigma_\B) := \sum_{a \in \St_\A} \sum_{b \in \St_\B} \sigma_\A(a) \cdot \sigma_\B(b) \cdot \outCNF(a,b) \in [0,1]$.
\end{definition}	

The definition of the value of a game in normal form follows:
\begin{definition}[Value of a game in normal form and optimal strategies]
	\label{def:alternative_value_game_normal_form}
	Consider a game in normal form $\formNF = \langle \St_\A,\St_\B,[0,1],\outCNF \rangle$ and a strategy $\sigma_\A \in \Dist(\St_\A)$ 
	for Player $\A$. 
	The \emph{value} of strategy $\sigma_\A$, 
	denoted $\va_\formNF(\sigma_\A)$ 
	is equal to: $\va_\formNF(\sigma_\A) := \inf_{\sigma_\B \in \Dist(\St_\B)} \outM_{\formNF}(\sigma_\A,\sigma_\B)$, and analogously for Player $\B$, with a $\sup$ instead of an $\inf$. 
	When $\sup_{\sigma_\A \in \Dist(\St_\A)} \va_\formNF(\sigma_\A) = \inf_{\sigma_\B \in \Dist(\St_\B)} \va_\formNF(\sigma_\B)$, it defines the \emph{value} of the game $\formNF$, denoted  $\va_\formNF$. 
	
	Note that 
	von Neuman's minimax theorem \cite{vonNeuman} ensures it does as soon as the game $\formNF$ is finite. A strategy $\sigma_\A \in \Dist(\St_\A)$ ensuring $\va_\formNF = \va_\formNF(\sigma_\A)$ is called \emph{optimal}. The set of all optimal strategies for Player $\A$ is denoted $\Opt_\A(\formNF) \subseteq \Dist(\St_\A)$, and analogously for Player $\B$. Von Neuman's minimax theorem ensures the existence of optimal strategies (for both players).  
\end{definition}

As it will be useful in Section~\ref{sec:optimal_game_forms}, we define a 
least fixed point operator 
in 
a game form given a partial valuation of the outcomes, with some complement in Appendix~\ref{proof:valuation_sequence_converging}.
\begin{definition}[Total valuation induced by a partial valuation]
	\label{def:partial_assignment}
	For a game form $\formNF
	$ and a partial valuation 
	$\alpha: \outComeNF \setminus E \rightarrow [0,1]$ for some 
	$E \subseteq \outComeNF$, we define the map $f^\formNF_\alpha: [0,1] \rightarrow [0,1]$ by, for all $y \in [0,1]$:  $f^\formNF_\alpha(y) := \va_{\gameNF{\formNF}{\alpha[y]}}$ where $\alpha[y]: \outComeNF \rightarrow [0,1]$ is such that $\alpha[y][E] = \{ y \}$ and $\restriction{\alpha[y]}{\outComeNF \setminus E} = \alpha$. The map $f_\alpha$ has a least fixed point (by monotonocity), denoted $v_\alpha \in [0,1]$. The valuation $\limval{\alpha} \in [0,1]^\outComeNF$ induced by the partial valuation $\alpha$ is then equal to $\limval{\alpha} = \alpha[v_\alpha]$.	
\end{definition}


	\section{Concurrent stochastic games}
	\label{sec:concurrentGames}

In this section, we define the formalism we use throughout this paper for concurrent graph games, strategies and values. 
\begin{definition}[Stochastic concurrent games]
	A finite \emph{stochastic concurrent arena} $\Aconc$ is a tuple 
	$\AConc$ where $\setA$ (resp. $\setB$) is the non-empty finite set of 
	actions 
        of Player $\A$ (resp. $\B$), $Q$ is the non-empty finite set of states, 
	$\distribSet$ is the non-empty set of Nature states, $\delta: Q \times \setA \times \setB \rightarrow \distribSet$ is the transition function, $\distribFunc: \distribSet \rightarrow \Dist(Q)$ is the distribution function. 
	A \emph{concurrent reachability game} is a pair $\Games{\Aconc}{\top}$ where $\top \in Q$ is a target 
state (for Player $\A$). It is supposed to be a self-looping sink: for all $a \in A$ and $b \in B$, we have $\Supp(\delta(\top,a,b)) = \{ \top \}$.
\end{definition}

In the following, the arena $\Aconc$ will always refer to the tuple $\AConc$ unless otherwise stated, and $\top$ to the  target in the game $\Games{\Aconc}{\top}$, that we assume fixed in the rest of the definitions. 
%
Let us now consider a crucial tool in our study: the notion of local interaction. These are game forms induced by the transition function $\delta$ in states of the game. 
\begin{definition}[Local interaction]
        The \emph{local interaction} at state $q \in Q$ is the game form $\formNF_q := \langle 
	\setA,\setB,\distribSet,\delta(q,\cdot,\cdot) \rangle$. That is, the strategies available for Player $\A$ (resp. $\B$) are the actions in $A$ (resp. $B$) and the outcomes are the Nature states
	. 
\end{definition}

Local interactions also allow us to define the probability transition to go from one state to another, given two local strategies.
\begin{definition}[Probability transition]
	\label{def:mu_state}
	Consider 
	a state $q \in Q$ and two local strategies $(\sigma_\A,\sigma_\B) \in \Dist(A) \times \Dist(B)$ in the game form $\formNF_q$. Let $q' \in Q$. The probability $\probTrans{q,q'}(\sigma_\A,\sigma_\B)$ to go from $q$ to $q'$ if the players opt for strategies $\sigma_\A$ and $\sigma_\B$ is equal to the outcome of the game form 
	$\formNF_q$ with 
	the value of a Nature state $d \in \distribSet$ equal to 
	the probability to go from $d$ to $q'$, i.e. it is given by the valuation $\distribFunc(\cdot)(q') \in [0,1]^\distribSet$. That is:
	$\probTrans{q,q'}(\sigma_\A,\sigma_\B) := \outM_{\gameNF{\formNF_q}{\distribFunc(\cdot)(q')}}(\sigma_\A,\sigma_\B)$.
\end{definition}

Let us now look at the strategies we consider in such concurrent games.
\begin{definition}[Strategies]
        A Player $\A$ strategy 
	is a map $\s_\A: Q^+ \rightarrow \Dist(A)$. 
	It is said to be \emph{positional} if, for all $\pi = \rho \cdot q \in Q^+$, we have $\s_\A(\pi) = \s_\A(q)$: the strategy only depends on the current state. 
	We denote by $\SetStrat{\Aconc}{\A}$ and $\SetPosStrat{\Aconc}{\A}$ 
	the set of all strategies and positional strategies respectively in arena $\Aconc$ for Player $\A$. 
	The definitions are analogous for Player $\B$.
\end{definition}

A pair of strategies then induces a probability measure over paths. 
\begin{definition}[Probability measure of paths given two strategies]
	\label{def:value_graph_game}
	For a pair of strategies $(\s_\A,\s_\B) \in \SetStrat{\Aconc}{\A} \times \SetStrat{\Aconc}{\B}$, we denote by $\s_\A^\pi: Q^+ \rightarrow \Dist(A)$ the Player $\A$ residual strategy after $\pi \in Q^+$ is seen: for all $\pi' \in Q^+$, $\s_\A^\pi(\pi') = \s_\A(\pi \cdot \pi')$. The residual strategy $\s_\B^\pi$ is defined analogously. Then, 
	the probability of occurrence of a finite path $\pi \in Q^+$ is defined inductively. For all starting states $q_0 \in Q$, for all $q \cdot \pi \in Q^+$, if $q \neq q_0$, we set $\prob{q_0}{\s_\A,\s_\B}(q) := 0$. Furthermore, $\prob{q_0}{\s_\A,\s_\B}(q_0) := 1$ and for all 
	$q \cdot \pi \in Q^+$, 
	we set: 
	\begin{displaymath}
	\prob{q_0}{\s_\A,\s_\B}(q_0 \cdot q \cdot \pi) := \probTrans{q_0,q}(\s_\A(q_0),\s_\B(q_0)) \cdot \prob{q}{\s_\A^{q_0},\s_\B^{q_0}}(q \cdot \pi)
	\end{displaymath}
	A probability measure $\prob{q_0}{\s_\A,\s_\B}$ is thus defined over the $\sigma$-algebra generated by cylinders (which are continuations of finite paths). Standardly (see e.g.~\cite{vardi85}), infinite sequences of states visiting some subset $Q' \subseteq Q$ is measurable, and we note $\prob{q_0}{\s_\A,\s_\B}(Q')$  (resp. $\prob{q_0}{\s_\A,\s_\B}(n,Q')$) the probability to reach $Q'$ (resp. in at most $n$ steps) from state $q_0$. 
\end{definition}

Finally, we can define what is the value of strategies (for both players) and of the game.
\begin{definition}[Value of strategies and of the game]
	\label{def:value_game_strat}
        The value $\val{\Aconc}{\s_\A}(q)$ 
	of a Player $\A$ strategy $\s_\A$ 
	from a state $q \in Q$ is equal to $\val{\Aconc}{\s_\A}(q) := \inf_{\s_\B \in \SetStrat{\Aconc}{\B}} \prob{q}{\s_\A,\s_\B}(\top)$. 
	The value $\val{\Aconc}{\A}(q)$ of the game for Player $\A$ from $q$ is: 
	$\val{\Aconc}{\A}(q) := \sup_{\s_\A \in \SetStrat{\Aconc}{\A}} \val{\Aconc}{\s_\A}(q)$. 
	It is analogous for Player $\B$, by inverting the $\inf$ and $\sup$. When equality of these two values holds, it defines the \emph{value} at state $q$, denoted $\val{\Aconc}{}(q)$: $\val{\Aconc}{}(q) := \val{\Aconc}{\A}(q) = \val{\Aconc}{\B}(q) \in [0,1]$. The value of the game 
is then given by the valuation $\val{\Aconc}{} \in [0,1]^Q$. Since the game is finite, 
	\cite{martin98} gives that this equality is always ensured. 
	A strategy $\s_\A \in \SetStrat{\Aconc}{\A}$ such that $\val{\Aconc}{\s_\A}(q) = \val{\Aconc}{\A}(q)$ (resp. $\val{\Aconc}{\s_\A}(q) \geq \val{\Aconc}{\A}(q) - \varepsilon$ for some $\varepsilon > 0$) is called a Player $\A$ \emph{optimal} strategy (resp. $\varepsilon$-\emph{optimal}) from state $q$. If $\val{\Aconc}{\s_\A} = \val{\Aconc}{\A}$, the strategy $\s_\A$ is \emph{uniformly} optimal. This is defined analogously for Player $\B$. For a valuation $v \in [0,1]^Q$ of the states, 
	a Player $\A$ strategy $\s_\A \in \SetStrat{\Aconc}{\A}$ 
	such that $v \preceq \val{\Aconc}{\s_\A}$ is said to \emph{guarantee} the valuation $v$. 
\end{definition}

\noindent \textbf{Value of the game and least fixed point.}
\label{subsec:value_lfp}
In the context of a reachability game, the value of the game is the least fixed point (lfp) of an operator on valuations on states. We define this operator here 
with some complements given in Appendix~\ref{appen:value_lfp}.
\begin{definition}[Valuation of the Nature states and operator on values]
	\label{def:operator}
	For $v \in [0,1]^Q$, we define the valuation $\mu_v \in [0,1]^\distribSet$ of the Nature states by $\mu_v(d) := \sum_{q \in Q} \distribFunc(d)(q) \cdot v(q)$ for all $d \in \distribSet$. For the operator $\Delta: [0,1]^Q \rightarrow [0,1]^Q$, for all valuations $v \in [0,1]^Q$, we set $\Delta(v)(\top) := 1$ and, for all $q \neq \top \in Q$, we set $\Delta(v)(q) := \va_{\gameNF{\formNF_q}{\mu_v}}$. 
\end{definition}
As the operator $\Delta$ is monotonous, it has an lfp 
for the product order $\preceq$. This lfp gives the value of the game. Furthermore, Player $\B$ has an optimal positional strategy:
\begin{theorem}[\cite{everett57,filar2012competitive}]
	\label{thm:m_delta}
	Let $\lfp$ denote the lfp of the operator $\Delta$. Then: 
	$\val{\Aconc}{} = \lfp$. Furthermore, there exists a positional strategy $\s_\B \in \SetPosStrat{\B}{\Aconc}$ for Player $\B$ ensuring $\val{\Aconc}{\s_\B} = \val{\Aconc}{} = \lfp$.
\end{theorem}

\noindent \textbf{Markov decision process induced by a positional strategy.}
Once a Player $\A$ positional strategy is fixed, we obtain a Markov decision process, which, informally, is a game where only one player (here, Player $\B$) plays (against probabilistic transitions).
\begin{definition}[Induced Markov decision process]
	Consider 
        a Player $\A$ positional strategy $\s_\A \in \SetPosStrat{\Aconc}{\A}$. 
	The \emph{Markov decision process} $\Gamma$ (MDP for short) induced by the strategy  $\s_\A$ is the triplet $\Gamma := \langle Q,B,\iota \rangle$ where $Q$ is the set of states, $B$ is the set of actions 
	and $\iota: Q \times B \rightarrow \Dist(Q)$ is a map associating to a state and an action a distribution over the states. For all $q \in Q$, $b \in B$ and $q' \in Q$, we set $\iota(q,b)(q') := \probTrans{q,q'}(\s_\A(q),b)$. 
\end{definition}
Note that the set of Player $\B$ strategies in an induced MDP $\Gamma$ is the same as in the concurrent game $\Aconc$. Furthermore, the useful objects in MDPs are the end components~\cite{dealfaro97}: informally, sub-MDPs that are strongly connected.
\begin{definition}[End component]
	\label{def:end_component}
	Consider 
        a Player $\A$ positional strategy $\s_\A \in \SetPosStrat{\Aconc}{\A}$  
        and consider the MDP $\Gamma$ induced by that strategy. An \emph{end component} (EC for short) $H$ in $\Gamma$ is a pair $(Q_H,\beta)$ such that $Q_H \subseteq Q$ is a subset of states and $\beta: Q_H \rightarrow \mathcal{P}(B) \setminus \emptyset$ associates to each state a non-empty set of actions compatible with the EC $H$ such that:
	\begin{itemize}
		\item for all $q \in Q_H$ and $b \in \beta(q)$, we have $\Supp(\iota(q,b)) \subseteq Q_H$;
		\item the underlying graph $(Q_H,E)$ is strongly connected where $(q,q') \in E$ iff $q' \in \Supp(\iota(q,\beta(q)))$.
	\end{itemize}
	We denote by $\distribSet_H \subseteq \distribSet$ the set of Nature states compatible with the EC $H$: $\distribSet_H = \{ d \in \distribSet \mid \Supp(d) \subseteq Q_H \}$. Note that, for all $q \in Q_H$ and $b \in \beta(q)$, we have $\delta(q,\Supp(\s_\A(q)),b) \subseteq \distribSet_H$.
\end{definition} 

The interest of ECs lies in the proposition below: in the MDP induced by a Player $\A$ strategy, for all Player $\B$ (positional) strategies (thus inducing a Markov chain), from all states, there is a non-zero probability to reach an EC from which it is impossible to exit.
\begin{proposition}[Complement~\ref{proof:prop_end_in_EC}]
	\label{prop:end_in_EC} Consider 
        a Player $\A$ positional strategy $\s_\A\in \SetPosStrat{\Aconc}{\A}$. Let $\mathcal{H}$ denote the set of all ECs in the MDP 
	induced by the strategy 
        $\s_\A$. For all Player $\B$ strategies $\s_\B \in \SetPosStrat{\Aconc}{\B}$, there exists a subset of end components $\mathcal{H}_{\s_\B} \subseteq \mathcal{H}$ called \emph{bottom strongly conneted components} (BSCC for short): 
	for all $H = (Q_H,\beta) \in \mathcal{H}_{\s_\B}$ and $q \in Q_{H}$, we have $\prob{q}{\s_\A,\s_\B}(Q \setminus Q_H) = 0$. Furthermore, if $q \in Q$, we have: $\prob{q}{\s_\A,\s_\B}(n,\cup_{H \in \mathcal{H}_{\s_\B}} H) > 0$ where $n = |Q|$.
\end{proposition}


	\section{Crucial proposition}
	\label{sec:crucial_prop}

We fix a concurrent reachability game $\langle \Aconc,T \rangle$ and a valuation $v \in [0,1]^Q$ of the states that Player $\A$ wants to guarantee. 
That is, she seeks a strategy $\s_\A$ ensuring 
that for all $q \in Q$, it holds $\val{\Aconc}{\s_\A}(q) \geq v(q)$. In particular, when $v = \lfp$, such a strategy $\s_\A$ would be optimal. 
We state a 
sufficient condition for Player $\A$ positional strategies to ensure such a property. 

Consider 
a Player $\A$ positional strategy $\s_\A \in \SetPosStrat{\A}{\Aconc}$. The probability distribution chosen by this strategy only depends on the current state. In fact, 
this strategy is built with one (local) strategy per local interaction: 
for all state $q \in Q$, $\s_\A(q) \in \Dist(A)$ is a strategy in the game form $\formNF_q$. As Player $\A$ wants to guarantee the valuation $v$, the valuation of interest of the outcomes of the game form $\formNF_q = \langle A,B,\distribSet,\delta(q,\cdot,\cdot) \rangle$ is $\mu_v \in [0,1]^\distribSet$ -- lifting the valuation $v$ to the Nature states. To ensure that $\val{\Aconc}{\s_\A}(q) \geq v(q)$, one may think that it suffices to choose $\s_\A(q)$ so that its value in the game in normal form $\gameNF{\formNF_q}{\mu_v}$ is at least $v(q)$, that is: $\va_{\gameNF{\formNF_q}{\mu_v}}(\s_\A(q)) \geq v(q)$. In that case, the strategy $\s_\A$ is said to locally dominate the valuation $v$:
\begin{definition}[Strategy locally dominating a valuation]
  A Player $\A$ positional strategy $\s_\A \in \SetPosStrat{\Aconc}{\A}$ \emph{locally dominates} the valuation $v$ if, for all $q \in Q$, we have: $\va_{\gameNF{\formNF_q}{\mu_v}}(\s_\A(q)) \geq v(q)$.
\end{definition}

However, this is not sufficient in the general case, as examplified in Figure~\ref{fig:arbitrarilyClose}
. For the valuation $v = \val{\Aconc}{}$ 
such that $v(q_0) = v(\top) = 1$ and $v(\bot) = 0$, a Player $\A$ positional strategy $\s_\A$ that 
plays the first row in $\formNF_{q_0}$ with probability 1 ensures that $\va_{\gameNF{\formNF_{q_0}}{\mu_v}}(\s_\A(q_0)) = 1 \geq v(q_0)$. However, we have seen that it does not ensure that $\val{\Aconc}{\s_\A}(q_0) = 1$ since, if Player $\B$ always plays the first column, the game indefinitely loops in $q_0$. The issue 
is that, in the MDP induced by the strategy $\s_\A$, the trivial end component $\{q_0\}$ is a trap, as it does not intersect the target set $\top$ -- and therefore, the probability to reach $\top$ from $q_0$ is equal to $0$ -- whereas $\val{\Aconc}{}(q_0) > 0$
. 
In fact, as soon as this issue is avoided, if the strategy $\s_\A$ locally dominates the valuation $v$, the desired property on $\s_\A$ holds. Indeed:
\begin{proposition}[Proof~\ref{proof:lem_sufficient_cond}]
	\label{prop:sufficient}
	Consider 
        a Player $\A$ positional strategy $\s_\A \in \SetPosStrat{\Aconc}{\A}$ locally dominating $v$, and assume 
that $v \preceq \lfp$. Assume that for all end components $H = (Q_H,\beta)$ in the MDP induced by the strategy $\s_\A$, if $Q_H \neq \{ \top \}$, for all $q_H \in Q_H$, we have $\val{\Aconc}{}(q_H) = 0$ (in other words, for all $q \in Q$, if $\val{\Aconc}{\s_\A}(q) = 0$ then $\val{\Aconc}{}(q) = 0$). In that case, for all $q \in Q$, we have $\val{\Aconc}{\s_\A}(q) \geq v(q)$ (i.e. the strategy $\s_\A$ guarantees the valuation $v$).
\end{proposition}
\begin{proofs}
	Consider some $\varepsilon > 0$ and, for $x \in \{ \varepsilon,\varepsilon/2 \}$, the valuations $v_{x} = v - x \in [0,1]^Q$. We show that 
        $\s_\A$ guarantees $v_\varepsilon$. As this holds for all $\varepsilon > 0$, it follows that $\s_\A$ guarantees $v$. 
	Consider an arbitrary positional strategy $\s_\B$ for Player $\B$. 
	Let $\kappa_\A$ be a Player $\A$ strategy guaranteeing $v_{\varepsilon/2}$ in $n \geq 0$ steps from every state (which exists since $v_{\varepsilon/2} \prec \lfp$) and a strategy $\kappa_\B$ for Player $\B$ optimal against $\kappa_\A$. So $\prob{q}{\kappa_\A,\kappa_\B}(n,\top) \geq v_{\varepsilon/2}(q)$ for all $q \in Q$. Now, for all $l \geq 0$, we consider the strategy $\s_\A^l$ that plays $\s_\A$ $l$ times and then plays $\kappa_\A$ (and similarly for a strategy $\s_\B^l$ for Player $\B$). As $\s_\A$ locally dominates $v$, it also locally dominates $v_{\varepsilon/2}$ which is obtained from $v$ by translation. Therefore, for any state $q \in Q$, if the local strategy $\s_\A(q)$ is played in $q$, then the convex combination of the values of the successors of $q$ w.r.t. the valuation $v_{\varepsilon/2}$ is at least $v_{\varepsilon/2}(q)$. In other words, the probability to reach $\top$ from $q$ in $1+n$ steps if the strategy $\s_\A^1$ is played is at least $v_{\varepsilon/2}(q)$: $\prob{q}{\s_\A^1,\s_\B^1}(1+n,\top) \geq v_{\varepsilon/2}(q)$. In fact, by induction, this holds for all $l \geq 0$: $\prob{q}{\s_\A^l,\s_\B^l}(l+n,\top) \geq v_{\varepsilon/2}(q)$. 
	Now, with strategies $\s_\A^l$ and $\s_\B^l$, consider the state of the game after $l$ steps: either it is in a BSCC (w.r.t. $\s_\A$ and $\s_\B$) or it is not. For a sufficiently large $l$, the probability not to have reached a BSCC is as close to 0 as we want. Furthermore, for a state $q_H$ in a BSCC $H$ that is not $\{ \top \}$, by assumption, we have that $\val{\Aconc}{}(q_H) = 0$, hence $\prob{q_H}{\kappa_\A,\kappa_\B}(\top) = 0$. In addition, if the state is in the trivial BSCC $\{ \top \}$, then $\top$ is reached. 
	Hence, for $l$ large enough, the two probabilities $\prob{q}{\s_\A^l,\s_\B^l}(l+n,\top)$ and $\prob{q}{\s_\A^l,\s_\B^l}(l,\top)$ are as close to one another as we want. 
	Finally, note that the strategies $\s_\A^l,\s_\B^l$ behave exactly like the strategies $\s_\A,\s_\B$ in the first $l$ steps. 
	That is, for $l$ large enough, and $q \in Q$, we have $\prob{q}{\s_\A,\s_\B}(\top) \geq \prob{q}{\s_\A,\s_\B}(l,\top) = \prob{q}{\s_\A^l,\s_\B^l}(l,\top) \geq \prob{q}{\s_\A^l,\s_\B^l}(l+n,\top) - \varepsilon/2 \geq v_{\varepsilon/2}(q) - \varepsilon/2 = v_\varepsilon(q)$.
\end{proofs}

Fix a Player $\A$ positional strategy $\s_\A$ locally dominating the valuation $v$ and let $\Gamma$ be the MDP induced by $\s_\A$.
For 
$\s_\A$ to guarantee the valuation $v$, it suffices to ensure that any EC 
in $\Gamma$ that is not the trivial EC $\{ \top \}$  has all its states of value 0. It does not necessarily hold for 
$\s_\A$ (recall the explanations before Proposition~\ref{prop:sufficient}). However, we do have the following: fix an EC $H$ in 
$\Gamma$. Then, all the states 
$H$ have the same value w.r.t. the valuation $v$. It is stated in the proposition below.

\begin{proposition}[Proof~\ref{proof:value_end_component}]
	\label{prop:value_end_component}
	Consider 
        a Player $\A$ positional strategy $\s_\A \in \SetPosStrat{\Aconc}{\A}$ locally dominating a valuation $v \in [0,1]^Q$. 
	For all EC $H = (Q_H,\beta)$ in the MDP induced by the strategy $\s_\A$, there exists 
	$v_H \in [0,1]$ such that, for all $q \in Q_H$, we have $v(q) = v_H$. Furthermore, for all $q \in Q_H$, we have $\va_{\gameNF{\formNF_q}{\mu_v}}(\s_\A(q)) = v(q)$.
\end{proposition}

	\section{Positional optimal and $\varepsilon$-optimal strategies}
	\label{sec:optimal_strat}

The aim of this section is, given a concurrent reachability game, to
determine exactly from which states Player $\A$ has an optimal
strategy. This, in turn, will give that whenever she has an optimal
strategy, she has one that is positional which therefore extends
Everett \cite{everett57} (the existence of positional
$\varepsilon$-optimal strategies). We fix a concurrent reachability
game $\langle \Aconc,\top \rangle$  for the rest of this section. Let us first introduce some terminology.
\begin{definition}[Maximizable and sub-maximizable states]
  A state $q \in Q$ from which Player $\A$ has (resp. does not have) an optimal strategy is called \emph{maximizable} (resp. \emph{sub-maximizable}). The set of such states is denoted $\OS_\A$ (resp. $\SubOS_\A$).
\end{definition}

The value of that game is given by the vector $\lfp \in [0,1]^Q$ (from Definition~\ref{def:operator}). We want to build an optimal (and positional) strategy for Player $\A$ when possible. To be optimal, a Player $\A$ positional strategy $\s_\A$ has to play optimally at each local interaction $\formNF_q$ (for $q \in Q$) with regard to the valuation $\mu_{\lfp} \in [0,1]^\distribSet$ (lifting the valuation $\lfp$ to Nature states). However, it is not sufficient in general: in the snow-ball game of Figure~\ref{fig:arbitrarilyClose}, when Player $\A$ plays optimally in $\formNF_{q_0}$ w.r.t. the valuation $\mu_{\lfp}$ (that is, plays the first line with probability 1), Player $\B$ can enforce the play never to leave the state $q_0 \neq \top$. Hence, locally, we want to have 
strategies that not only play optimally but, regardless of the choice of Player $\B$, have a non-zero probability to get closer to the target $\top$. Such strategies will be called \emph{progressive strategies}. To properly define them, we introduce the following notation.

\begin{definition}[Optimal action]
	\label{def:optimal_act}
	Let $q \in Q$ be a state of the game. Consider the game in normal form $\gameNF{\formNF_q}{\mu_{\lfp}}$. For all 
	strategies $\sigma_\A \in \Dist(\St_\A)
	$, we define the set $B_{\sigma_\A}$ of \emph{optimal actions} 
	w.r.t. the strategy $\sigma_\A$ by $B_{\sigma_\A} := \{ b \in B \mid \outM_{\gameNF{\formNF_q}{\mu_{\lfp}}}(\sigma_\A,b) = \va_{\gameNF{\formNF_q}{\mu_{\lfp}}}(\sigma_\A) \}$. 
\end{definition}

In Figure~\ref{fig:example_prog_risk}, the set $B_{\sigma_\A}$ of optimal actions w.r.t. the strategy $\sigma_\A$ are represented in bold purple: the weighted values of these actions is the value of the strategy: $1/2$. 

We can now define the set of \emph{progressive} strategies.
\begin{definition}[Progressive strategies]
	\label{def:progressive_strat} 
	Consider a state $q \in Q$ and a set of states $\Gex \subseteq Q$ that Player $\A$ wants to reach
	. The set of Nature states $\Gex_\distribSet \subseteq \distribSet$ corresponds to the Nature states with a non-zero probability to reach the set $\Gex$: $\Gex_\distribSet := \{ d \in \distribSet \mid \Supp(\distribFunc(d)) \cap \Gex \neq \emptyset \}$. Then, the set of \emph{progressive strategies} $\Prg_q(\Gex)$ at state $q$ w.r.t. $\Gex$ is defined by $\Prg_q(\Gex) := \{ \sigma_\A \in \Opt_\A(\gameNF{\formNF_q}{\mu_{\lfp}}) \mid \forall b \in B_{\sigma_\A},\; \delta(q,\Supp(\sigma_\A),b) \cap \Gex_\distribSet  \neq \emptyset \}$. 
\end{definition}
In Figure~\ref{fig:example_prog_risk}, the Nature states in $\Gex_\distribSet$ are arbitrarily chosen for the example and circled in green. The depicted strategy is progressive as, for all bold purple actions, there is a green-circled state in the support of the strategy (the circled $3/4$).

However, in an arbitrary game, some states may be sub-maximizable
. In that case, playing optimally implies avoiding these states. Given a set $\Bex \subseteq Q$ of states to avoid, an optimal strategy that has a non-zero probability to reach that set of states $\Bex$ is called \emph{risky}.
\begin{definition}[Risky strategies]
	\label{def:risky_strat}
	Let $q \in Q$ be a state of the game and $\Bex \subseteq Q$ be a set of sub-maximizable states. The corresponding set of Nature states $\Bex_\distribSet \subseteq \distribSet$ is defined similarly to $\Gex_\distribSet$ in Definition~\ref{def:progressive_strat}: $\Bex_\distribSet := \{ d \in \distribSet \mid \Supp(\distribFunc(d)) \cap \Bex \neq \emptyset \}$. Then, the set of \emph{risky strategies} $\Rsk_q(\Bex)$ at state $q$ w.r.t. $\Bex$ is defined by $\Rsk_q(\Bex) := \{ \sigma_\A \in \Opt_\A(\gameNF{\formNF_q}{\mu_{\lfp}}) \mid \exists b \in B_{\sigma_\A},\; \delta(q,\Supp(\sigma_\A),b) \cap \Bex_\distribSet \neq \emptyset \}$. 
\end{definition}
In Figure~\ref{fig:example_prog_risk}, the set of Nature states $\Bex_\distribSet$ are also arbitrarily chosen for the example and circled in red. The strategy $\sigma_\A$ is not risky since no red-squared state appears in the intersection of the support of $\sigma_\A$ and the purple actions in $B_{\sigma_A}$.

In fact, we want for local strategies to be \emph{efficient}, that is both progressive and not risky. 
\begin{definition}[Efficient strategies]
	\label{def:efficient_strat}
	Let $q \in Q$ be a state of the game and $\Gex,\Bex \subseteq Q$ be sets of states. The set of \emph{efficient strategies} $\Eff_q(\Gex,\Bex)$ at state $q$ w.r.t. $\Gex$ and $\Bex$ is defined by $\Eff_q(\Gex,\Bex) := \Prg_q(\Gex) \setminus \Rsk_q(\Bex)$. 
\end{definition}
In Figure~\ref{fig:example_prog_risk}, the strategy $\sigma_\A$ is efficient as it is both progressive and not risky.

We can now compute inductively the set of maximizable and sub-maximizable states. 
First, given a set of sub-maximizable states $\Bex$, we define iteratively below 
a set of \emph{secure} states w.r.t. $\Bex$, there are the states with a non-zero probability to get closer to the target $\top$ while avoiding the set $\Bex$. The construction is illustrated in Figure~\ref{fig:drawing}.

\begin{definition}[Secure states]
	\label{def:secure_states}
	Consider 
	a set of states $\Bex \subseteq Q$
	. 
	We set $\Scr_0(\Bex) := \{ \top \}$ and, for all $i \geq 0$, $\Scr_{i+1}(\Bex) := \Scr_{i}(\Bex) \cup \{ q \in Q \setminus \Bex \mid \Eff_q(\Scr_{i}(\Bex),\Bex) \neq \emptyset \}$. The set $\Scr(\Bex)$ of states \emph{secure} w.r.t. $\Bex$ is: $\Scr(\Bex) := \cup_{n \in \N} \Scr_n(\Bex) \cup \lfp^{-1}[0]$.
\end{definition}   

\begin{figure}
	\begin{minipage}[b]{0.5\linewidth}
		\hspace*{-1cm}
		\centering
		\includegraphics[scale=1]{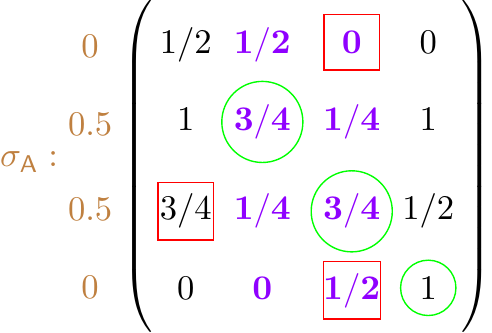}
		\caption{A game in normal form with an optimal strategy depicted in brown on the left. Its value is $1/2 = 1/2 \cdot 3/4 + 1/2 \cdot 1/4$.}
		\label{fig:example_prog_risk}         
	\end{minipage}
	\hspace{5pt}
	\begin{minipage}[b]{0.5\linewidth}
		\centering
		\includegraphics[scale=0.3]{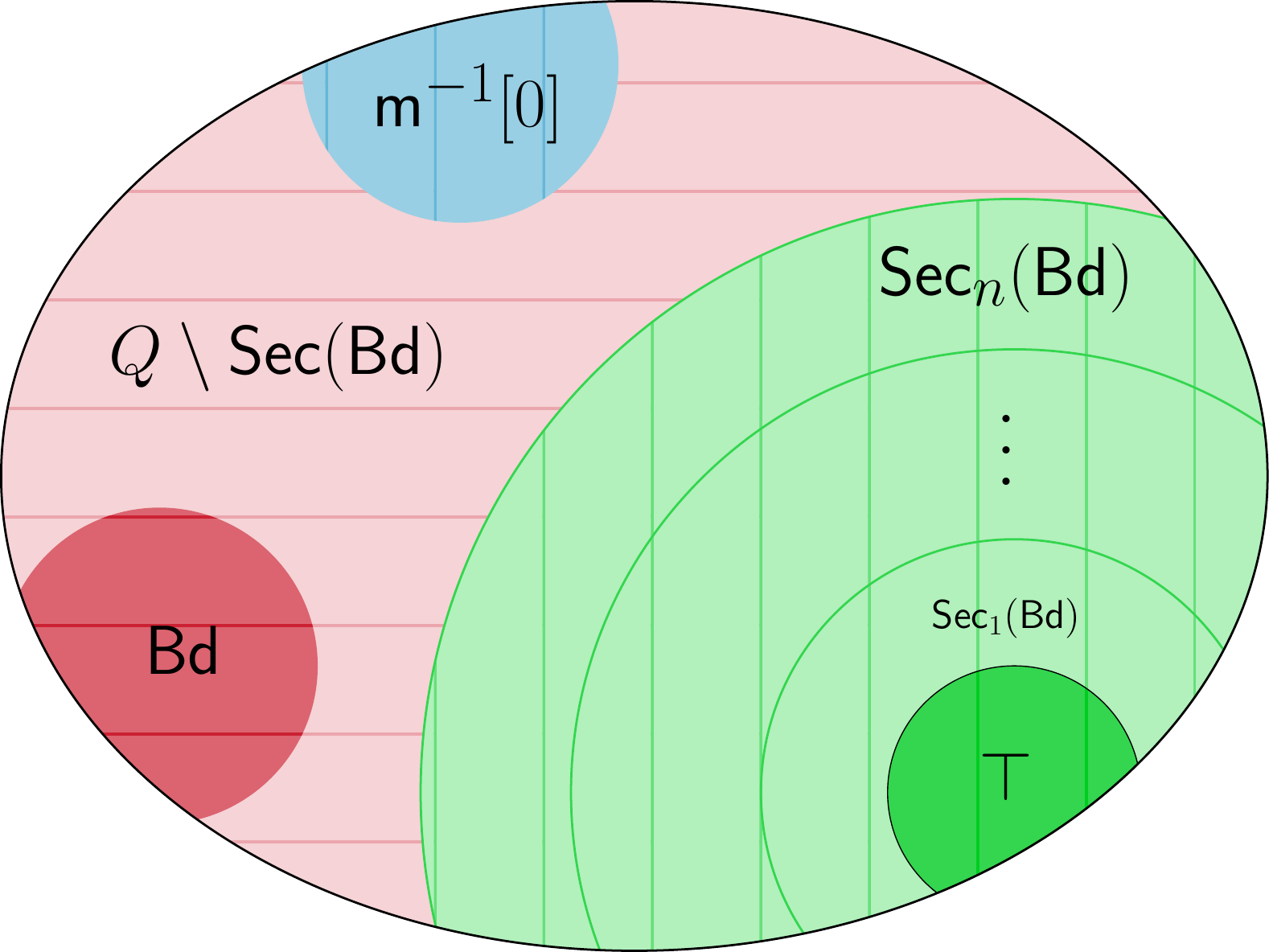}
		\caption{The construction of Definition~\ref{def:secure_states} of the set of states $\Scr(\Bex)$: it is the reunion of the blue and green vertical stripe areas.}
		\label{fig:drawing}              
	\end{minipage}
\end{figure}


Note that, as the game $\Aconc$ is finite, this procedure ends in at most $n = |Q|$ steps. Furthermore, the states of value 0 are added since any state of value 0 is maximizable
. The interest of this construction lies in the lemma below: if all states in $\Bex$ are sub-maximizable, then all states in $Q \setminus \Scr(\Bex)$ also are.

\begin{lemma}[Proof~\ref{proof:bad_increase}]
	\label{lem:bad_increase}
        Assume that a set of states $\Bex$ is such that $\Bex \subseteq \SubOS_\A$. Then, the set of states $Q \setminus \Scr(\Bex)$ is such that $Q \setminus \Scr(\Bex) \subseteq \SubOS_\A$ (these correspond to the red horizontal stripe areas in Figure~\ref{fig:drawing}).
\end{lemma}
\begin{proofs}
	For an arbitrary Player $\A$ strategy $\s_\A \in \SetStrat{\Aconc}{\A}$ to be optimal, it roughly needs, on all relevant paths, to be optimal. More precisely, on any finite path $\pi = \pi' \cdot q \in Q^+$ with a non-zero probability to occur if Player $\B$ plays (locally) optimal actions against the strategy $\s_\A$ (called a relevant path), the strategy $\s_\A$ needs to play an optimal (local) strategy in the local interaction $\formNF_q$ and it\footnote{In fact, the residual strategy $\s_\A^{\pi'}$.} has to be optimal from $q$ in the reachability game. Therefore, on all relevant paths, the strategy $\s_\A$, locally, has to play optimal strategies that are not risky. However, in any local interaction of a state $q \in Q \setminus \Scr(\Bex)$, there is no efficient strategies available to Player $\A$. Therefore, if the game starts from a state $q \in Q \setminus \Scr(\Bex)$ an optimal strategy $\s_\A$ for Player $\A$ (which therefore is locally optimal but not progressive) would allow Player $\B$ to ensure staying in the set $Q \setminus \Scr(\Bex)$ while playing optimal actions. In that case, 
	the game never leaves the set $Q \setminus \Scr(\Bex)$, which induces a value of 0, whereas $\val{\Aconc}{}(q) > 0$ since $q \notin \Scr(\Bex)$. Thus, there is no optimal strategy for Player $\A$ from a state in $Q \setminus \Scr(\Bex)$.
\end{proofs}

We define inductively the set of bad states (which, in turn, will correspond to the set of sub-maximizable states) below.
\begin{definition}[Set of sub-maximizable states]
	\label{def:bad_states}
	Let $\Bad_0 := \emptyset$ and, for all $i \geq 0$, $\Bad_{i+1} := Q \setminus \Scr(\Bad_i)$. Then, the set $\Bad$ of bad states is equal to $\Bad := \cup_{n \in \N} \Bad_n$ for $n = |Q|$.
\end{definition}
Note that, as in the case of the set of secure states, since the game $\Aconc$ is finite, this procedure ends in at most $n = |Q|$ steps.
Lemma~\ref{lem:bad_increase} ensures that the set of states $\Bad$ is included in $\SubOS_\A$
. In addition, we have that there exists a Player $\A$ positional strategy optimal from all states $q$ in its complement $\Scr(\Bad) = Q \setminus \Bad$, as stated in the lemma below.
\begin{lemma}[Proof~\ref{proof:pos_optimal_strat}]
	\label{lem:pos_optimal_strat}
	For all $\varepsilon > 0$, there exists a positional strategy $\s_\A \in \SetPosStrat{\A}{\Aconc}$ s.t.:
	\begin{itemize}
		\item for all $q \in \Scr(\Bad)$, we have $\val{\Aconc}{\s_\A}(q) = \lfp(q)$;
		\item for all $q \in \Bad$, we have $\val{\Aconc}{\s_\A}(q) \geq  \lfp(q) - \varepsilon$.
	\end{itemize}
	In particular, it follows that $\Scr(\Bad) \subseteq \OS_\A$.
\end{lemma}




\begin{figure}
	\centering
	\includegraphics[scale=0.3]{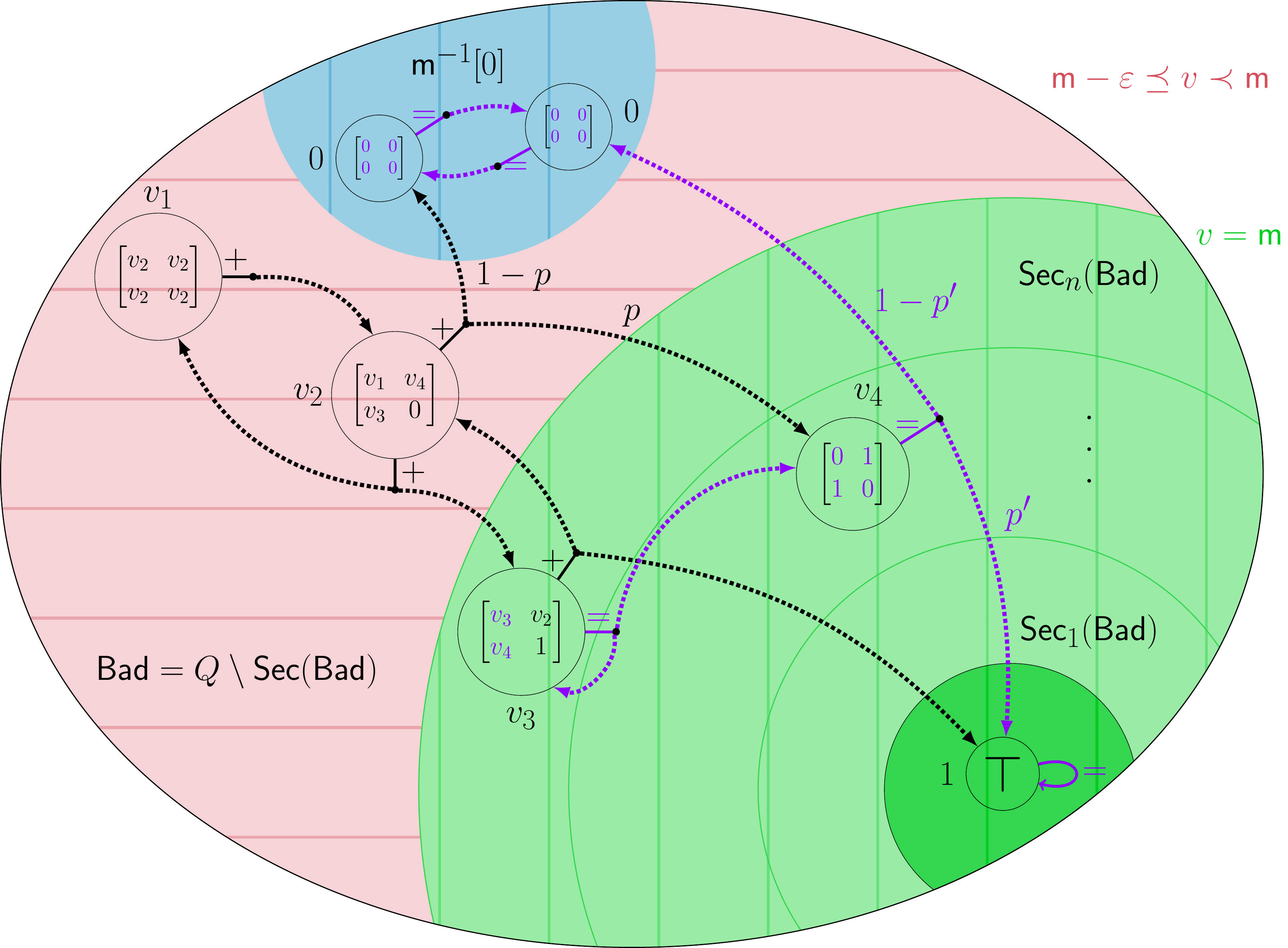}
	\caption{An illustration of the proof of Lemma~\ref{lem:pos_optimal_strat} on the MDP induced by the strategy $\s_\A$. Labels $v_1,\ldots,v_4$ is the value of the corresponding states given by the valuation $v$.}
	\label{fig:drawingArrow}   
\end{figure}
\begin{proofs}
	To prove this lemma, we define a Player $\A$ positional strategy $\s_\A \in \SetPosStrat{\Aconc}{\A}$, a valuation $v \in [0,1]^Q$ of the states, prove that the strategy $\s_\A$ locally dominates that valuation and prove that the only EC compatible with $\s_\A$ that is not the target has value 0. This will show that is guarantees the valuation $v$ by applying Proposition~\ref{prop:sufficient}. As we want the strategy $\s_\A$ to be optimal from all secure states, we consider a partial valuation $v$ such that $\restriction{v}{\Scr(\Bad)} := \restriction{\lfp}{\Scr(\Bad)}$ (we will define it later on $\Bad$). Then, on all secure states $q \in \Scr_i(\Bad)$, we set $\s_\A(q)$ to be an efficient strategy w.r.t. $\Scr_{i-1}(\Bad)$ and $\Bad$, i.e. $\s_\A(q) \in \Eff_q(\Scr_{i-1}(\Bad),\Bad)$. In particular, $\s_\A(q)$ is optimal in the game form $\formNF_q$ w.r.t. the valuation $\mu_\lfp$. However, we know that no strategy can be optimal from states in $\Bad$. Hence, we consider a valuation $v$ that is $\varepsilon$-close to the valuation $\lfp$ on states in $\Bad$ for a well-chosen $\varepsilon > 0$. This $\varepsilon$ 
	is chosen so that the value of the local strategy $\s_\A(q)$ for $q \in \Scr(\Bad)$ is 
	at least $v(q)$ w.r.t. the valuation $\mu_v$\footnote{Specifically, $\varepsilon$ has to be chosen smaller than the smallest difference between the values of an optimal actions $b \in B_{\s_\A(q)}$ and a non-optimal action $b \in B_{\s_\A(q)}$.}. 
	We can now define the valuation $v$ and the strategy $\s_\A$ on $\Bad$ such that the value of $\s_\A(q)$ in $\formNF_q$ w.r.t. $\mu_v$ is greater than $v(q)$: $\va_{\gameNF{\formNF_q}{\mu_{v}}}(\s_\A(q)) > v(q)$ (this requires a careful use the fact that the operator $\Delta$ from Section~\ref{sec:concurrentGames} is $1$-Lipschitz). 
	The valuation $v$ and the strategy $\s_\A$ are now completely defined on $Q$.	
	By definition, the strategy $\s_\A$ locally dominates the valuation $v$. 
	
	The MDP induced by the strategy $\s_\A$ is schematically depicted in Figure~\ref{fig:drawingArrow}.  The different 
	split arrows appearing in the figure correspond to the actions (or columns in the local interactions) available to Player $\B$. Black $+$-labeled-split arrows correspond to the actions of Player $\B$ that increase the value of $v$ (i.e. in a state $q$, such that the convex combination -- w.r.t. to the probabilities chosen by the strategy $\s_\A$ -- of the values w.r.t. $v$ of the successor states of $q$ is greater than $v(q)$). For instance, we have $v_2 < p \cdot v_4 + (1-p) \cdot 0$, where the probability $p \in [0,1]$ is set by the strategy $\s_\A$. On the other hand, purple $=$-labeled-split arrows correspond to the actions whose values do not increase the value of the state. For instance $v_4 = (1-p') \cdot 0 + p' \cdot 1$. We can see that the only split arrows exiting states in $\Bad$ (the red horizontal stripe area) are black (since $\va_{\gameNF{\formNF_q}{\mu_{v}}}(\s_\A(q)) > v(q)$ for all $q \in \Bad$)
	. However, from a secure state $q \in \Scr(\Bad)$ (the green and blue vertical stripe areas) there are also purple split arrows
	. Note that, in these secure states $q \in \Scr(\Bad)$, purple split arrows correspond to the optimal actions $B_{\s_\A(q)}$ at the local interaction $\formNF_q$. Furthermore, these split arrows cannot exit the set of secure states $\Scr(\Bad)$ since the local strategy $\s_\A(q)$ is not risky.
	
	
	We can then prove that the strategy $\s_\A$ guarantees the valuation $v$ by applying Proposition~\ref{prop:sufficient}: since $\s_\A$ locally dominates the valuation $v$, it remains to show that all the ECs different from $\{ \top \}$ have only states of value 0. In the figure, this corresponds to having ECs only in the blue upper circle and dark green bottom right inner circle areas. In fact, Proposition~\ref{prop:value_end_component} gives that any state $q$ in an EC ensures $\va_{\gameNF{\formNF_q}{\mu_{v}}}(\s_\A(q)) = v(q)$, which implies that no state in $\Bad$ can be in an EC. This can be seen in the figure between the states of value $v_1$ and $v_2$: because of the black arrow from $v_1$ to $v_2$, we necessarily have $v_1 < v_2$. Then, $v_2$ cannot loop (with probability one) to $v_1$ since this would imply $v_2 < v_1$. As all the split arrows are black for states in $\Bad$, no EC can appear in this region. Furthermore, the optimal actions in the secure states always have a non-zero probability to get closer to the target $\top$. In the figure, this corresponds to the fact that there is always one tip of a purple split arrow that goes down in the $(\Scr_i(\Bad))_{i \in \N}$ hierarchy (since the strategy $\s_\A(q)$ is progressive): in the example, from $v_3$ to $v_4$ and from $v_4$ to the target $\top$. Therefore, the only loop (with probability one) that can occur in the set $(\Scr_i(\Bad))_{i \in \N}$ is at the target $\top$. We conclude by applying Proposition~\ref{prop:sufficient}.

\end{proofs}

Overall, we obtain the theorem below summarizing the results proved in this section.
\begin{theorem}[Proof~\ref{proof:positional_optimal}]
	\label{thm:positional_optimal}
	In a concurrent reachability game $\langle \Aconc,\top \rangle$, we have $\Bad = \SubOS_\A$ and $\Scr(\Bad) = \OS_\A$. Furthermore, for all $\varepsilon > 0$, there is a Player $\A$ positional strategy $\s_\A$ optimal from all states in $\OS_\A$ and $\varepsilon$-optimal from all states in $\SubOS_\A$.
\end{theorem}

\begin{figure}
	\centering
	\includegraphics[scale=0.7]{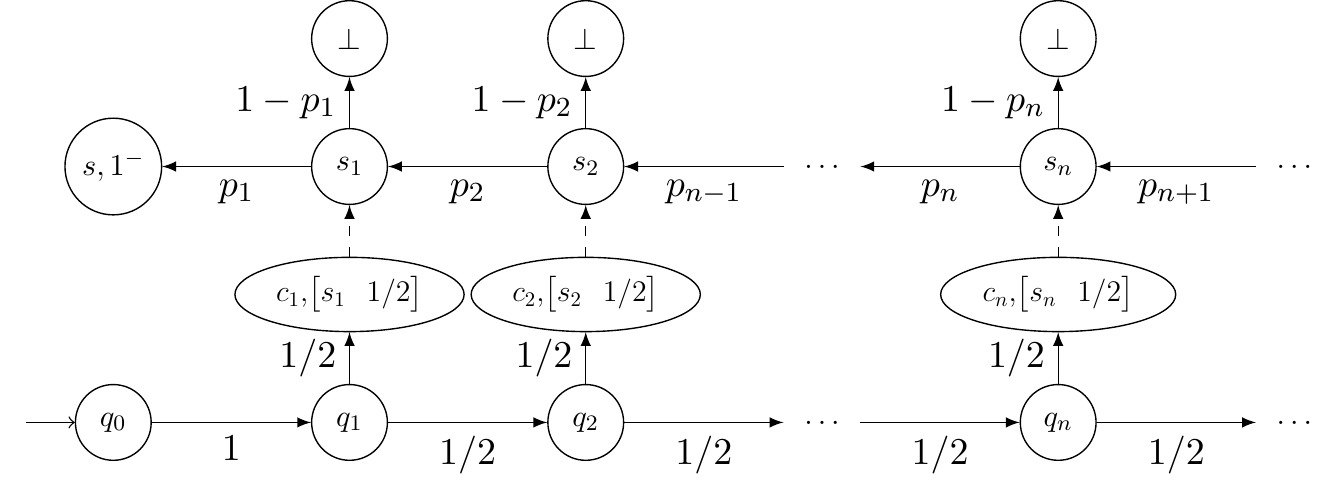}
	\caption{An infinite concurrent reachability game $\Aconc$ (the Nature states are omitted). The probabilities $p_k$ are such that, for all $i \geq 1$, the value of the state $s_i$ is $\val{\Aconc}{}(s_i) = \Pi_{k = 1}^i p_k = (1/2 + 1/2^i)$.
	}
	\label{fig:counter_example} 
\end{figure}

\smallskip\noindent \textbf{Infinite arenas.} In this paper, we only consider finite arenas and the constructions we have exhibited and results we have shown hold in that setting. Note that Theorem~\ref{thm:positional_optimal} does not hold on infinite arenas (i.e. with an infinite number of states): 
Figure~\ref{fig:counter_example} depicts an infinite concurrent reachability game where the state $q_0$ is maximizable but, from $q_0$, Player $\A$ does not have any positional optimal strategy. Indeed, in state $s$ is plugged the game of Figure~\ref{fig:arbitrarilyClose}, whose value is $1$ but Player $\A$ does not have an optimal strategy. Then, for all $i \geq 0$, the probability to reach $s$ from $s_i$ is equal to $v_i = (1/2 + 1/2^i) > 1/2$. Hence, if Player $\A$ plays an $0 < \varepsilon_i$-optimal strategy in $s$ such that $(1 - \varepsilon_i) \cdot q_i > 1/2$, then the value of the state $s_i$ is greater than $1/2$. In that case, in the states $c_i$, Player $\B$ plays the second columns obtaining the value $1/2$. This induces that the value in all states $q_i$ is $1/2$. However, this is only possible if Player $\A$ has (infinite) memory, 
since the greater the index $i$ considered, the smaller the value of $\varepsilon_i$ needs to be to ensure $(1 - \varepsilon_i) \cdot q_i \geq 1/2$ while still ensuring $\varepsilon_i > 0$ (since Player $\A$ does not have an optimal strategy from $s$). In particular, for any Player $\A$ positional strategy $\s_\A$ from $q_0$ that is $0 < \varepsilon$-optimal in $s$, the value -- w.r.t. the strategy $\s_\A$ -- of all states $\s_i$ for indexes $i$ such that $(1 - \varepsilon) \cdot q_i < 1/2$  is smaller than $1/2$. In which case, Player $\B$ plays the first column in $c_i$, thus obtaining a value smaller than $1/2$. It follows that the value of all states $(q_n)_{n \geq 0}$ -- w.r.t. the strategy $\s_\A$ -- is smaller than $1/2$. Hence, any Player $\A$ positional strategy is not optimal from $q_0$. Appendix~\ref{appen:infinite_games} gives additional details. Note that, when considering MDPs instead of two-player games, optimal strategies need not exist but when they do there necessarily are positional ones (see for instance \cite{DBLP:conf/concur/KieferMST20}).

\smallskip\noindent \textbf{Computing the set of maximizable states.}
Finally, consider the problem, given a finite concurrent reachability game, to effectively compute the set of maximizable and sub-maximizable states (assuming the probability distribution of the Nature states are rational). In fact, this can be done by using the theory of the reals. 
\begin{definition}[First-order theory of the reals]
	\label{def:fo_reals}
	The \emph{first-order theory of the reals} (denoted $\mathsf{FO}$-\R) corresponds to the well-formed sentences of first-order logic (i.e. with universal and existential quandtificators)
	, also involving logical combinations of equalities and inequalities of real polynomials, with integer coefficients.
\end{definition}
The first-order theory of the reals is decidable \cite{RE89}, i.e. determining if a given formula belonging to that theory is true is decidable. Now, let us consider a finite concurrent reachability game $\Aconc$ and a state $q \in Q$. It is possible to encode, with an $\mathsf{FO}$-$\R$ formula, that the state $q$ is maximizable, i.e. $q \in \OS_\A$. First, note that, given two positional strategies $\s_\A$ and $\s_\B$ for both players, it is possible to compute the value of the game with the theory of reals: it amounts to finding the least fixed point of the operator $\Delta$ with the strategies of both players fixed. Then, $q$ being maximizable, denoting $u := \val{\Aconc}{}(q) \in [0,1]$ its value, is equivalent to having a Player $\A$ positional strategy ensuring at least $u$ (against all Player $\B$ positional strategies) and no Player $\A$ positional strategy ensures more than $u$ (as $\varepsilon$-optimal positional strategies always exists for Player $\A$ \cite{everett57}). This can be expressed in $\mathsf{FO}$-\R. The theorem below follows.
\begin{theorem}[Complement~\ref{appen:computing_optimal_states}]
	\label{thm:optimal_state_computable}
	In a finite concurrent reachability game with rational distributions, the set of maximizable states is computable.
\end{theorem}

	\section{Maximizable states and game forms}
	\label{sec:optimal_game_forms}

In the previous section, we were given a concurrent reachability game and we considered a construction to compute exactly the sets of maximizable and sub-maximizable states. It is rather cumbersome as it requires two nested fixed point procedures. Now, we would like to have a structural condition ensuring that if a game is built correctly (i.e. built from reach-maximizable local interactions), then all states are maximizable. More specifically, 
in this section, we characterize exactly the \emph{reach-maximizable} game forms, that is the game forms such that every reachability game built with these game forms as local interactions have only maximizable states. 

First, let us characterize a necessary condition for game forms to be reach-maximizable. We want for reach-maximizable game forms to behave properly when used individually. That is, from a game form $\formNF$ and a partial valuation $\alpha: \outComeNF \setminus E \rightarrow [0,1]$ of the outcomes, we define a three-state reachability game $\langle \Aconc_{(\formNF,\alpha)},\top \rangle$
. Note that such games were previously studied in \cite{KL74}. We illustrate this construction on an example.

\begin{figure}
	\begin{minipage}[b]{0.6\linewidth}
		\centering
		\includegraphics[scale=0.7]{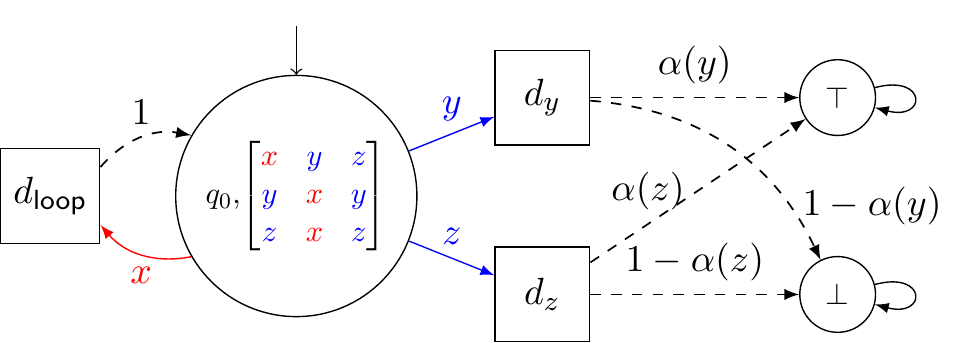}
		\caption{The three-state reachability game $\langle \Aconc_{(\formNF,\alpha)},\top \rangle$ built from the game form $\formNF$ for some partial valuation $\alpha: \outComeNF \setminus E \rightarrow [0,1]$ with $E = \{ x \}$.}
		\label{fig:one_state}           
	\end{minipage}
	\begin{minipage}[b]{0.35\linewidth}
		\centering
		\includegraphics[scale=1.2]{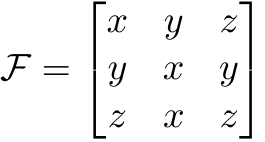}
		\caption{The game form that constitutes the local interaction in the state $q_0$.}
		\label{fig:formNF}         
	\end{minipage}
\end{figure}
\begin{example}
	In Figure~\ref{fig:one_state}, a three-state reachability game $\langle \Aconc_{(\formNF,\alpha)},\top \rangle$ is built from a game form $\formNF = \langle \St_\A,\St_\B,\{ x,y,z \},\outCNF \rangle$ -- with $\outCNF$ depicted in Figure~\ref{fig:formNF} -- and a partial valuation $\alpha: \{ y,z \} \rightarrow [0,1]$. We have a one-to-one correspondence between the outcomes of the game form $\formNF$ and the Nature states of the reachability game $\langle \Aconc_{(\formNF,\alpha)},\top \rangle$ via the bijection $g: \{ x,y,z \} \rightarrow \distribSet$ such that $g(x) = d_{\mathsf{loop}}$ and for $u \in \{ y,z \}$, $g(u) = d_{u}$. Furthermore, in the reachability game $\langle \Aconc_{(\formNF,\alpha)},\top \rangle$, we have $\lfp(\top) = 1$ and $\lfp(\bot) = 0$. Therefore, for $u \in \{ y,z \}$, we have $\mu_{\lfp} \circ g(u) = \alpha(u)$. In fact, this game is built so that $v_\alpha = \lfp(q_0)$ and $\mu_{\lfp} 
	= \limval{\alpha} \circ g^{-1}$ (recall that $\limval{\alpha}$ is the (total) valuation induced by the partial valuation $\alpha$ from Definition~\ref{def:partial_assignment}).
	
	%
	Let us now determine at which condition on the pair $(\formNF,\alpha)$ is the starting state $q_0$ maximizable in $\Aconc_{(\formNF,\alpha)}$. If we have $v_\alpha = \lfp(q_0) = 0$, the state $q_0$ is maximizable in any case. Now, assume that $v_\alpha = \lfp(q_0) > 0$. Recall the construction of the previous section, specifically the set of secure states w.r.t. a set of bad states (Definition~\ref{def:secure_states}). Initially, $\Bad_0 = \emptyset$, so we want for the state $q_0$ to be in $\Scr(\emptyset)$, i.e. 
	we want (and need) 
	an efficient strategy in the state $q_0$ where the set of good states $\Gex$ 
	is the target $\Gex = \{ \top \}$ and the set of bad states is empty. In that case, the set of efficient strategies coincide with the set of progressive strategies. Thus, 
	$q_0$ is maximizable if and only if $\Prg_{q_0}(\{ \top \}) \neq \emptyset$. We assume for simplicity that $\alpha(y),\alpha(z) > 0$, 
	hence the set Nature states $\Gex_\distribSet$ with a non-zero probability to reach $\top$ is $\{ g(y),g(z) \} \subseteq \distribSet$. 
	By definition of $\Prg$ (Definition~\ref{def:progressive_strat}), $\Prg_{q_0}(\{ \top \}) \neq \emptyset$ amounts to have an optimal strategy $\sigma_\A$ 
	in $\gameNF{\formNF_{q_0}}{\mu_{\lfp}}$ such that, for all $b \in B_{\sigma_\A}: \delta(q_0,\Supp(\sigma_\A),b) \cap \{ g(y),g(z) \} \neq \emptyset$ or, equivalently, $\delta(q_0,\Supp(\sigma_\A),b) \not \subseteq \{ g(x) \}$. 
	In terms of $\formNF$ and $\alpha$, the state $q_0$ is maximizable if and only if there is an optimal strategy $\sigma_\A$ in $\gameNF{\formNF}{\limval{\alpha}}$ such that, for all $b \in B_{\sigma_\A}: \outCNF(\Supp(\sigma_\A),b) \not \subseteq \{ x \} = E$ if the partial valuation $\alpha$ is defined as $\alpha: \outComeNF \setminus E \rightarrow [0,1]$ for $\outComeNF = \{ x,y,z \}$ and $E = \{ x \}$. 
\end{example}

This suggests the definition below of \emph{reach-maximizable game form} w.r.t. a partial valuation.
\begin{definition}[Reach-maximizable game forms w.r.t. a partial valuation]
	\label{def:correct_gf_partial_valuation}
	Consider a game form $\formNF
	$ and a partial valuation of the outcomes $\alpha: \outComeNF \setminus E \rightarrow [0,1]$
	. The game form $\formNF$ is \emph{reach-maximizable} w.r.t. the partial valuation $\alpha$ if $v_\alpha = 0$ or there exists an optimal strategy $\sigma_\A \in \Opt_\A(\gameNF{\formNF}{\limval{\alpha}})$ such that for all $b \in 
	B_{\sigma_\A}$, we have $\outCNF(\Supp(\sigma_\A),b) \not \subseteq E$. 
	Such strategies are said to be \emph{reach-maximizing} w.r.t. $\alpha$.
\end{definition}

This definition was chosen to ensure the lemma below.
\begin{lemma}[Proof~\ref{proof:one_state_game_opt_equiv_correct}]
	\label{prop:one_state_game_opt_equiv_correct}
	Consider a game form $\formNF
	$ and a partial valuation of the outcomes $\alpha: \outComeNF \setminus E \rightarrow [0,1]$
	. The initial state (and thus all states) in the three-state reachability game $\Aconc_{(\formNF,\alpha)}$ is maximizable if and only if the game form $\formNF$ is reach-maximizable w.r.t. the partial valuation $\alpha$.
\end{lemma}

The definition of reach-maximizable game form is then obtained via a universal quantification over the partial valuations considered.
\begin{definition}[Reach-maximizable game form]
	\label{def:correct_gf}
	Consider a game form $\formNF = \langle \St_\A,\St_\B,\outComeNF,\outCNF \rangle$. It is a \emph{reach-maximizable} (RM for short) game form if it is reach-maximizable w.r.t. all partial valuations $\alpha: \outComeNF \setminus E \rightarrow [0,1]$.
\end{definition}

Lemma~\ref{prop:one_state_game_opt_equiv_correct} gives that RM game forms behave properly when 
used individually, such as in three-state reachability games. Let us now look at how such game forms behave collectively, that is we consider concurrent reachability games where all local interactions are RM. 
In fact, in such a setting,
all states are maximizable. This is 
stated in the lemma below. 
\begin{lemma}[Proof~\ref{proof:lem_correct_opt_everywhere}]
	\label{lem:correct_opt_everywhere}
	Consider a concurrent reachability game $\langle \Aconc,\top \rangle$ and assume that all local interactions are RM game forms. Then, all states are maximizable: $Q = \OS_\A$.
\end{lemma}
\begin{proofs}
	We show that $Q = \OS_\A$ by showing that $\Bad = \emptyset$, which is equivalent since, by Theorem~\ref{thm:positional_optimal}, we have $\Bad = \SubOS_\A = Q \setminus \OS_\A$. That is, we consider the iterative construction of the set of sub-maximizable states of the previous section and we show that $\Bad_1 = Q \setminus (\Scr(\Bad_0)) = \emptyset = \Bad_0$ (see Definition~\ref{def:bad_states}), which induces that $\Bad = \emptyset$. Let us assume towards a contradiction that $Q \setminus (\Scr_n(\emptyset) \cup \lfp^{-1}[0]) \neq \emptyset$ for $n = |Q|$. Since $\Rsk_q(\emptyset) = \emptyset$ for all $q \in q$, any efficient strategy in a state $q$ w.r.t. to the sets $\Scr_n(\emptyset)$ and $\emptyset$ is in fact a progressive strategy w.r.t. the set $\Scr_n(\emptyset)$. Hence, the goal is to exhibit such a progressive strategy in a state $q \in Q \setminus \Scr(\emptyset)$, thus showing a contradiction with the fact that $q \notin \Scr(\emptyset)$. We consider the states with the greatest value -- w.r.t. $\lfp$ -- as we can hope that they are the more likely to have progressive strategies. That is, for $x := \max_{q \in Q \setminus \Scr_n(\emptyset)} \lfp(q) > 0$ 
	the maximum of $\lfp$, we set $Q_x := \lfp^{-1}[x] \setminus \Scr_n(\emptyset) \neq \emptyset$ the set of states realizing that maximum. We want to use the assumption that all local interactions are RM. That is, we need to define a partial valuation on the outcomes of the local interactions, i.e. on Nature states. First, let us define its domain. We can find intuition in the example of the three-state reachability game in Figure~\ref{fig:one_state}: the outcome that is not valued by the partial valuation considered is the Nature state looping on the state $q_0$. Note that its value w.r.t. $\mu_\lfp$ is the same as the value of the state $q_0$ w.r.t. $\lfp$. In our case, we consider the set of Nature states $\distribSet_x$ realizing this value $x$ that cannot reach the set $\Scr_n(\emptyset)$, that is $\distribSet_x := \mu_{\lfp}^{-1}[x] \setminus \Scr_n(\emptyset)_\distribSet$. Then, we define the partial valuation of the Nature states $\alpha: \distribSet \setminus \distribSet_x \rightarrow [0,1]$ by $\alpha := \restriction{\mu_\lfp}{\distribSet \setminus \distribSet_x}$. 
	Now, we can show that there exists a state $q \in Q_x$ such that $\limval{\alpha} = \mu_\lfp$ in the game form $\formNF_q$. By maximality of $x$, we can prove that any local strategy $\sigma_A$ in $\formNF_q$ that is reach-maximizing w.r.t. the partial valuation $\alpha$ of the outcomes of $\formNF_q$ is a progressive strategy 
	w.r.t. $\Scr_n(\emptyset)$ in $\formNF_q$. Equivalently, $\sigma_\A$ is efficient w.r.t. $\Scr_n(\emptyset)$ and $\emptyset$. Hence the contradiction with the fact that $q \notin \Scr(\emptyset)$.
\end{proofs}

Overall, we obtain the theorem below.
\begin{theorem}[Proof~\ref{proof:main}]
	\label{thm:main}
	For a set of game forms $\mathcal{G}$, all states in all concurrent reachability games with local interactions in $\mathcal{G}$ are maximizable if and only if all game forms in $\mathcal{G}$ are RM.
\end{theorem}


\smallskip \noindent \textbf{Deciding if game forms are RM.} Consider the following decision problem $\mathsf{RMGF}$: given a game form, decide if it is a RM game form. 
We proved Theorem~\ref{thm:optimal_state_computable} by showing that the fact that a state is maximizable in a concurrent reachability game can be encoded in the theory of the reals ($\mathsf{FO}$-\R). Since Lemma~\ref{prop:one_state_game_opt_equiv_correct} ensures that a game form $\formNF$ is RM w.r.t. a partial valuation $\alpha$ if and only if the initial state in the three-state reachability game $\Aconc_{(\formNF,\alpha)}$ is maximizable, it follows that, via a universal quantification over partial valuations, the fact that a game form is RM can be encoded in the theory of the reals. Note that it can also be encoded directly from the definition of RM game form.
We obtain the theorem below.
\begin{proposition}[Complement~\ref{proof:decidability}]
	\label{prop:decidability}
	The problem $\mathsf{RMGF}$ is decidable.
\end{proposition}

\smallskip \noindent \textbf{Determined game forms and RM game forms} In \cite{FromLocalToGlobal}, the authors have studied a problem similar to the one we considered in this section: determining the game forms ensuring that, when used as local interaction in a concurrent game (with an arbitrary Borel winning condition
), the game is determined (i.e. either of the players has a winning strategy). 
The authors have shown that these game forms exactly correspond to \emph{determined} game forms. These roughly correspond to game forms where, for all subsets of outcomes $E \subseteq \outComeNF$, there is either of line of outcomes in $E$ or a column of outcomes in $\outComeNF \setminus E$, as formally defined below.
\begin{definition}[Determined game forms]
	Consider a game form $\formNF = \langle \St_\A,\St_\B,\outComeNF,\outCNF \rangle$. It is \emph{determined} if, for all subsets of outcomes $E \subseteq \outComeNF$, either there exists some $a \in \St_\A$ such that $\outCNF(a,\St_\B) \subseteq E$ or there exists some $b \in \St_\B$ such that $\outCNF(\St_\A,b) \subseteq \outComeNF \setminus E$.
\end{definition}
In fact, they proved an equivalence between turn-based games and concurrent games using determined game forms as local interactions, which holds also when the game is stochastic. In fact, positional optimal strategies exists for both players in turn-based reachability games \cite{CJH04}, it is also the case in concurrent reachability games with determined local interactions. This result, combined with Theorem~\ref{thm:main} gives immediately that determined game forms are RM. Interestingly, determined game forms can also be characterized with the least fixed point operator as in the proposition below.
\begin{proposition}[Proof~\ref{appen:determined_game_form}]
	\label{prop:determined_gf}
	A game form $\formNF
	$ is \emph{determined} if and only if, for all partial valuations $\alpha: \outComeNF \setminus E \rightarrow [0,1]$ of the outcomes, we have $v_\alpha = f^\formNF_\alpha(0)$. 
In particular, this implies that all determined game forms are RM.
\end{proposition}

	\section{Future Work}
	\label{sec:conclusion}
	In this paper we give a double-fixed-point procedure to compute
maximizable and sub-maximizable states in a stochastic concurrent
reachability (finite) game. Our procedure yields \emph{de facto}
positional witnesses for the strategies.
As further natural work, we seek studying more general objectives. It
is however interesting to notice that, as mentioned in the
introduction, it will not be so easy since even B\"uchi games do not
enjoy positional almost optimal strategies~\cite[Theorem 2]{AM04b}.

We also plan to better grasp RM game forms, and understand
what are RM game forms for the two players, or analyze the
complexity of the $\mathsf{RMGF}$ problem.


        \bibliography{biblio}
	
	\appendix

\section{Complements on Section~\ref{sec:gameForms}}
We make an straightforward remark that comes directly from the definition of optimal strategies
\begin{remark}
	\label{rmq:optimal_strategies}
	In a game in normal form $\formNF = \langle \St_\A,\St_\B,[0,1],\outCNF 
	\rangle$, an optimal strategy $\sigma_\A \in \Dist(\St_\A)$ (resp. $\sigma_\B \in \Dist(\St_\A)$) for Player $\A$ (resp. $\B$) ensures that for all strategy $\sigma_\B \in \Dist(B)$ (resp. $\sigma_\A \in \Dist(A)$), we have $\va_\formNF \leq \outM_\formNF(\sigma_\A,\sigma_\B)$ (resp. $\outM_\formNF(\sigma_\A,\sigma_\B) \leq \va_\formNF$).
\end{remark}

We also have the following observation.
\begin{observation}
	\label{obs:smaller_valuation_smaller_outcome}
	Consider a game form $\formNF = \langle \St_\A,\St_\B,\outComeNF,\outCNF \rangle$, two valuations $v,v' \in [0,1]^\outComeNF$ and $x \in \R$ such that $v + x \preceq v'$. Consider an arbitrary Player $\A$ strategy $\sigma_\A \in \Dist(\St_\A)$. Then, the value of the strategy $\sigma_\A$ in $\formNF$ w.r.t. $v$ plus $x$ is lower than or equal its value w.r.t. $v'$: $\va_{\gameNF{\formNF}{v}}(\sigma_\A) + x \leq \va_{\gameNF{\formNF}{v'}}(\sigma_\A)$. Following, this also holds for the value of the game: $\va_{\gameNF{\formNF}{v}} + x \leq \va_{\gameNF{\formNF}{v'}}$.
\end{observation}
\begin{proof}
	Consider two arbitrary strategies $\sigma_\A,\sigma_\B$ for Player $\A$ and $\B$. We have:
	\begin{align*}
		\outM_{\gameNF{\formNF}{v'}}(\sigma_\A,\sigma_\B) & = \sum_{a \in \St_\A} \sum_{b \in \St_\B} \sigma_\A(a) \cdot \sigma_\B(b) \cdot \underbrace{v' \circ \outCNF(a,b)}_{\geq v \circ \outCNF(a,b) + x} \\ & \geq \sum_{a \in \St_\A} \sum_{b \in \St_\B} \sigma_\A(a) \cdot \sigma_\B(b) \cdot v \circ \outCNF(a,b) + x \\ & = \outM_{\gameNF{\formNF}{v'}}(\sigma_\A,\sigma_\B) + x
	\end{align*}
	It follows that $\va_{\gameNF{\formNF}{v}}(\sigma_\A) + x \leq \va_{\gameNF{\formNF}{v'}}(\sigma_\A)$ for any arbitrary strategy $\sigma_\A$. In particular, for an optimal strategy $\sigma_\A \in \Opt_\A(\gameNF{\formNF}{v})$ in the game form $\formNF$ w.r.t. the valuation $v$, we have: $\va_{\gameNF{\formNF}{v}} + x = \va_{\gameNF{\formNF}{v}}(\sigma_\A) + x \leq \va_{\gameNF{\formNF}{v'}}(\sigma_\A) \leq \va_{\gameNF{\formNF}{v'}}$. 
\end{proof}


\subsection{Complements on Definition~\ref{def:partial_assignment}}
\label{proof:valuation_sequence_converging}
Consider a game form $\formNF$ and a partial valuation $\alpha: \outComeNF \setminus E \rightarrow [0,1]$ for some subset of outcomes $E$. For all $y,y' \in [0,1]$, if $y \leq y'$, then $\alpha[y] \preceq \alpha[y']$, hence, by Observation~\ref{obs:smaller_valuation_smaller_outcome}, we have:
\begin{displaymath}
f^\formNF_\alpha(y) = \va_{\gameNF{\formNF}{\alpha[y]}} \leq \va_{\gameNF{\formNF}{\alpha[y']}} \leq f^\formNF_\alpha(y')
\end{displaymath}
That is, the function $f^\formNF_\alpha: [0,1] \rightarrow [0,1]$ preserves the relation $\leq$ (i.e. is non-decreasing). By Knaster-Tarski theorem \cite{tarski55}, $f^\formNF_\alpha$ admits a least fixed point in $[0,1]$.


\section{Complement on Section~\ref{sec:concurrentGames}}
In the following, we will be using the lemma below relating the transition probability $\probTrans{q,q'}$ between two states $q,q'$ and the valuation on Nature states lifting a valuation on states.
\begin{proposition}
	\label{prop:valuation}
	Consider a valuation of the states $v \in [0,1]^Q$, a state $q \in Q$ and strategies $\sigma_\A,\sigma_\B \in \Dist(A) \times \Dist(B)$ for both players in the game form $\formNF_q$. We have the following relation:
	\[\sum_{q' \in Q} \probTrans{q,q'}(\sigma_\A,\sigma_\B)
	\cdot v(q') = \outM_{\gameNF{\formNF_q}{\mu_v}}(\sigma_\A,\sigma_\B)\]
\end{proposition}
\begin{proof}
	The result comes immediately by writting the definitions, by inverting sums over the states and the actions:
	\begin{align*}
		\sum_{q' \in Q} \probTrans{q,q'}(\sigma_\A,\sigma_\B)
		\cdot v(q') & = \sum_{q' \in Q} \outM_{\gameNF{\formNF_q}{\distribFunc(\cdot)(q')}}(\sigma_\A,\sigma_\B)
		\cdot v(q') \\ & = \sum_{q' \in Q} \left( \sum_{a \in A} \sum_{b \in B} \sigma_\A(a) \cdot \sigma_\B(b) \cdot \distribFunc(\delta(q,a,b))(q')
		\right) \cdot v(q') \\
		& = \sum_{a \in A} \sum_{b \in B} \sigma_\A(a) \cdot \sigma_\B(b) \cdot \left( \sum_{q' \in Q} \distribFunc(\delta(q,a,b))(q')
		\cdot v(q') \right)\\
		& = \sum_{a \in A} \sum_{b \in B} \sigma_\A(a) \cdot \sigma_\B(b) \cdot \mu_v(\delta(q,a,b))\\
		& = \outM_{\gameNF{\formNF_q}{\mu_v}}(\sigma_\A,\sigma_\B)
	\end{align*}
\end{proof}


\subsection{Complements on Value of the game and least fixed point}
\label{appen:value_lfp}


Consider $\langle \Aconc,\top \rangle$ a concurrent reachability game and the operator $\Delta$ from Definition~\ref{def:operator}. Let $V$ denote the set of valuations $V := \{ v \in [0,1]^Q \mid v(\top) = 1 \}$. We state the proposition below giving some properties on $\Delta$ and $V$.
\begin{proposition}
	\label{prop:property_Delta}
	The operator $\Delta$ and the set $V$ ensure the following properties:
	\begin{enumerate}
		\item $(V,\preceq)$ is a complete lattice with minimal element denoted $v_0$;
		\item $\Delta[V] \subseteq V$;
		\item $\Delta$ is non-decreasing;
		;
		\item $\Delta$ is 1-Lipschitz w.r.t. $\dNorm$ (the infinity norm on $[0,1]^Q$).
	\end{enumerate}
\end{proposition}
\begin{proof}
	\begin{enumerate}
		\item The relation $\preceq \subseteq [0,1]^Q \times [0,1]^Q$ is a partial order on $V$. All subset $A \subseteq V$ of $V$ has an infimum $m_A \in V$, defined by, for all $x \in Q$, we have $m_A(x) = \inf_{v \in A} v(x)$, and a supremum $M_A \in V$, defined by, for all $x \in Q$, we have $M_A(x) = \sup_{v \in A} v(x)$. The minimal element $v_0$ of $V$ is defined by $v_0(x) = 0$ for all $x \in Q \setminus \{ \top \}$.
		\item For all $v \in V$ and $q \in Q$, we have:
		\begin{itemize}
			\item if $q = \top$, $\Delta(v)(q) = 1$;
			\item otherwise $\Delta(v)(q) = \va_{\gameNF{\formNF_q}{\mu_v}} \in [0,1]$ since $\mu_v \in [0,1]^\distribSet$.
		\end{itemize}
		That is, $\Delta(v) \in V$. It follows that $\Delta(V) \subseteq V$.
		\item Consider two elements $v,v' \in [0,1]^Q$ such that $v \preceq v'$. For all $d \in \distribSet$, we have $\mu_v(d) = \sum_{q \in Q} \mu_q(d) \cdot v(q) \leq \sum_{q \in Q} \mu_q(d) \cdot v'(q) = \mu_{v'}(d)$. Therefore, $\mu_v \preceq \mu_{v'}$. Hence, by Observation~\ref{obs:smaller_valuation_smaller_outcome}, for all $q \in Q$, we have $\va_{\gameNF{\formNF_q}{\mu_v}} \leq \va_{\gameNF{\formNF_q}{\mu_{v'}}}$. It follows that $\Delta(v) \preceq \Delta(v')$.
		\item Let $q \in Q \setminus \{ \top \}$. Let us prove that the $q$-th component of $\Delta$ is 1-Lipschitz. Let us denote by $\formNF$ the game form $\formNF_q
		$. Consider 
		two valuations $v,v' \in [0,1]^Q$
		. First, consider an arbitrary pair of strategies $(\kappa_\A,\kappa_\B) \in \Dist(A) \times \Dist(B)$. Then, we have the following (that we will refer to as $(1)$):
		\begin{align*}
			|\outM_{\gameNF{\formNF_q}{\mu_v}}(\kappa_\A,\kappa_\B) - \outM_{\gameNF{\formNF_q}{\mu_{v'}}}(\kappa_\A,\kappa_\B)| & = |\sum_{a \in \St_\A} \sum_{b \in \St_\B} \kappa_\A(a) \cdot \kappa_\B(b) \cdot (\mu_v(\outCNF(a,b)) - \mu_{v'}(\outCNF(a,b)))| \\
			& \leq \sum_{a \in \St_\A} \sum_{b \in \St_\B} \kappa_\A(a) \cdot \kappa_\B(b) \cdot |\underbrace{\mu_v(\outCNF(a,b)) - \mu_{v'}(\outCNF(a,b))}_{=\sum_{q \in Q} \distribFunc(d)(q) \cdot (v(q) - v'(q))}| \\
			& \leq \sum_{a \in \St_\A} \sum_{b \in \St_\B} \kappa_\A(a) \cdot \kappa_\B(b) \cdot \left( \sum_{q \in Q} \distribFunc(d)(q) \cdot \underbrace{|v(q) - v'(q)|}_{\leq \dNorm(v,v')
			} \right)\\
			& \leq \left( \sum_{a \in \St_\A} \sum_{b \in \St_\B} \kappa_\A(a) \cdot \kappa_\B(b) \cdot \left( \sum_{q \in Q} \distribFunc(d)(q) \right) \right) \cdot \dNorm(v,v') \\
			& = \left( \sum_{a \in \St_\A} \sum_{b \in \St_\B} \kappa_\A(a) \cdot \kappa_\B(b) \right) \cdot \dNorm(v,v') \\
			& = \dNorm(v,v')
		\end{align*}
		Now, consider two pairs of strategies $(\sigma_\A,\sigma_\B) \in \Dist(A)  \times \Dist(B)$ and $(\sigma_\A',\sigma_\B') \in \Dist(A) \times \Dist(B)$ that are optimal for both players in the games in normal form $\gameNF{\formNF_q}{\mu_{v}}$ and $\gameNF{\formNF_q}{\mu_{v'}}$ respectively. We have the following:
		\begin{align*}
			\outM_{\gameNF{\formNF_q}{\mu_{v}}}(\sigma_\A,\sigma_\B) & = \va_{\gameNF{\formNF_q}{\mu_{v}}} & \text{ since } \sigma_\A,\sigma_\B \text{ are optimal in } \gameNF{\formNF_q}{\mu_{v}}\\
			& \leq \outM_{\gameNF{\formNF_q}{\mu_{v}}}(\sigma_\A,\sigma_\B') & \text{ since } \sigma_\A \text{ is optimal in } \gameNF{\formNF_q}{\mu_{v}} \\
			& \leq \outM_{\gameNF{\formNF_q}{\mu_{v'}}}(\sigma_\A,\sigma_\B') + \dNorm(v,v') & \text{ by } (1)\\
			& \leq \va_{\gameNF{\formNF_q}{\mu_{v'}}} + \dNorm(v,v') & \text{ since } \sigma_\B' \text{ is optimal in } \gameNF{\formNF_q}{\mu_{v'}}\\
			& = \outM_{\gameNF{\formNF_q}{\mu_{v'}}}(\sigma_\A',\sigma_\B') + \dNorm(v,v') & \text{ since } \sigma_\A',\sigma_\B' \text{ are optimal in } \gameNF{\formNF_q}{\mu_{v'}}
		\end{align*}
		Similarly, we obtain: $\outM_{\gameNF{\formNF_q}{\mu_{v'}}}(\sigma_\A',\sigma_\B') \leq \outM_{\gameNF{\formNF_q}{\mu_{v}}}(\sigma_\A,\sigma_\B) + \dNorm(v,v')$. 
		It follows that:
		\begin{displaymath}
		|\outM_{\gameNF{\formNF_q}{\mu_{v}}}(\sigma_\A,\sigma_\B) - \outM_{\gameNF{\formNF_q}{\mu_{v'}}}(\sigma_\A',\sigma_\B')| \leq \dNorm(v,v')
		\end{displaymath}
		
		Then, we have:
		\begin{align*}
			|\Delta(v)(q) - \Delta(v')(q)| = |\va_{\gameNF{\formNF_q}{\mu_{v}}} - \va_{\gameNF{\formNF_q}{\mu_{v'}}}|  = 
			|\outM_{\gameNF{\formNF_q}{\mu_{v}}}(\sigma_\A,\sigma_\B) - \outM_{\gameNF{\formNF_q}{\mu_{v'}}}(\sigma_\A',\sigma_\B')| \leq \dNorm(v,v') 
		\end{align*}
		
	\end{enumerate}
	Therefore, all $q$-th component of the function $\Delta$ are 1-Lipschitz 
	, and therefore the whole function $\Delta$ is 1-Lipschitz with regard to the distance $\dNorm$.
\end{proof}

The value of the game is now given by least fixed point of the function $\Delta$ on $V$. 
\begin{definition}
	Let $\lfp$ be the least fixed point of the function $\Delta$ on $V$. Note that its existence is ensured by Knaster-Tarski \cite{tarski55} theorem with points $1$ and $2$.
\end{definition}

In the following, we prove (the already existing) Theorem~\ref{thm:m_delta}. First, let us state the useful proposition below allowing us to express the probability to reach the target $\top$ in less than $n$ steps inductively.
\begin{proposition}
	\label{rmq:relation_proba}
	Consider two strategies $\nu_\A,\nu_\B \in \SetStrat{\Aconc}{\A} \times \SetStrat{\Aconc}{\B}$ for Player $\A$ and $\B$ and a state $q \in Q \setminus \{ \top \}$. Then, for all $n \geq 0$, we have the following relation
	:
	\begin{displaymath}
	\prob{q}{\nu_\A,\nu_\B}(n+1,\top) = \sum_{q' \in Q}
	\probTrans{q,q'}(\nu_\A(q),\nu_\B(q)) \cdot \prob{q'}{\nu_\A^q,\nu_\B^q}(n,\top)
	\end{displaymath}
\end{proposition}
\begin{proof}
	We fix $\nu_\A,\nu_\B,n$ and $q$ as in the proposition. We have:
	\begin{align*}
		\prob{q}{\nu_\A,\nu_\B}(n+1,\top) & = \sum_{\pi \in \mathsf{Rch}^{n+1}(q,\top)} \prob{q}{\nu_\A,\nu_\B}(\pi) & \text{ by definition of }\prob{q}{\nu_\A,\nu_\B}(n+1,\top) \\
		& = \sum_{q' \in Q} \left(\sum_{\pi \in \mathsf{Rch}^n(q',\top)} \prob{q}{\nu_\A,\nu_\B}(q \cdot \pi) \right) & \text{ by definition of } \mathsf{Rch}(q,\top)\\
		& = \sum_{q' \in Q} \probTrans{q,q'}(\nu_\A(q),\nu_\B(q)) \cdot \left(\sum_{\pi \in \mathsf{Rch}^n(q',\top)} \prob{q'}{\nu_\A^q,\nu_\B^q}(\pi) \right) 
		& \text{ by definition of } \prob{q'}{\nu_\A^q,\nu_\B^q}(q \cdot \pi)\\
		& = \sum_{q' \in Q} \probTrans{q,q'}(\nu_\A(q),\nu_\B(q)) \cdot \prob{q'}{\nu^q_\A,\nu^q_\B}(n,\top) 
		& \text{ by definition of } \prob{q'}{\nu^q_\A,\nu^q_\B}(n,\top)\\
	\end{align*}
\end{proof}

Now, we state in the lemma below that Player $\B$ has a positional strategy $\s_\B$ whose value is less than or equal to $\lfp(q)$ from all state $q \in Q$.
\begin{lemma}
	\label{lem:leq_m_delta}
	Consider a concurrent reachability game $\langle \Aconc,\top \rangle$ 
	and $\Delta: [0,1]^Q \rightarrow [0,1]^Q$ the operator on values defined in Definition~\ref{def:operator}. There exists a positional strategy $\s_\B \in \SetPosStrat{\B}{\Aconc}$ for Player $\B$ such that, for all $q \in Q$: $\val{\Aconc}{\s_\B}(q) \leq \lfp(q)$.
\end{lemma}
\begin{proof}
	We consider a positional strategy $\s_\B: Q^+ \rightarrow \Dist(B)$ for Player $\B$ ensuring, for all $q \in Q$, $\s_\B(q)$ is an optimal strategy for Player $\B$ in the game form $\gameNF{\formNF_q}{\mu_{\lfp}}$: $\s_\B(q) \in \Opt_\B(\gameNF{\formNF_q}{\mu_{\lfp}}) \neq \emptyset$. 
	
	Consider now some state $k \in Q$ and let us show that for all Player $\A$ strategies $\s_\A: Q^+ \rightarrow \Dist(A)$, we have $\prob{k}{\s_\A,\s_\B}(\top) \leq \lfp(k)$.
	In fact, we show by induction on $n \in \N$ the property $\mathcal{P}(n)$: for all strategies $\s_\A: Q^+ \rightarrow \Dist(A)$ and for all $q \in Q$, $\prob{q}{\s_\A,\s_\B}(n,\top) \leq \lfp(q)$. 
	
	The case $n = 0$ is straightforward, since regardless of the strategy $\s_\A: Q^+ \rightarrow \Dist(A)$ considered, for all $q \in Q$, we have $\prob{q}{\s_\A,\s_\B}(0,\top) = 0$ if $q \neq \top$ and $\prob{q}{\s_\A,\s_\B}(0,\top) = 1$ if $q = \top$. It follows that $\prob{q}{\s_\A,\s_\B}(0,\top) \leq \lfp(q)$ since $\lfp \in V$.
	
	Let us now assume that the property $\mathcal{P}(n)$ holds for some $n \in \N$. For all strategies $\s_\A: Q^+ \rightarrow \Dist(A)$ we have $\prob{\top}{\s_\A,\s_\B}(n+1,\top) = 1 = \lfp(\top)$. Consider now a state $q \in Q \setminus \{ \top \}$ and a Player $\A$ strategy $\s_\A: Q^+ \rightarrow \Dist(A)$. Since the strategy $\s_\B$ is positional, it is equal to its residual strategy: $\s_\B^q = \s_\B$. Therefore, we have the following:
	\begin{align*}
		\prob{q}{\s_\A,\s_\B}(n+1,\top) & = \sum_{q' \in Q}
		\probTrans{q,q'}(\s_\A(q),\s_\B(q))
		\cdot \prob{q'}{\s_\A^q,\s_\B}(n,\top) & \text{ by Proposition}~\ref{rmq:relation_proba} \text{ and since } \s_\B^q = \s_\B \\
		& \leq \sum_{q' \in Q} \probTrans{q,q'}(\s_\A(q),\s_\B(q))
		\cdot \lfp(q') & \text{ by } \mathcal{P}(n) \\
		& = \outM_{\gameNF{\formNF_q}{\mu_{\lfp}}}(\s_\A(q),\s_\B(q)) & \text{ by Proposition~\ref{prop:valuation} } \\
		& \leq \va_{\gameNF{\formNF_q}{\mu_{\lfp}}} & \text{ since } \s_\B(q) \in \Opt_\B(\gameNF{\formNF_q}{\mu_{\lfp}}) \\
		& = \Delta(\lfp)(q) & \text{ by definition of } \Delta\\
		& = \lfp(q) & \text{ since } \lfp \text{ is a fixed point of } \Delta
	\end{align*}
	
	We can conclude that $\mathcal{P}(n+1)$ holds. It follows that $\mathcal{P}(n)$ holds for all $n \in \N$. 
	
	If we consider an arbitrary strategy $\s_\A: Q^+ \rightarrow \Dist(A)$ for Player $\A$, we have that for all $n \in \N$, $\prob{k}{\s_\A,\s_\B}(n,\top) \leq \lfp(k)$. Therefore, $\prob{k}{\nu_\A,\nu_\B}(\top) = \lim\limits_{n \rightarrow \infty} \prob{k}{\s_\A,\s_\B}(n,\top) \leq \lfp(k)$. Hence, the positional strategy $\s_\B$ for Player $\B$ ensures:
	\begin{displaymath}
	\val{\Aconc}{\s_\B}(k) \leq \lfp(k)
	\end{displaymath}
\end{proof}

The case of Player $\A$ is not symmetric as she does not have an optimal strategy in the general case, however, for all $\varepsilon > 0$, she has strategies guaranteeing the value $\lfp(q) - \varepsilon$ from all state $q \in Q$. This is stated in the lemma below. 

\begin{lemma}
	\label{lem:geq_m_delta}
	Consider a concurrent stochastic game $\langle \Aconc,\top \rangle$ with $\Aconc = \AConc$ and $\Delta: [0,1]^Q \rightarrow [0,1]^Q$ the operator on values defined in Definition~\ref{def:operator}. For all $\varepsilon > 0$, there exists a Player $\A$ strategy $\s_\A \in \SetPosStrat{\A}{\Aconc}$ 
	and $n \in \N$ such that, for all $q \in Q$ and $\s_\B \in \SetStrat{\Aconc}{\B}$: $\prob{q}{\s_\A,\s_\B}(n,\top) \geq \lfp(q) - \varepsilon$. Hence, $\val{\Aconc}{\s_\A}(q) \geq \lfp(q) - \varepsilon$ and 
	$\val{\Aconc}{\A}(q) \geq \lfp(q)$.
\end{lemma}

Before proving this lemma, let us consider a sequence of vectors in $V$. Let $v_0 \in [0,1]^Q$ be least element of $V$ with regard to the relation $\preceq$ (see Proposition~\ref{prop:property_Delta}). Then, for all $n \in \N$, we define $v_{n+1} = \Delta(v_n) \in V$ since $\Delta[V] \subseteq V$. We have the following proposition:
\begin{proposition}
	\label{prop:lim_vn}
	The sequence $(v_n)_{n \in \N}$ has a limit, with regard to the infinity norm $\dNorm$ on $[0,1]^q$, and it is equal to $\lfp$: $	v_n \underset{n \rightarrow \infty}{\rightarrow} \lfp$.
	\label{prop:limit_sequence_fix_point}
\end{proposition}
\begin{proof}
	This is given by Kleene fixed-point theorem with points 1,2 and 4 of Proposition~\ref{prop:property_Delta}.
\end{proof}

We can proceed to the proof of Lemma~\ref{lem:geq_m_delta}.
\begin{proof}
	We exhibit a sequence of Player $\A$ strategies $(\s_{n})_{n \in \N} \in (\SetStrat{\A}{\Aconc})^\N$ whose values are arbitrarily close to $\lfp$. 
	Let $\s_0: Q^+ \rightarrow \Dist(A)$ be an arbitrary Player $\A$ strategy and for all $n \in \N$, let $\s_{n+1}: Q^+ \rightarrow \Dist(A)$ be such that, for all $q \in Q$, $\s_{n+1}(q)$ is an optimal strategy for Player $\A$ in the game form $\gameNF{\formNF_q}{\mu_{v_n}}$: $\s_{n+1}(q) \in \Opt_\A(\gameNF{\formNF_q}{\mu_{v_n}}) \neq \emptyset$. Furthermore, for all $q \in Q$, we set the residual strategy of $\s_{n+1}$ to be equal to $\s_n$: $\s_{n+1}^q := \s_n$.
	
	Let us prove by induction the property $\mathcal{R}(n)$: for all states $q \in Q$ and for all strategies $\s_\B: Q^+ \rightarrow \Dist(B)$ for Player $\B$, we have $\prob{q}{\s_n,\s_\B}(n,\top) \geq v_n(q)$. The case $n = 0$ is straightforward since $v_0$ is such that $v_0(q) = 0$ if $q \neq \top$ and $v_0(\top) = 1$. Let us now assume that $\mathcal{R}(n)$ holds for some $n \in \N$. Consider a state $q \in Q \setminus \{ \top \}$ and a strategy $\s_\B: Q^+ \rightarrow \Dist(B)$ for Player $\B$. We have the following:
	\begin{align*}
		\prob{q}{\s_{n+1},\s_\B}(n+1,\top) & = \sum_{q' \in Q}
		\probTrans{q,q'}(\s_{n+1}(q),\s_\B(q))
		\cdot \prob{q'}{\s_n,\s_\B^q}(n,\top) & \text{ by Proposition}~\ref{rmq:relation_proba} \text{ and since } \s_{n+1}^q = \s_n \\
		& \geq \sum_{q' \in Q} \probTrans{q,q'}(\s_{n+1}(q),\s_\B(q))
		\cdot v_n(q') & \text{ by } \mathcal{R}(n) \\
		& = \outM_{\gameNF{\formNF_q}{\mu_{v_n}}}(\s_{n+1}(q),\s_\B(q)) & \text{ by Proposition~\ref{prop:valuation} } \\
		& \geq \va_{\gameNF{\formNF_q}{\mu_{v_n}}} & \text{ since } \s_{n+1}(q) \in \Opt_\A(\gameNF{\formNF_q}{\mu_{v_n}}) \\
		& = \Delta(v_n)(q) & \text{ by definition of } \Delta\\
		& = v_{n+1}(q) & \text{ by definition of } v_{n+1}
	\end{align*}
	
	We can conclude that $\mathcal{R}(n+1)$ holds. It follows that $\mathcal{R}(n)$ holds for all $n \in \N$. 
	
	Let $n \in \N$.	For all states $q \in Q$ and Player $\B$ strategies $\s_\B: Q^+ \rightarrow \Dist(B)$, we have $\prob{q}{\s_n,\s_\B}(n,\top) \geq v_n(q)$. 
	Now, if we consider some $\varepsilon > 0$, we have by Proposition~\ref{prop:limit_sequence_fix_point} that there exists some $N \in \N$ such that for all $n \geq N$, we have $\dNorm(v_n,\lfp) \leq \varepsilon$. In that case, the strategy $s_N$ ensures: 
	\begin{displaymath}
	\prob{q}{\s_N,\s_\B}(\top) \geq \prob{q}{\s_N,\s_\B}(N,\top) \geq v_N(q) \geq \lfp(q) - \varepsilon
	\end{displaymath}
	It follows that: 
	\begin{displaymath}
	\val{\Aconc}{\A}(q) = \sup_{\s_\A \in \SetStrat{\Aconc}{\A}} \val{\Aconc}{\s_\A}(q) \geq \lfp(q)
	\end{displaymath}
\end{proof}	

The combination of these two lemmas proves Theorem~\ref{thm:m_delta}.


\subsection{An expansion of Proposition~\ref{rmq:relation_proba}}
\label{subsec:expan_prop}
\begin{proposition}
	\label{prop:relation_proba_inf}
	Consider two strategies $\s_\A,\s_\B \in \SetStrat{\Aconc}{\A} \times \SetStrat{\Aconc}{\B}$ for Player $\A$ and $\B$ and a starting state $q \in Q$. Then, we have the following relation:
	\begin{displaymath}
	\prob{q}{\s_\A,\s_\B}(\top) = \sum_{q' \in Q}
	\probTrans{q,q'}(\s_\A(q),\s_\B(q)) \cdot \prob{q'}{\s_\A^q,\s_\B^q}(\top)
	\end{displaymath}
\end{proposition}
\begin{proof}
	This holds straightforwardly if $q = \top$ since $\top$ is self-looping. Assume now that $q \neq \top$. Let $\varepsilon > 0$. For all $q' \in Q$, let us denote $n_{q'} \in \N$ an index such that $\prob{q}{\s_\A^q,\s_\B^q}(n_{q'},\top) \geq \prob{q'}{\s_\A^q,\s_\B^q}(\top) - \varepsilon$ which exists since $\lim\limits_{n \rightarrow \infty} \prob{q'}{\s_\A^q,\s_\B^q}(n,\top) = \prob{q'}{\s_\A^q,\s_\B^q}(\top)$. Let $n := \max_{q \in Q} n_q \in \N$ since $Q$ is finite. We have, by Proposition~\ref{rmq:relation_proba}:
	\begin{align*}
		\prob{q}{\s_\A,\s_\B}(\top) & \geq \prob{q}{\s_\A,\s_\B}(n+1,\top) = \sum_{q' \in Q}	\probTrans{q,q'}(\s_\A(q),\s_\B(q)) \cdot \prob{q'}{\s_\A^q,\s_\B^q}(n,\top) \\
		& \geq \sum_{q' \in Q}	\probTrans{q,q'}(\s_\A(q),\s_\B(q)) \cdot \prob{q'}{\s_\A^q,\s_\B^q}(n_{q'},\top) \\
		& \geq \sum_{q' \in Q}	\probTrans{q,q'}(\s_\A(q),\s_\B(q)) \cdot (\prob{q'}{\s_\A^q,\s_\B^q}(\top) - \varepsilon) \\
		& = \sum_{q' \in Q}	\probTrans{q,q'}(\s_\A(q),\s_\B(q)) \cdot \prob{q'}{\s_\A^q,\s_\B^q}(\top) - \varepsilon
	\end{align*}
	As this holds for all $\varepsilon > 0$, it follows that $\prob{q}{\s_\A,\s_\B}(\top) \geq \sum_{q' \in Q}
	\probTrans{q,q'}(\s_\A(q),\s_\B(q)) \cdot \prob{q'}{\s_\A^q,\s_\B^q}(\top)$.
	
	Reciprocally, let $\varepsilon > 0$ and let $n \in \N$ be such that $\prob{q}{\s_\A,\s_\B}(n,\top) \geq \prob{q}{\s_\A,\s_\B}(\top) - \varepsilon$. Then:
	\begin{align*}
		\prob{q}{\s_\A,\s_\B}(\top) - \varepsilon & \leq \prob{q}{\s_\A,\s_\B}(n+1,\top) = \sum_{q' \in Q}	\probTrans{q,q'}(\s_\A(q),\s_\B(q)) \cdot \prob{q'}{\s_\A^q,\s_\B^q}(n,\top) \\
		& \leq \sum_{q' \in Q}	\probTrans{q,q'}(\s_\A(q),\s_\B(q)) \cdot \prob{q'}{\s_\A^q,\s_\B^q}(\top)
	\end{align*}
	As this holds for all $\varepsilon > 0$, we have $\prob{q}{\s_\A,\s_\B}(\top) \leq \sum_{q' \in Q}
	\probTrans{q,q'}(\s_\A(q),\s_\B(q)) \cdot \prob{q'}{\s_\A^q,\s_\B^q}(\top)$.
\end{proof}

\subsection{Complement on Proposition~\ref{prop:end_in_EC}}
\label{proof:prop_end_in_EC}
Once two strategies $\s_\A,\s_\B$ are fixed, we obtain a Markov chain. In this setting, BSCCs correspond to strongly connected components from which it is impossible to exit (i.e. given the two strategies $\s_\A$ and $\s_\B$). Note that any BSCC in the resulting Markov chain is (in terms of the set of states) an EC in the MDP induced by the Player $\A$ strategy $\s_\A$. Then, by Theorem 10.27 from \cite{BK08} (for instance), with probability 1, the set of states seen infinitely often in an infinite path forms a BSCC with probability 1. Therefore, from all states of the Markov chain, there is a non-zero probability to reach a BSCC. It follows that, for all states $q \in Q$ of the Markov chain, there is a finite path, from $q$, of length at most $n = |Q|$ with a non-zero probability to occur that, ends up in a BSCC.

\section{Complements on Section~\ref{sec:crucial_prop}}
\label{appen:crucial_prop}

\subsection{Proof of Proposition~\ref{prop:sufficient}}
\label{proof:lem_sufficient_cond}
\begin{proof}
	We consider a concurrent stochastic reachability game $\langle \Aconc,\top \rangle$, a valuation $v \in [0,1]^Q$ of the states such that $v \preceq \lfp$ and a Player $\A$ positional strategy $\s_\A \in \SetPosStrat{\Aconc}{\A}$ such that $\va_{\gameNF{\formNF_q}{\mu_{v}}}(\s_\A(q)) \geq v(q)$ for all $q \in Q$. Furthermore, we assume that for all end component $H = (Q_H,\beta)$ in the Markov decision process $\Gamma$ induced by the strategy $\s_\A$, if $Q_H \neq \{ \top \}$ then for all $q \in Q_H$, we have $\lfp(q) = 0$.
	
	Let us prove that, for all $q \in Q$, we have $\val{\Aconc}{\s_\A}(q) \geq v(q)$. In the Markov decision process induced by the strategy $\s_\A$, Player $\B$ plays alone a safety game. Hence, she has an optimal positional strategy $\s_\B$ (this is given by Theorem~\ref{thm:m_delta}), which ensures that, for all $q \in Q$, we have $\prob{q}{\s_\A,\s_\B}(\top) = \val{\Aconc}{\s_\A}(q)$. Let $\varepsilon > 0$. Let us prove that we have $\prob{q}{\s_\A,\s_\B}(\top) \geq v(q) - \varepsilon$. Since this would hold for all $\varepsilon > 0$, we would have $\val{\Aconc}{\s_\A}(q) = \prob{q}{\s_\A,\s_\B}(\top) \geq v(q)$.
	
	Let $v' \in [0,1]^Q$ be the valuation such that, for all $q \in Q$, we have $v'(q) = \max(0,v(q) - \varepsilon/2)$. Since $v \preceq \lfp$ 
	by Lemma~\ref{lem:geq_m_delta}, there exists $N \in \N$ and a strategy $\kappa_\A \in \SetStrat{\Aconc}{\A}$ such that for all states $q \in Q$ and Player $\B$ strategies $\s_\B' \in \SetStrat{\Aconc}{\B}$, we have $\prob{q}{\kappa_\A,\s_\B'}(N,\top) \geq v'(q)$.
	For all $l \geq 0$, we define the strategy $s_\A^l \in \SetStrat{\Aconc}{\A}$ by, for all $\pi \in Q^+$
	: 
	\[ s_\A^l(\pi) = \begin{cases} 
	\s_\A(\pi) & \text{ if } |\pi| \leq l \\
	\kappa_\A(\pi'') & \text{ otherwise, for } \pi = \pi' \cdot \pi'',\; |\pi'| = l
	\end{cases}
	\]
	Let us show, for all $l \geq 0$, the property $\mathcal{P}(l)$ holds: for all $q \in Q$ and strategy $\tilde{\s}_\B \in \SetStrat{\Aconc}{\B}$, we have $\prob{q}{\s_\A^l,\tilde{\s}_\B}(l+N,\top) \geq v'(q)$. First, note that $v - \varepsilon/2 \preceq v'$, hence $\mu_v - \varepsilon/2 \preceq \mu_{v'}$. Thus, for $q \in Q$, we have, by Observation~\ref{obs:smaller_valuation_smaller_outcome} and assumption of the lemma, that:
	\begin{displaymath}
	\va_{\gameNF{\formNF_q}{\mu_{v'}}}(\s_\A(q)) \geq \va_{\gameNF{\formNF_q}{\mu_{v}}}(\s_\A(q)) - \varepsilon/2 \geq v(q) - \varepsilon/2
	\end{displaymath}
	In addition, $0 \preceq v'$ and $0 \preceq \mu_{v'}$. Hence, $\va_{\gameNF{\formNF_q}{\mu_{v'}}}(\s_\A(q)) \geq 0$. It follows that: $$\va_{\gameNF{\formNF_q}{\mu_{v'}}}(\s_\A(q)) \geq v'(q)$$
	
	Now, by choice of the strategy $\kappa_\A$, the property $\mathcal{P}(0)$ holds. Assume now that $\mathcal{P}(l)$ holds for some $l \geq 0$. Consider a strategy $\tilde{\s}_\B \in \SetStrat{\Aconc}{\B}$ for player $\B$ and a state $q \in Q$. Note that we have $(s_\A^{l+1})^q = s_\A^l$. Furthermore:
	\begin{align*}
		\prob{q}{\s_\A^{l+1}}{\tilde{\s}_\B}(l+1+N,\top) & = \sum_{q' \in Q}
		\probTrans{q,q'}(\s_\A^{l+1}(q),\tilde{\s}_\B(q))
		\cdot \prob{q'}{\s_\A^{l}}{\tilde{\s}_\B^q}(l+N,\top) & \text{ by Remark}~\ref{rmq:relation_proba} \\
		& \geq \sum_{q' \in Q}
		\probTrans{q,q'}(\s_\A^{l+1}(q),\tilde{\s}_\B(q))
		\cdot v'(q) & \text{ by }\mathcal{P}(l)\\
		& = \outM_{\gameNF{\formNF_q}{\mu_{v'}}}(\s^{l+1}_\A(q),\tilde{\s}_\B(q)) & \text{ by Proposition~\ref{prop:valuation} } \\
		& = \outM_{\gameNF{\formNF_q}{\mu_{v'}}}(\s_\A(q),\tilde{\s}_\B(q)) & \text{ by definition of }\s_\A^{l+1} \\
		& \geq \va_{\gameNF{\formNF_q}{\mu_{v'}}}(\s_\A(q)) & \text{ by definition of }\va \\
		& \geq v'(q)
	\end{align*}
	
	Therefore, the property holds for all $l \geq 0$. 
	
	Consider now a strategy $\kappa_\B \in \SetStrat{\Aconc}{\B}$ that is optimal against $\kappa_\A$. That is, for all $q \in Q$, we have $\val{\Aconc}{\kappa_\A}(q) = \prob{q}{\kappa_\A,\kappa_\B}(\top)$. Then, for $l \geq 0$, we define a strategy $\s_\B^l$ for Player $\B$ similarly to how we define $\s_\A^l$, for all $\pi \in Q^+$: 
	\[ \s_\B^l(\pi) = \begin{cases} 
	\s_\B(\pi) & \text{ if } |\pi| \leq l \\
	\kappa_\B(\pi'') & \text{ otherwise, if } \pi = \pi' \cdot \pi'',\; |\pi'| = l
	\end{cases}
	\] 
	Now, in the MDP $\Gamma$ induced by the strategy $\s_\A$, let us denote by $\mathcal{H}_0$ the set of ECs $H = (Q_H,\beta_H)$ such that $Q_H \neq \{ \top \}$ and by $H_1$ the EC whose set of states is $\{ \top \}$. 
	
	
	In the game $\Aconc$, 
	for a state $q \in Q$, a subset of states $S \subseteq Q$ and some $k \geq 0$, we denote by $\mathsf{In}^k(q,S)$ the set of paths $\mathsf{In}^k(q,S) = \{ \pi \in Q^+ \mid \pi_0 = q,\; |\pi| \leq k,\; \pi_l \in S \}$ of length less than or equal to $k$ whose $l$-th state is in $S$. Then, we have the following partition for $k \geq l$:
	\begin{align*}
		\mathsf{Rch}^{k}(q,\top) & = \biguplus_{H \in \mathcal{H}_0} \mathsf{Rch}^{k}(q,\top) \cap \mathsf{In}^{k}(q,H) \\
		& \uplus \mathsf{Rch}^{k}(q,\top) \cap \mathsf{In}^{k}(q,H_1) \\
		& \uplus \mathsf{Rch}^{k}(q,\top) \setminus \mathsf{In}^{k}(q,\cup_{H \in \mathcal{H}} H)
	\end{align*}
	Consider some EC $H_0 \in \mathcal{H}_0$. By assumption of the lemma, we have, 
	for all $q \in H_0$:
	$$0 = \val{\Aconc}{}(q) \geq \val{\Aconc}{\kappa_\A}(q) = \prob{q}{\kappa_\A,\kappa_\B}(\top)$$
	
	Furthermore:
	$$1 = \val{\Aconc}{}(\top) = \val{\Aconc}{\kappa_\A}(\top) = \prob{\top}{\kappa_\A,\kappa_\B}(\top)$$
	
	For an EC $H \in \mathcal{H}$, a state $q \in Q$ and some $l \geq 0$, let us denote by $I^l(q,H)$ the set $\mathsf{Rch}^{N+l}(q,\top) \cap \mathsf{In}^{N+l}(q,H)$. For all EC $H \in \mathcal{H}$, by definition of the strategies $\s_\A^l$ and $\s_\B^l$, we have:
	$$\sum_{\pi \in I^l(q,H)} \prob{q}{\s_\A^l,\s_\B^l}(\pi) = \sum_{\pi \in \mathsf{In}^{N+l}(q,H)} \prob{q}{\s_\A^l,\s_\B^l}(\pi_0 \cdots \pi_l) \cdot \prob{\pi_l}{\kappa_\A,\kappa_\B}(N,\top)$$
	If $H \in \mathcal{H}_0$, for all $\pi \in I^l(q,H)$, we have $\prob{\pi_l}{\kappa_N,\kappa_\B}(\top) = 0$ since $\pi_l \in H$. Therefore:
	\begin{equation}
	\label{eqn:Il_Hzero}
	\sum_{\pi \in I^l(q,H)} \prob{q}{\s_\A^l,\s_\B^l}(\pi) = 0
	\end{equation}
	Similarly, for all $\pi \in I^l(q,H_1)$, we have $\prob{\pi_l}{\kappa_N,\kappa_\B}(\top) = 1$ since $\pi_l = \top$. Hence:
	\begin{equation}
	\label{eqn:Il_Hone}
	\sum_{\pi \in I^l(q,H)} \prob{q}{\s_\A^l,\s_\B^l}(\pi) = \sum_{\pi \in \mathsf{In}^{N+l}(q,H)} \prob{q}{\s_\A^l,\s_\B^l}(\pi_0 \cdots \pi_l) = \sum_{\pi \in \mathsf{Rch}^{l}(q,H)} \prob{q}{\s_\A^l,\s^l_\B}(\pi) = \sum_{\pi \in \mathsf{Rch}^{l}(q,H)} \prob{q}{\s_\A,\s_\B}(\pi)
	\end{equation}		
	Finally, for a state $q \in Q$, let $J^l(q) = \mathsf{Rch}^{N+l}(q,\top) \setminus \mathsf{In}^{N+l}(q,\cup_{H \in \mathcal{H}} H)$. Recall the set of BSCCs $\mathcal{H}_{\s_\B} \subseteq \mathcal{H}$ that is impossible to exit (from Proposition~\ref{prop:end_in_EC}). We have:
	\begin{align*}
		\sum_{\pi \in J^l(q)} \prob{q}{\s_\A^l,\s_\B^l}(\pi) & \leq 1 - \sum_{\pi \in \mathsf{In}^{N+l}(q,\cup_{H \in \mathcal{H}} H)} \prob{q}{\s_\A^l,\s_\B^l}(\pi) \leq 1 - \sum_{\pi \in \mathsf{In}^{N+l}(q,\cup_{H \in \mathcal{H}_{\s_\B}} H)} \prob{q}{\s_\A^l,\s_\B^l}(\pi) \\ & = 1 - \sum_{\pi \in \mathsf{Rch}^{l}(q,\cup_{H \in \mathcal{H}_{\s_\B}} H)} \prob{q}{\s_\A,\s_\B}(\pi)
	\end{align*}
	
	Let $U = \cup_{H \in \mathcal{H}_{\s_\B}} H$ and $n = |Q|$. For all $q \in Q \setminus U$, we have $p_q = \prob{q}{\s_\A,\s_\B}(n,U) > 0$. Let $p = \min_{q \in Q} p_q > 0$ (since $Q$ is finite) the minimum of such probabilities. It follows that, for all $q \in Q \setminus U$, we have:
	$$\sum_{\pi \notin \mathsf{Rch}^{n}(q,U) \wedge |\pi| = n} \prob{q}{\s_\A,\s_\B}(\pi) \leq (1 - p)$$
	Then:
	\begin{align*}
		\sum_{\pi \notin \mathsf{Rch}^{2n}(q,U) \wedge |\pi| = 2n} \prob{q}{\s_\A,\s_\B}(\pi) & = \sum_{\pi \notin \mathsf{Rch}^{n}(q,U) \wedge |\pi| = n} \prob{q}{\s_\A,\s_\B}(\pi) \cdot \left( \sum_{\pi' \notin \mathsf{Rch}^{n}(\head(\pi),U) \wedge |\pi'| = n} \prob{\pi_n}{\s_\A,\s_\B}(\pi') \right) \\
		& \leq \sum_{\pi \notin \mathsf{Rch}^{n}(q,U) \wedge |\pi| = n} \prob{q}{\s_\A,\s_\B}(\pi) \cdot (1 - p) \\
		& \leq (1 -p)^2
	\end{align*}
	In fact, for all $k \geq 0$, we have:
	$$\sum_{\pi \notin \mathsf{Rch}^{k \cdot n}(q,U) \wedge |\pi| = k \cdot n} \prob{q}{\s_\A,\s_\B}(\pi) \leq (1 - p)^k$$
	Let $l \geq 0$ be such that $(1 - p)^{l/n} \leq \varepsilon/2$ (which exists since $1 - p < 1$). In that case, for all $q \in Q$, we have: 
	\begin{equation}
	\label{eqn:Il_Hother}
	\sum_{\pi \in J^l(q)} \prob{q}{\s_\A^l,\s_\B^l}(\pi) \leq 1 - (1 - \sum_{\pi \notin \mathsf{Rch}^{l}(q,U) \wedge |\pi| = l} \prob{q}{\s_\A,\s_\B}(\pi) ) \leq (1 - p)^{l/n} \leq \varepsilon/2
	\end{equation}
	
	Finally, let $q \in Q$. We have the following, by definition of $I^l(q,H)$ and $J^l(q)$:
	\begin{equation}
	\label{eqn:partition_rch}
	\mathsf{Rch}^{N+l}(q,\top) = \biguplus_{H \in \mathcal{H}_0} I^l(q,H) \uplus I^l(q,H_1) \uplus J^l(q)
	\end{equation}
	
	In addition, we have:
	\begin{align*}
		\sum_{\pi \in \mathsf{Rch}^{N+l}(q,\top)} \prob{q}{\s^l_\A,\s^l_\B}(\pi) & = \sum_{H \in \mathcal{H}_0} \sum_{\pi \in I^l(q,H)} \prob{q}{\s^l_\A,\s^l_\B}(\pi) + \sum_{\pi \in I^l(q,H_1)} \prob{q}{\s^l_\A,\s^l_\B}(\pi) + \sum_{\pi \in J^l(q)} \prob{q}{\s^l_\A,\s^l_\B}(\pi) & \text{ by }(\ref{eqn:partition_rch})\\
		& = \sum_{\pi \in I^l(q,H_1)} \prob{q}{\s^l_\A,\s^l_\B}(\pi) + \sum_{\pi \in J^l(q)} \prob{q}{\s^l_\A,\s^l_\B}(\pi) & \text{ by }(\ref{eqn:Il_Hzero}) \\
		& = \sum_{\pi \in \mathsf{Rch}^{l}(q,\top)} \prob{q}{\s_\A,\s_\B}(\pi) + \sum_{\pi \in J^l(q)} \prob{q}{\s^l_\A,\s^l_\B}(\pi) & \text{ by }(\ref{eqn:Il_Hone}) \\
		& \leq \sum_{\pi \in \mathsf{Rch}^{l}(q,\top)} \prob{q}{\s_\A,\s_\B}(\pi) + \varepsilon/2 &\text{ by }(\ref{eqn:Il_Hother}) \\
		& \leq \sum_{\pi \in \mathsf{Rch}(q,\top)} \prob{q}{\s_\A,\s_\B}(\pi) + \varepsilon/2 \\
		& = \prob{q}{\s_\A,\s_\B}(\top) + \varepsilon/2
	\end{align*}
	Furthermore, by $\mathcal{P}(l)$, we have:
	$$\prob{q}{\s_\A^l,\tilde{\s}_\B}(N+l,\top) = \sum_{\pi \in \mathsf{Rch}^{N+l}(q,\top)} \prob{q}{\s^l_\A,\s^l_\B}(\pi) \geq v'(q) \geq v(q) - \varepsilon/2$$
	
	Overall, we have:
	$$v(q) - \varepsilon/2 \leq \sum_{\pi \in \mathsf{Rch}^{N+l}(q,\top)} \prob{q}{\s^l_\A,\s^l_\B}(\pi) \leq \prob{q}{\s_\A,\s_\B}(\top) + \varepsilon/2$$
	
	Finally, by choice of the strategy $\s_\B$ for Player $\B$, we have $\prob{q}{\s_\A,\s_\B}(\top) = \val{\Aconc}{\s_\A}(q)$. That is:
	$$v(q) - \varepsilon \leq \val{\Aconc}{\s_\A}(q)$$
	
	Since this holds for all $\varepsilon > 0$, this shows that:
	$$\val{\Aconc}{\s_\A}(q) \leq v(q)$$	
\end{proof}

\subsection{Proof of Proposition~\ref{prop:value_end_component}}
\label{proof:value_end_component}
\begin{proof}
	Consider a Player $\A$ positional strategy $\s_\A \in \SetPosStrat{\Aconc}{\A}$ locally dominating a valuation $v \in [0,1]^Q$ and an end component $H = (Q_H,\beta)$ in the MDP induced by the strategy $\s_\A$. Let $x \in Q_H$ be such that $v(x) = \max_{q \in Q_H} v(q)$ (which exists since $Q_H$ is finite). We set $v_H := v(x)$. Then, for any $b \in \beta(x)$, we have:
	\begin{displaymath}
	v_H = v(x) \leq \va_{\gameNF{\formNF_x}{\mu_v}}(\s_\A(x)) \leq \outM_{\gameNF{\formNF_x}{\mu_v}}(\s_\A(x),b)
	\end{displaymath}
	In addition, by Proposition~\ref{prop:valuation}, we have:
	\begin{displaymath}
	\outM_{\gameNF{\formNF_x}{\mu_v}}(\s_\A(x),b) = \sum_{q \in Q} \probTrans{q,q'}(\s_\A(x),b) \cdot v(q) = \sum_{q \in Q} \iota(x,b)(q)
	\cdot v(q)\\
	\end{displaymath}
	Since $H$ is an end component and $b \in \beta(x)$, we have $\Supp(\iota(x,b)) \subseteq Q_H$. Furthermore, $v_H$ is the maximum of $v$ over states in $Q_H$, which implies:
	\begin{displaymath}
	\sum_{q \in Q} \iota(x,b)(q) \cdot v(q) =
	\sum_{q \in Q_H} \iota(x,b)(q) \cdot v(q) \leq \sum_{q \in Q_H} \iota(x,b)(q) \cdot v_H 
	= v_H
	\end{displaymath}
	Overall, we get:
	\begin{displaymath}
	v_H = v(x) \leq \va_{\gameNF{\formNF_x}{\mu_v}}(\s_\A(x)) \leq \sum_{q \in Q_H} \iota(x,b)(q) \cdot v(q) \leq \sum_{q \in Q_H} \iota(x,b)(q) \cdot v_H = v_H
	\end{displaymath}
	Hence, all the above inequalities are in fact equalities, which means that we have $v(x) = \va_{\gameNF{\formNF_x}{\mu_v}}(\s_\A(x))$ and, for all $q \in \Supp(\iota(x,b))$, $v(q) = v_H$. This holds for all $x \in Q_H$ such that $v(x) = v_H$ and for all $b \in \beta(x)$.
	
	Consider now the underlying graph $G_H = (Q_H,E)$ of the end component $H$ such that for $q,q' \in Q_H$, we have $(q,q') \in E$ if and only if $q' \in \Supp(\iota(q,\beta(q)))$. What we have proven is that all successors $q$ of $x$ in the graph $G_H$ are such that $v(q) = v_H$. By propagating the property, we have that all states $q$ reachable from $x$ in the graph $G_H$ are such that $v(q) = v_H$. As the graph $G_H$ is strongly connected (since $H$ is an end component), this implies that, for all $q \in Q_H$, we have $v_H = v(q) = \va_{\gameNF{\formNF_x}{\mu_v}}(\s_\A(q))$.
\end{proof}

\section{Complements on Section~\ref{sec:optimal_strat}}
\label{appen:optimal_strat}

\subsection{Proof of Lemma~\ref{lem:bad_increase}}
\label{proof:bad_increase}
Before proving this lemma, we introduce the following notation: for a finite path $\pi = \pi_1 \cdots \pi_n \in Q^+$ and for $1 \leq i \leq n$, we denote by $\pi_{\leq i}$ the finite paths $\pi_{\leq i} := \pi_1 \cdots \pi_i \in Q^+$. Now, let us state and prove a necessary condition for a Player $\A$ strategy to be optimal. First, we define the notion, given a Player $\A$ strategy $\sigma_\A$, of relevant paths w.r.t. the strategy $\s_\A$ which informally are paths with non-zero probability to occur if Player $\B$ locally plays optimal actions against the strategy $\s_\A$. That is:
\begin{definition}
	Consider 
	a Player $\A$ strategy $\s_\A \in \SetStrat{\Aconc}{\A}$ and a state $q$. The set $\mathsf{RP}_{\s_\A}(q)$ of \emph{relevant paths} w.r.t. the strategy $\s_\A$ from $q$ is equal to: 
	\begin{displaymath}
		\mathsf{RP}_{\s_\A}(q) := \{ \pi = \pi_1 \cdots \pi_n \in Q^+ \mid \pi_1 = q,\; \forall 1 \leq i \leq n-1,\; \exists b \in B_{\s_\A(\pi_{\leq i})},\; \probTrans{\pi_i,\pi_{i+1}}(\s_\A(\pi_{\leq i}),b) > 0 \}
	\end{displaymath}
\end{definition}
Then, to be optimal, a Player $\A$ strategy needs, on all relevant paths w.r.t. $\s_\A$ to plays optimally both locally (in the local interaction) and globally (in the concurrent game), that is:
\begin{proposition}
	\label{prop:neccessary_condition_optimal_stratey}
	Consider a Player $\A$ strategy $\s_\A \in \SetStrat{\Aconc}{\A}$ and assume that it is optimal from a state $q$, i.e. $\val{\Aconc}{\s_\A}(q) = \val{\Aconc}{}(q) = \lfp(q)$. (Recall that the value of the states is given by the vector $\lfp$). Then, for any compatible path $\pi = \pi' \cdot q' \in \mathsf{RP}_{\s_\A}(q)$, we have:
	\begin{itemize}
		\item the (local) strategy $\s_\A(\pi)$ is optimal in the game form $\formNF_q$ w.r.t. the valuation $\mu_\lfp$: $\s_\A(\pi) \in \Opt_\A(\gameNF{\formNF_q}{\mu_\lfp})$;
		\item for all $b \in B_{\s_\A(\pi)}$, for all $q'' \in Q$ such that $\probTrans{q',q''}(\s_\A(\pi),b) > 0$, the residual strategy $\s_\A^{\pi}$ is optimal from $q''$: $\val{\Aconc}{\s_\A^{\pi}}(q'') = \val{\Aconc}{}(q'')$.
	\end{itemize}
\end{proposition}
\begin{proof}
	Let us prove by induction on $n \in \N^*$ the property $\mathcal{P}(n)$ stating this proposition for all paths of length at most $n$. Let us show $\mathcal{P}(1)$. 
	
	Consider an optimal strategy $\s_\A$ from a state $q$. The only relevant path w.r.t. the strategy $\s_\A$ of length $1$ is $\pi = q \in \mathsf{RP}_{\s_\A}(q)$. Consider some $b \in B_{\s_\A(q)}$. Now, let $\s'_\B \in \SetStrat{\Aconc}{\B}$ be a Player $\B$ strategy that is optimal against the residual strategy $\s_\A^q$ of Player $\A$: for all $q' \in Q$, we have $\prob{q'}{\s_\A^q,\s_\B'}(\top) = \val{\Aconc}{\s_\A^q}(q')$. Let us now define the strategy $\s_\B \in \SetStrat{\Aconc}{\B}$ by $\s_\B(q) := b$ and $\s_\B^q := \s_\B'$. Now, since the strategy $\s_\A$ is optimal from $q$, we have: 
	\begin{displaymath}
		\prob{q}{\s_\A,\s_\B}(\top) \geq \val{\Aconc}{\s_\A}(q) = \val{\Aconc}{}(q) = \lfp(q)
	\end{displaymath}
	Furthermore, by Proposition~\ref{prop:relation_proba_inf}, we have:
	\begin{align*}
		\lfp(q) \leq \prob{q}{\s_\A,\s_\B}(\top) & = \sum_{q' \in Q}
		\probTrans{q,q'}(\s_\A(q),b) \cdot \prob{q'}{\s_\A^q,\s_\B^q}(\top) & \text{ by Proposition~\ref{prop:relation_proba_inf} }\\
		& = \sum_{q' \in Q}
		\probTrans{q,q'}(\s_\A(q),b) \cdot \prob{q'}{\s_\A^q,\s_\B'}(\top) & \text{ since }\s_\B^q=\s_\B' \\
		& = \sum_{q' \in Q}
		\probTrans{q,q'}(\s_\A(q),b) \cdot \val{\Aconc}{\s_\A^q}(q') & \text{ by definition of }\s_\B' \\
		& \leq \sum_{q' \in Q}
		\probTrans{q,q'}(\s_\A(q),b) \cdot \lfp(q) & \text{ since }\val{\Aconc}{\s_\A}(q) \leq \val{\Aconc}{}(q) = \lfp(q)\\
		& = \outM_{\gameNF{\formNF_q}{\mu_\lfp}}(\s_\A(q),b) & \text{ by Proposition~\ref{prop:valuation} } \\
		& = \va_{\gameNF{\formNF_q}{\mu_\lfp}}(\s_\A(q)) & \text{ since }\s_\B(q) \in B_{\s_\A(q)} \\
		& \leq \va_{\gameNF{\formNF_q}{\mu_\lfp}}  = \lfp(q) 
	\end{align*}
	All the above inequalities are in in fact equalities. In particular, we have $\va_{\gameNF{\formNF_q}{\mu_\lfp}}(\s_\A(q)) = \va_{\gameNF{\formNF_q}{\mu_\lfp}}$, that is, $\s_\A(q) \in \Opt_\A(\gameNF{\formNF_q}{\mu_\lfp})$. Furthermore, for all $q' \in Q$ such that $\probTrans{q,q'}(\s_\A(q),b) > 0$, we have $\val{\Aconc}{\s_\A^q}(q') = \val{\Aconc}{}(q')$. Since this holds for all $b \in B_{\s_\A(q)}$, this proves $\mathcal{P}(1)$. 
	
	Let us now assume that $\mathcal{P}(n)$ holds for some $n \in \N^*$ and consider a relevant path $\pi = \pi_1 \cdots \pi_{n+1} \in \mathsf{RP}_{\s_\A}(q)$ of length $n+1$. In particular, $\pi_1 = q$ is a relevant path of length $1 \leq n$. Furthermore, there exists an optimal action $b \in B_{\s_\A(\pi_1)}$ such that $\probTrans{\pi_1,\pi_{2}}(\s_\A(\pi_{\leq 1}),b) > 0$. Hence, by $\mathcal{P}(n)$, the residual strategy $\s_\A^q$ is optimal from $\pi_2$. Then, the path $\pi' = \pi_2 \cdots \pi_{n+1}$ is of length $n$ and is a relevant path from $\pi_2$ w.r.t. the strategy $\s_\A^q$: $\pi' \in \mathsf{RP}_{\s_\A^q}(\pi_2)$. Hence, by $\mathcal{P}(n)$: 
	\begin{itemize}
		\item the (local) strategy $\s_\A^q(\pi') = \s_\A(q \cdot \pi') = \s_\A(\pi)$ is optimal in the game form $\formNF_{\pi_{n+1}}$ w.r.t. the valuation $\mu_\lfp$: $\s_\A(\pi) \in \Opt_\A(\gameNF{\formNF_{\pi_{n+1}}}{\mu_\lfp})$;
		\item for all $b \in B_{\s_\A^q(\pi')} = B_{\s_\A(\pi)}$, for all $q'' \in Q$ such that $\probTrans{q',q''}(\s_\A^q(\pi'),b) = \probTrans{q',q''}(\s_\A(\pi),b) > 0$, the residual strategy $(\s_\A^q)^{\pi'} = \s_\A^{q \cdot \pi'} = \s_\A^\pi$ is optimal from $q''$: $\val{\Aconc}{\s_\A^{\pi}}(q'') = \val{\Aconc}{}(q'')$.
	\end{itemize}
	Hence, $\mathcal{P}(n+1)$. In fact, $\mathcal{P}(i)$ holds for all $i \in \N^*$, which proves Proposition~\ref{prop:neccessary_condition_optimal_stratey}.
\end{proof}

We can now proceed to the proof of Lemma~\ref{lem:bad_increase}.
\begin{proof}
	We consider the concurrent reachability game $\langle \Aconc,\top \rangle$ and assume that the set of states $\Bex$ is sub-maximizable, i.e. such that $\Bex \subseteq \SubOS_\A$. Now, consider a state $q \in Q \setminus \Scr(\Bex)$. Let us assume towards a contradiction that there exists a Player $\A$ strategy $\s_\A \in \SetStrat{\Aconc}{\A}$ that is optimal from $q$: $\val{\Aconc}{\s_\A}(q) = \val{\Aconc}{}(q)$. We exhibit a sequence of Player $\B$ strategies $(\s_i)_{i \in \N} \in (\SetStrat{\Aconc}{\B})^\N$ such that, for all $i \in \N$, the strategy $\s_i$ ensures the following property $\mathcal{P}(i)$:
	\begin{itemize}
		\item for all $\pi \in Q^+$ of length at most $i-1$, we have $\s_i(\pi) = \s_{i-1}(\pi)$ (irrelevant for $i = 0,1$);
		\item for all $\pi \in Q^+$ compatible with the strategies $\s_\A,\s_i$ from $q$ (i.e. such that $\prob{q}{\s_\A,\s_i}(\pi) > 0$) and of length at most $i+1$:
		\begin{itemize}
			\item $\pi$ is a relevant path from $q$ w.r.t. the strategy $\s_\A$: $\pi \in \mathsf{RP}_{\s_\A}(q)$;
			\item $\pi \in ((Q \setminus \Scr(\Bex)) \cup \lfp^{-1}[0])^+$.
		\end{itemize}
	\end{itemize}
	An arbitrary strategy $\s_0 \in \SetStrat{\Aconc}{\B}$ ensures the property $\mathcal{P}(0)$ since $q \in Q \setminus \Scr(\Bex)$. Assume now the property $\mathcal{P}(n)$ for some $n \in \N$ is ensured by the strategy $\s_n$. We want to define $\s_{n+1}$. For all path $\pi \in Q^+$ of length at most $n$, we set $\s_{n+1}(\pi) := \s_n(\pi)$ thus ensuring the first item. Note that this also ensures that, for all path $\pi$ of length at most $n+1$, we have $\prob{q}{\s_\A,\s_{n+1}}(\pi) = \prob{q}{\s_\A,\s_n}(\pi)$. Then, consider some path $\pi = \pi_1 \cdots \pi_{n+1} \in Q^+$ of length $n+1$. If $\pi$ is not compatible with the strategies $\s_\A,\s_n$, we define $\s_{n+1}(\pi)$ arbitrarily. Otherwise,  by $\mathcal{P}(n)$, $\pi$ is a relevant path w.r.t. the strategy $\s_\A$ from $q$. Hence, by Proposition~\ref{prop:neccessary_condition_optimal_stratey}, since the strategy $\s_\A$ is assumed optimal from $q$, we have $\s_\A(\pi) \in \Opt_\A(\gameNF{\formNF_{\pi_{n+1}}}{\mu_\lfp})$ and $\s_\A(\pi) \notin \Rsk_{\pi_{n+1}}(\Bex)$ since all states $\Bex$ are sub-maximizable. In addition, also by $\mathcal{P}(n)$, we have: $\pi_{n+1} \in (Q \setminus \Scr(\Bex)) \cup \lfp^{-1}[0]$, thus there are two possibilities:
	\begin{itemize}
		\item assume that $\pi_{n+1} \in Q \setminus \Scr(\Bex)$. It follows that $\Eff_{\pi_{n+1}}(\Scr(\Bex),\Bex) = \emptyset$. Hence, $\s_\A(\pi) \in \Opt_\A(\gameNF{\formNF_{\pi_{n+1}}}{\mu_\lfp})$ and $\s_\A(\pi) \notin \Prg_{\pi_{n+1}}(\Scr(\Bex))$. Therefore, by definition of $\Prg_{\pi_{n+1}}$, there exists an optimal action $b \in B_{\s_\A(\pi)}$ such that $\delta(q,\Supp(\s_\A(\pi)),b) \cap (\Scr(\Bex))_\distribSet = \emptyset$. We set $\s_{n+1}(\pi) := b$. It follows that all states $q' \in Q$ such that $\probTrans{\pi_{n+1},q'}(\s_\A(\pi),\s_{n+1}(\pi)) > 0$ ensure $q \in Q \setminus \Scr(\Bex)$.
		\item assume now $\pi_{n+1} \in \lfp^{-1}[0]$. Let $b \in B$ be such that its value, w.r.t. the valuation $\mu_\lfp$, in the game form $\formNF_{\pi_{n+1}}$ is 0: $\va_{\gameNF{\formNF_{\pi_{n+1}}}{\mu_\lfp}}(b) = 0$. We set $\s_{n+1}(\pi) := b$. In particular, we have $\s_{n+1}(\pi) \in B_{\s_\A(\pi)}$. Furthermore, for all states $q' \in Q$ such that $\probTrans{\pi_{n+1},q'}(\s_\A(\pi),\s_{n+1}(\pi)) > 0$, we have $q \in \lfp^{-1}[0]$.
	\end{itemize}
	The strategy $\s_\B$ is defined arbitrarily on all other paths. With these choices, the property $\mathcal{P}(n+1)$ is ensured by the strategy $\s_{n+1}$. (Note that the second item holds since, on all compatible paths, the strategy $\s_\B$ plays an optimal action w.r.t. the strategy $\s_\A$). Therefore, $\mathcal{P}(i)$ is ensured by the strategy $\s_i$ for all $i \in \N$. We can then consider the limit-strategy $\s_\B \in \SetStrat{\Aconc}{\B}$ such that, for all $\pi  = \pi_1 \cdots \pi_n \in Q^+$ of length $n$ for some $n \in \N$, we have $\s_\B(\pi) = \s_n(\pi)$. Note that, the first item of the property $\mathcal{P}$ ensures that $\s_\B(\pi_{\leq i}) = \s_{n}(\pi_{\leq i})$ for all $i \leq n$. Then, any finite path $\pi \in Q^+$ compatible with the strategies $\s_\A,\s_\B$ is such that $\pi \in ((Q \setminus \Scr(\Bex)) \cup \lfp^{-1}[0])^+$. 
	
	Hence, for all $k \in \N$, for all paths $\pi \in \mathsf{Rch}^k(q,\top)$, we have $\prob{q}{\s_\A,\s_\B}(\pi) = 0$. It follows that: $\prob{q}{\s_\A,\s_\B}(k,\top) = \sum_{\pi \in \mathsf{Rch}^k(q,\top)} \prob{q}{\s_\A,\s_\B}(\pi) = 0$. Thus, we have $\prob{q}{\s_\A,\s_\B}(\top) = \lim\limits_{k \rightarrow \infty} \prob{q}{\s_\A,\s_\B}(k,\top) = 0$. Therefore, $\val{\Aconc}{\s_\A}(q) = 0$ and	the strategy $\s_\A$ is not optimal from $q$ since $\val{\Aconc}{}(q) > 0$ (as $\lfp^{-1}[0] \subseteq \Scr(\Bex)$).

		.
\end{proof}

\subsection{Proof of Lemma~\ref{lem:pos_optimal_strat}}
\label{proof:pos_optimal_strat}
Before proving this lemma, let us consider the following proposition ensuring, for all $\varepsilon > 0$ the existence of a valuation $\varepsilon$-close to $\lfp$ that strictly increases -- w.r.t. $\Delta$ -- on a given set. This will be used to specify the valuation $v$ we consider on the set of bad states. Specifically:
\begin{proposition}
	\label{prop:increasing_valuation}
	Consider a concurrent reachability game $\langle \Aconc,\top \rangle$ with its values given by the least fixed point $\lfp \in [0,1]^Q$ of the operator $\Delta: [0,1]^Q \rightarrow [0,1]^Q$, a set of states $G \subseteq Q$ such that $(Q \setminus G) \cap \lfp^{-1}[0] = \emptyset$ and $\varepsilon > 0$. There exists a valuation $v \in [0,1]^Q$ such that $v \preceq \lfp$, $\norm{\lfp}{v} \leq \varepsilon$ (the infinity norm on $[0,1]^Q$), $\restriction{v}{G} = \restriction{\lfp}{G}$ and for all $q \in Q \setminus G$: 
	$\Delta(v)(q) = \va_{\gameNF{\formNF_q}{\mu_{v}}}
	> v(q)$.
\end{proposition}

The proof of this proposition is rather technical and relies on the analytical properties of the function $\Delta$
, hence we proceed to the proof of Lemma~\ref{lem:pos_optimal_strat} while admitting (for now) Proposition~\ref{prop:increasing_valuation}.
\begin{proof}
	Let $\varepsilon_i > 0$. We want to exhibit a Player $\A$ positional strategy $\s_\A \in \SetPosStrat{\Aconc}{\A}$ that is optimal from all states $q \in \Scr(\Bad)$: $\val{\Aconc}{\s_\A}(q) = \val{\Aconc}{}(q) = \lfp(q)$. Consider some $q \in \Scr(\Bad)$. If $\lfp(q) = 0$ or $q = \top$, then any Player $\A$ strategy is optimal from $q$. Now, we assume that $\lfp(q) > 0$ and $q \neq \top$. In that case, we have $q \in \cup_{n \in \N} \Scr_n(\Bad)$. Let $i \in \N$ be the smallest index such that $q \in \Scr_{i}(\Bad)$. We have $i \geq 1$ since $q \neq \top$ and $\Scr_{0}(\Bad) = \{ \top \}$. Therefore, by definition of $\Scr_{i}(\Bad)$, there is an efficient strategy $\sigma_q \in \Eff_q(\Scr_{i-1}(\Bad),\Bad)$ at state $q$. We set $\s_\A(q) := \sigma_q$. 
	
	Now, we want to define the valuation $v \in [0,1]^Q$ of the states that the strategy $\s_\A$ will guarantee on all states (in $\Scr(\Bad)$ and $\Bad$). Since we want the strategy $\s_\A$ to be optimal from all states in $\Scr(\Bad)$, we will consider a valuation $v$ such that $\restriction{v}{\Scr(\Bad)} := \restriction{\lfp}{\Scr(\Bad)}$. However, any state in $\Bad$ is sub-maximizable, therefore the strategy $\s_\A$ cannot guarantee the valuation $\lfp$ from states in $\Bad$. However, it can guarantee a valuation that is $\varepsilon$-close to $\lfp$ for a well chosen $\varepsilon$. Specifically, consider a state $q \in \Scr(\Bad)$ and consider the smallest difference between the value of a non-optimal action (in $\formNF_q$) and the value of the local strategy $\s_\A(q)$. That is, we let:
	\begin{displaymath}
		\eta_q := \min_{b \in B \setminus B_{\s_\A(q)}} \outM_{\gameNF{\formNF_q}{\mu_\lfp}}(\s_\A(q),b) - \va_{\gameNF{\formNF_q}{\mu_{\lfp}}}(\s_\A(q)) > 0
	\end{displaymath}
	If the set $B \setminus B_{\s_\A(q)}$ is empty, then we set $\eta_q = 1$. Furthermore, $\eta_q > 0$ by definition of the set of optimal actions $B_{\s_\A(q)}$. We consider this construction on all states $q \in \Scr(\Bad)$ and we let $\eta := \min_{q \in \Scr(\Bad)} \eta_q > 0$ since there a finite number of states. Let $v$ be a valuation such as in Proposition~\ref{prop:increasing_valuation} for $G := \Scr(\Bad)$ (note that $(Q \setminus G) \cap \lfp^{-1}[0] = \emptyset$) and $\varepsilon = \min(\eta,\varepsilon_i) > 0$.
	
	For any state $q \in \Bad$
	, by definition of the valuation $v$, we have $\Delta(v)(q) = \va_{\gameNF{\formNF_q}{\mu_v}} > v(q)$. We set $\s_\A(q)$ such that $\va_{\gameNF{\formNF_q}{\mu_v}}(\s_\A(q)) > v(q)$. The strategy $\s_\A$ is now completely defined as it is positional and defined on all states in $Q$. Let us show that the strategy $\s_\A$ locally dominates the valuation $v$. Straightforwardly, for all states $q \in \Bad$, we have $\outM_{\gameNF{\formNF_q}{\mu_v}}(\s_\A(q),b) \geq v(q)$. Consider now some state $q \in \Scr(\Bad)$. Recall that the set $\Bad_\distribSet$ refers to the set of Nature states whose support intersect $\Bad$: $\Bad_\distribSet = \{ d \in \distribSet \mid \Supp(d) \cap \Bad \neq \emptyset \}$. In particular, for all Nature states $d \in \distribSet \setminus \Bad_\distribSet$ that is not in that set, we have $\Supp(d) \subseteq Q \setminus \Bad = \Scr(\Bad)$. Hence, since $\restriction{v}{\Scr(\Bad)} = \restriction{\lfp}{\Scr(\Bad)}$:
	\begin{displaymath}
		\mu_v(d) = \sum_{q \in Q} \distribFunc(q)(q) \cdot v(q) = \sum_{q \in \Supp(d)} \distribFunc(q)(q) \cdot \lfp(q) = \mu_\lfp(q)
	\end{displaymath}
	Furthermore, the value $\va_{\gameNF{\formNF_q}{\mu_v}}(\s_\A(q))$ of the local strategy $\s_\A(q)$ is equal to:
	\begin{displaymath}
		\va_{\gameNF{\formNF_q}{\mu_v}}(\s_\A(q)) = \min_{b \in B} \outM_{\gameNF{\formNF_q}{\mu_\lfp}}(\s_\A(q),b)
	\end{displaymath}
	Hence, we consider some $b \in B$. There are two possibilities:
	\begin{itemize}
		\item Assume that $b \in B_{\s_\A(q)}$. We have $\s_\A(q) \in \Eff_q(\Scr_{i-1}(\Bad),\Bad)$, therefore $\s_\A(q) \notin \Rsk_q(\Bad)$. Hence, by definition of $\Rsk$, we have $\delta(q,\Supp(\s_\A(q)),b) \cap \Bad_\distribSet = \emptyset$. That is, for all actions $a \in \Supp(\s_\A(q))$, we have $\delta(q,a,b) \notin \Bad_\distribSet$ and $\mu_v(\delta(q,a,b)) = \mu_\lfp(\delta(q,a,b))$. Therefore:
		\begin{align*}
			\outM_{\gameNF{\formNF_q}{\mu_v}}(\s_\A(q),b) & = \sum_{a \in A} \s_\A(q)(a) \cdot \mu_v \circ \delta(q,a,b)\\
			& = \sum_{a \in \Supp(\s_\A(q))} \s_\A(q)(a) \cdot \mu_v \circ \delta(q,a,b) \\
			& = \sum_{a \in \Supp(\s_\A(q))} \s_\A(q)(a) \cdot \mu_\lfp \circ \delta(q,a,b) \\
			& = \outM_{\gameNF{\formNF_q}{\mu_\lfp}}(\s_\A(q),b) =  \va_{\gameNF{\formNF_q}{\mu_{\lfp}}}(\s_\A(q)) = \lfp(q) = v(q)
		\end{align*}
		\item Assume now that $b \in B \setminus B_{\s_\A(q)}$. Since $\eta_q \geq \eta$ and by choice of $v$, we have $\lfp - \eta_q \preceq v$ and $\mu_\lfp - \eta_q \preceq \mu_v$. Thus, by Observation~\ref{obs:smaller_valuation_smaller_outcome}:
		\begin{displaymath}
			\outM_{\gameNF{\formNF_q}{\mu_v}}(\s_\A(q),b) \geq \outM_{\gameNF{\formNF_q}{\mu_\lfp}}(\s_\A(q),b) - \eta_q \geq \va_{\gameNF{\formNF_q}{\mu_{\lfp}}}(\s_\A(q)) = \lfp(q) = v(q)
		\end{displaymath}
	\end{itemize}
	Overall, we have $\outM_{\gameNF{\formNF_q}{\mu_v}}(\s_\A(q),b) \geq v(q)$. As this holds for all $b \in B$, we have $\va_{\gameNF{\formNF_q}{\mu_v}}(\s_\A(q)) \geq v(q)$. This holds for all states $q \in \Scr(\Bad)$. We obtain that the strategy $\s_\A$ locally dominates the valuation $v$.
	
	Let us now apply Proposition~\ref{prop:sufficient} to show that the strategy $\s_\A$ guarantees the valuation $v$. Consider an EC $H = (Q_H,\beta)$ in the MDP induced by the strategy $\s_\A$ such that $Q_h \neq \{ \top \}$. For all $q \in \Bad$, by definition of the strategy $\s_\A$, we have $\va_{\gameNF{\formNF_q}{\mu_v}}(\s_\A(q)) > v(q)$. Hence, by Proposition~\ref{prop:value_end_component}, we have $Q_H \cap \Bad = \emptyset$. Now, assume towards a contradiction that $Q_H \cap (\cup_{n \in \N} \Scr_n(\Bad)) \neq \emptyset$. Let $i \in \N$ be the smallest index such that $Q_H \cap \Scr_i(\Bad) \neq \emptyset$. Note that $i \geq 1$ since $\Scr_0(\Bad) = \{ \top \}$ and $Q_H \neq \{ \top \}$ by assumption. Recall that the local strategy $\s_\A(q)$ is chosen efficient, i.e.: $\s_\A(q) \in \Eff_q(\Scr_{i-1}(\Bad),\Bad)$. In particular, we have $\s_\A(q) \in \Prg_q(\Scr_{i-1}(\Bad))$. That is, for all $b \in B_{\s_\A(q)}$, we have $\delta(q,\Supp(\sigma_\A),b) \cap (\Scr_{i-1}(\Bad))_\distribSet \neq \emptyset$. Now, let $b \in \beta(q) \neq \emptyset$. 
	\begin{itemize}
		\item We argue that $b \in B_{\s_\A(q)}$. Proposition~\ref{prop:value_end_component} gives that there is a $v_H \in [0,1]$ such that all states $q' \in Q_H$ are such that $v_H = v(q') = \lfp(q')$. Therefore, any Nature state $d \in \distribSet_H$ that is compatible with the EC $H$ is such that: $\mu_v(d) = v_H = \mu_\lfp(d)$. Furthermore, note that $\delta(q,\Supp(\s_\A(q)),b) \subseteq \distribSet_H$. Therefore:
		\begin{displaymath}
			\outM_{\gameNF{\formNF_q}{\mu_\lfp}}(\s_\A(q),b) = \sum_{a \in \Supp(\s_\A(q))} \s_\A(q)(a) \cdot \mu_\lfp \circ \delta(q,a,b) = \sum_{a \in \Supp(\s_\A(q))} \s_\A(q)(a) \cdot v_H = v_H = \lfp(q)
		\end{displaymath}
		Furthermore, $\lfp(q) = \va_{\gameNF{\formNF_q}{\mu_{\lfp}}}(\s_\A(q))$ since $\s_\A(q) \in \Opt_\A(\gameNF{\formNF_q}{\mu_\lfp})$. That is, $\outM_{\gameNF{\formNF_q}{\mu_\lfp}}(\s_\A(q),b) = \va_{\gameNF{\formNF_q}{\mu_{\lfp}}}(\s_\A(q))$ and $b \in B_{\s_\A(q)}$.
		\item We argue that $\delta(q,\Supp(\sigma_\A),b) \cap (\Scr_{i-1}(\Bad))_\distribSet = \emptyset$. 	Let $d \in \delta(q,\Supp(\s_\A(q)),b)$ be a Nature state that is compatible with the EC $H$. We have $\Supp(d) \subseteq Q_H$. Furthermore, by minimality of $i$, we have $Q_H \cap \Scr_{i-1}(\Bad) = \emptyset$. Therefore, $\Supp(d) \cap \Scr_{i-1}(\Bad) = \emptyset$. That is, $d \notin (\Scr_{i-1}(\Bad))_\distribSet$. As this holds for all such Nature states $d \in \delta(q,\Supp(\s_\A(q)),b)$, it follows that $\delta(q,\Supp(\sigma_\A),b) \cap (\Scr_{i-1}(\Bad))_\distribSet = \emptyset$.		
	\end{itemize}
	Hence the contradiction. In fact, $Q_H \cap (\cup_{n \in \N} \Scr_n(\Bad)) = \emptyset$. That is, $Q_H \subseteq \lfp^{-1}[0]$. As this holds for all ECs that is not the target $\top$, we can conclude by applying Proposition~\ref{prop:sufficient}.
\end{proof}

Consider now Proposition~\ref{prop:increasing_valuation}. In fact, we prove a slightly more general result on arbitrary non-decreasing 1-Lipschitz functions. 
\begin{proposition}
	\label{prop:increasing_valuation_aux}
	Let $n \geq 1$. Consider a function $f: [0,1]^n \rightarrow [0,1]^n$ that is non-decreasing and $1$-Lipschitz
	. Assume that its lowest fixed point $m \in [0,1]$ is such that, for all $i \in \llbracket 1,n \rrbracket$, we have $m(i) > 0$. Then, for all $\varepsilon > 0$, there exists a valuation $v \in [0,1]^n$ such that $v \preceq \lfp$, $\norm{\lfp}{v} \leq \varepsilon$ and for all $i \in \llbracket 1,n \rrbracket$: $f(v)(i) > v(i)$.
\end{proposition}
\begin{proof}
	First, 	let us show by induction on $k$ the following property $\mathcal{P}(k)$: assume that there exists a vector $w \in [0,1]^Q$ such that $w \preceq \lfp$, $w \preceq f(w)$ and for all $i \in \llbracket 1,n \rrbracket$, $w(q) < f^k(w)(q)$. Then, there exists $w' \in [0,1]^n$ such that $w \preceq w' \preceq \lfp$ and for all $i \in \llbracket 1,n \rrbracket$, $w'(i) < f(w')(i)$. 
	
	The property $\mathcal{P}(1)$ straightforwardly holds. Consider now some $k \geq 1$ and assume that $^\mathcal{P}(k)$ holds and assume that there is a $w \in [0,1]^n$ such that $w \preceq \lfp$, $w \preceq f(w)$ and for all $i \in \llbracket 1,n \rrbracket$, $w(i) < f^{k+1}(w)(i)$. Note that for all $j \in \N$, we have $f^{j}(w) \preceq \lfp$. Now, let $n_= = \{ i \in \llbracket 1,n \rrbracket \mid w(i) = f^{k}(i) \}$ and $n_\uparrow = \llbracket 1,n \rrbracket \setminus n_= = \{ i \in \llbracket 1,n \rrbracket \mid w(i) < f^{k}(w)(i) \}$. We define:
	\begin{displaymath}
		m_= := \min_{i \in n_=} f^{k+1}(w)(i) - f^{k}(w)(i) = \min_{i \in n_=} f^{k+1}(w)(q) - w(q) > 0 
	\end{displaymath}
	and:
	\begin{displaymath}
		m_\uparrow := \min_{i \in n_\uparrow} f^{k}(w)(i) - w(i) > 0 
	\end{displaymath}
	Let $m := \min(m_=,m_\uparrow)$ and $w' \in [0,1]^n$ be such that:
	\begin{itemize}
		\item $\restriction{w'}{n_=} = \restriction{w}{n_=} = \restriction{f^k(w)}{n_=}$;
		\item $\restriction{w'}{n_\uparrow} = \restriction{f^k(w)}{n_\uparrow} - m/2 \succeq \restriction{w}{n_\uparrow}$.
	\end{itemize}
	With this choice, we have $w' \preceq f^k(w) \preceq \lfp$. Furthermore, we have:
	\begin{itemize}
		\item $w \preceq w'$;
		\item $f^k(w) - m/2 \preceq w'$.
	\end{itemize}
	Furthermore, note that $\norm{f^{k+1}(w)}{f(f^k(w) - m/2)} \leq \norm{f^k(w)}{f^k(w)-m/2} = m/2$. Hence, for all $i \in \llbracket 1,n \rrbracket$, we have: $f^{k+1}(w)(i) - m/2 \leq f(f^k(w) - m/2)(i)$. Now, let us show that $w' \preceq f(w')$. Let $i \in \llbracket 1,n \rrbracket$:
	\begin{itemize}
		\item if $i \in n_=$: $w'(i) = w(i) \leq f(w)(i) \leq f(w')(i)$;
		\item if $i \in n_\uparrow$: $w'(i) = f^{k}(w)(i) - m/2 \leq f^{k+1}(w)(i) - m/2 \leq f(f^k(w) - m/2)(i) \leq f(w')(i)$.
	\end{itemize}
	Finally, let us show that, for all $i \in \llbracket 1,n \rrbracket$, we have $w'(i) < f^k(w')(i)$. Let $i \in \llbracket 1,n \rrbracket$.
	\begin{itemize}
		\item if $i \in n_=$: $w'(i) = w(i) \leq f^{k+1}(w)(i) - m < f^{k+1}(w)(i) - m/2 \leq f(f^{k}(w) - m/2)(i) \leq f(w')(i) \leq f^{k}(w')(i)$;
		\item if $i \in n_\uparrow$: $w'(i) = f^k(w)(i) - m/2 < f^k(w)(i) \leq f^k(w')(i)$.
	\end{itemize}
	We can then apply $\mathcal{P}(k)$ on $w'$ to exhibit a vector $w'' \in [0,1]^Q$ such that $w \preceq w' \preceq w'' \preceq \lfp$, $w'' \preceq f(w'')$ and for all $i \in \llbracket 1,n \rrbracket$, $w''(i) < f(w'')(i)$. Overall, $\mathcal{P}(k+1)$ holds and $\mathcal{P}(j)$ holds for all $j \in \N$.
	
	Now, let $\eta := \min_{i \in \llbracket 1,n \rrbracket} \lfp(i) > 0$, $\iota := \min(\eta,\varepsilon) > 0$ and $w \in [0,1]^n$ be the valuation such that for all $i \in \llbracket 1,n \rrbracket$, we have $w(i) := \lfp(i) - \iota < \lfp(i)$. First, let us argue that $w \preceq f(w)$. Assume towards a contradiction that there is some $i \in \llbracket 1,n \rrbracket$ such that $f(w)(i) < w(i)$. Then, $f(w)(i) \leq f(\lfp)(i)$ since $w \preceq \lfp$. Furthermore:
	\begin{displaymath}
		\lfp(i) = f(\lfp)(i) \leq f(w)(i) + \norm{\lfp}{w} < w(i) + \iota = \lfp(q)
	\end{displaymath}
	Hence the contradiction. In fact, $w(i) \leq f(w)(i)$ for all $i \in \llbracket 1,n \rrbracket$. Thus, $w \preceq f(w)$. Now, consider the sequence $(w_n)_{n \in \N}$ defined by $w_0 := w$ and for all $k \in \N$, $w_{k+1} := f(w_k) = f^{k+1}(w_0)$. We have, for all $k \in \N$, $w_k \preceq w_{k+1}$. Hence, this sequence converges. In fact, its limit is equal to $\lfp$ (this directly derives from Kleene fixed-point theorem).
	%
	
	We can conclude that there exists a $k \in \N$ such that, for all $i \in \llbracket 1,n \rrbracket$, we have $w(i) < w_k(i) = f^k(w)(i)$ since $w(i) < \lfp(i)$. We can then apply $\mathcal{P}(k)$ to obtain a valuation $v \in [0,1]^k$ such that $w \preceq v \preceq \lfp$ and for all $i \in \llbracket 1,n \rrbracket$, $f(v)(i) > v(i)$. Furthermore, since $\norm{\lfp}{v} \leq \varepsilon$, we have $\norm{\lfp}{v} \leq \varepsilon$.
\end{proof}

The proof of Proposition~\ref{prop:increasing_valuation} then follows.
\begin{proof}
	Consider some set of states $G \subseteq Q$ such that $(Q \setminus G) \cap \lfp^{-1}[0] = \emptyset$ and $\varepsilon > 0$. 
	The goal is to find a valuation $v \in [0,1]^Q$ such that $v \preceq \lfp$, $\norm{\lfp}{v} \leq \varepsilon$, $\restriction{v}{G} = \restriction{\lfp}{G}$ and for all $q \in Q \setminus G$, $\Delta(v)(q) = \va_{\gameNF{\formNF_q}{\mu_{v}}} > v(q)$.

	Let us define the function $\tilde{\Delta}: [0,1]^Q \rightarrow [0,1]^Q$ by, for all $v \in [0,1]^Q$ and $q \in Q$, $\tilde{\Delta}(v)(q) := \lfp(q)$ if $q \in G$ and $\tilde{\Delta}(v)(q) := \Delta(v)(q)$ otherwise. Note that, as the function $\Delta$, $\tilde{\Delta}$ is non-decreasing and 1-Lipschitz. We can then apply Proposition~\ref{prop:increasing_valuation_aux} to exhibit such a valuation $v$.

\end{proof}

\subsection{Proof of Theorem~\ref{thm:positional_optimal}}
\label{proof:positional_optimal}
\begin{proof}
	Initially, $\Bad_0 = \emptyset \subseteq \SubOS_\A$. Then, by Lemma~\ref{lem:bad_increase}, for all $i \geq 0$, we have $\Bad_{i+1} = Q \setminus \Scr(\Bad_i) \subseteq \SubOS_\A$. In particular, $\Bad = \Bad_n \subseteq \SubOS_\A$. Furthermore, by Lemma~\ref{lem:pos_optimal_strat}, there exists a Player $\A$ positional strategy from all states in $\Scr(\Bad) = Q \setminus \Bad$. Hence, $\Scr(\Bad) \subseteq \OS_\A$. As we have $Q = \Bad \uplus \Scr(\Bad) = \OS_\A \uplus \SubOS_\A$, it follows that: $\Bad = \SubOS_\A$ and $\Scr(\Bad) = \OS_\A$. Then it amounts to applying Lemma~\ref{lem:pos_optimal_strat}.
\end{proof}

\subsection{Complements on infinite games}
\label{appen:infinite_games}
First, note that the sequence of probabilities $(1/2 + 1/2^i)_{i \in \N^*}$ is decreasing and is then well-defined since we have $p_i \in [0,1]$ for all $i \in \N^*$. We consider now the game from the state $q_0$. For all $i \in \N^*$, there is a unique path $\pi_i \in Q^+$ (with a non-zero probability to occur) from $q_0$ to $c_i$ (regardless of the strategies considered): $\pi_i := \pi'_i \cdot c_i \in Q^+$ with $\pi' := q_0 \cdots q_i \in Q^+$. Consider two arbitrary strategies $\s_\A$ and $\s_\B$ for Player $\A$ and $\B$. Let us denote by $v_i \in [0,1]$ (resp. $w_i \in [0,1]$) the value w.r.t. the pair of strategies $(\s_\A,\s_\B)$ of the state $q_i$ (resp. $c_i$): 
\begin{displaymath}
	v_i = \prob{q_i}{\s_\A^{\pi'_i},\s_\B^{\pi'_i}}(\top)
\end{displaymath}
\begin{displaymath}
	w_i = \prob{c_i}{\s_\A^{\pi_i},\s_\B^{\pi_i}}(\top)
\end{displaymath}
Now, for all $k \in \N$ and $l \in \N^*$, we have the following relation between the values $v_k$ and $v_{k+l}$ of the states $q_k \in Q$ and $q_{k+l} \in Q$:
\begin{displaymath}
	v_k = \sum_{j = 0}^{l}
	\frac{1}{2^{j+1}} \cdot w_{k+j} + \frac{1}{2^{l+1}} \cdot v_{k+l}
\end{displaymath}
Given the game form at the states $c_i$, we have $w_i \leq \frac{1}{2}$ for all $i \in \N^*$. In fact, for all $l \in \N^*$, we have:
\begin{align*}
v_0 & = \sum_{j = 0}^{l}
\frac{1}{2^{j+1}} \cdot w_{j} + \frac{1}{2^{l+1}} \cdot v_{l} \\
& = \sum_{j = 0}^{l}
\frac{1}{2^{j+1}} \cdot \frac{1}{2} + \sum_{j = 0}^{l}
\frac{1}{2^{j+1}} \cdot (w_{j} - \frac{1}{2}) + \frac{1}{2^{l+1}} \cdot \frac{1}{2} + \frac{1}{2^{l+1}} \cdot (v_l - \frac{1}{2}) \\
& = \frac{1}{2} + \sum_{j = 0}^{l}
\frac{1}{2^{j+1}} \cdot (w_{j} - \frac{1}{2}) + \frac{1}{2^{l+2}} \cdot (v_l - \frac{1}{2})
\end{align*}
Note that, for all $l \in \N^*$, we have $\sum_{j = 0}^{l}
\frac{1}{2^{j+1}} \cdot (w_{j} - \frac{1}{2}) \leq 0$ with this inequality being strict if and only if there exists an $j \leq l$ such that $w_j < \frac{1}{2}$. Furthermore, $\frac{1}{2^{l+2}} \cdot (v_l - \frac{1}{2}) \underset{l \rightarrow \infty}{\longrightarrow} 0$. It follows that:
\begin{displaymath}
	v_0 \leq \frac{1}{2} \text{ and } (v_0 < \frac{1}{2} \Leftrightarrow \exists j \in \N^*,\; w_j < \frac{1}{2})
\end{displaymath}
Then, we can build a Player $\A$ strategy realizing the value $\frac{1}{2}$ from $q_0$. Indeed, it suffices to play, in $s$, for a sufficiently small $\varepsilon_i > 0$ a $\varepsilon_i$-optimal strategy (ensuring the value at least $1 - \varepsilon_i$) if the state $c_i$ has been previously seen, for some $i \in \N$. Specifically, $\varepsilon_i$ has to be chosen so that $(1 - \varepsilon_i) \cdot (1/2 + 1/2^i) \geq 1/2$. With this choice, we have $\val{\Aconc}{\s_\A^{\pi_i \cdot s_i}}(s_i) \geq \frac{1}{2}$ and it follows that $\val{\Aconc}{\s_\A^{\pi_i}}(c_i) = \frac{1}{2}$, which ensures $\val{\Aconc}{\s_\A}(q_0) = \frac{1}{2}$ for all $i \in \N$. However, a Player $\A$ positional strategy $\s_\A$ is $\varepsilon$-optimal in $s$ for some fixed $\varepsilon > 0$ that does not depend on the state $c_i$ seen. It follows that there is some $i$ such that $\val{\Aconc}{\s_\A}(s_i) < \frac{1}{2}$ and $\val{\Aconc}{\s_\A}(c_i) < \frac{1}{2}$. Then, we can conclude that $\val{\Aconc}{\s_\A}(q_0) < \frac{1}{2}$.

\subsection{Complements on computing the set of maximizable states (Theorem~\ref{thm:optimal_state_computable})}
\label{appen:computing_optimal_states}
First, we have that the value of a game in normal form can be encoded in the first order theory of the reals.
\begin{proposition}[Value of a game in normal form in the theory of the reals]
	\label{prop:value_gnf_theory_reals}
	Consider a game in normal form $\formNF$ and a value $v$. The fact that $v = \va_{\formNF}$ can be encoded in $\mathsf{FO}$-$\R$. This encodes the predicate $\mathsf{VAL}(\formNF,v)$.
\end{proposition}
\begin{proof}
	A strategy for Player $\A$ is encoded via a probability associated with each available action with the constraint that the sum is equal to $1$, and similarly for Player $\B$. Then, we have $v = \va_{\formNF}$ if and only if there exists a Player $\A$ strategy $\sigma_\A$ whose value (i.e. the minimum over all actions available $b$ to Player $\B$ of the outcome of $\sigma_\A$ and $b$) is at least $v$, and similarly for a Player $\B$ strategy whose value has to be at most $v$. We assume that $\St_\A = \llbracket 1,n \rrbracket$ and $\St_\B = \llbracket 1,k \rrbracket$, the predicate $\mathsf{ValGF}(\formNF,v)$ can be encoded with the $\mathsf{FO}$-$\R$ formula:
	\begin{align*}
		\exists p_1,\ldots,p_n,\; & q_1,\ldots,q_k: \\
		& \bigwedge_{1 \leq i \leq n} (0 \leq p_i \leq 1) \wedge \sum_{i = 1}^{n} p_i = 1 \; \wedge \\
		& \bigwedge_{1 \leq j \leq k} (0 \leq q_j \leq 1) \wedge \sum_{j = 1}^{k} q_j = 1 \; \wedge \\ 
		& \bigwedge_{1 \leq j \leq k} \sum_{i = 1}^{n} p_i \cdot \outCNF(i,j) \geq v \; \wedge \\
		& \bigwedge_{1 \leq i \leq n} \sum_{j = 1}^{k} q_j \cdot \outCNF(i,j) \leq v
	\end{align*}
\end{proof}

Consider now a concurrent reachability game. We would like to encode, once two positional strategies for Player $\A$ and Player $\B$ are fixed, the value of the states in $\mathsf{FO}$-$\R$. A game where both strategies are fixed corresponds to a Markov chain. It can also be seen as another concurrent reachability game where the players have only one possible action. In that case, the value of the states is given by the least fixed point of the function $\Delta$, which can be encoded in $\mathsf{FO}$-$\R$, as stated below.
\begin{proposition}[Value in a Markov chain in the theory of the reals]
	\label{prop:value_reach_game_theory_reals}
	Consider a reachability game $\langle \Aconc,\top \rangle$ with rational distribution, a valuation $w \in [0,1]^Q$ of the states, and two positional strategies $\s_\A,\s_\B$ for both players. The fact that $w = \lfp = \val{\Aconc}{}$ can be encoded in $\mathsf{FO}$-$\R$. This encodes the predicate $\mathsf{ValReachGame}(\Aconc,\s_\A,\s_\B,w)$.
\end{proposition}
\begin{proof}
	We assume that $Q = \llbracket 1,n \rrbracket$ with $\top = n$. Consider an input valuation $i \in [0,1]^Q$. We encode the predicate $\mathsf{Val}^\Delta(i,o)$ stating that $o = \Delta(i)$ with the $\mathsf{FO}$-$\R$ formula below:
	\begin{align*}
		o(n) = 1 \; \wedge \bigwedge_{1 \leq i \leq n-1} 0 \leq o(i) \leq 1 \; \wedge \bigwedge_{1 \leq i \leq n-1} \mathsf{ValGF}(\formNF_i,o(i))
	\end{align*}
	with
	\begin{displaymath}
		\formNF_i = \langle A,B,\distribSet,\outCNF \rangle
	\end{displaymath}
	and\footnote{Note that this is in $\mathsf{FO}$-$\R$ because the distribution is rational.}
	\begin{displaymath}
		\outCNF(a,b) = \s_\A(i)(a) \cdot \s_\A(i)(b) \cdot \sum_{i = 1}^{n-1} \distribFunc(\delta(i,a,b)) \cdot v(i) \in [0,1]
	\end{displaymath}
	Note that this outcome function is directly encoded in the predicate $\mathsf{ValGF}(\formNF_i,o(i))$.
	
	It can then be encoded that $w$ is the least fixed point of the function $\Delta$, that is the predicate $\mathsf{ValReachGame}(\Aconc,\s_\A,\s_\B,w)$:
	\begin{align*}
		\mathsf{Val}^\Delta(w,w),\; \wedge \forall u, \bigwedge_{1 \leq i \leq n} 0 \leq u(i) \leq w(i) \wedge u \neq w \Rightarrow \lnot \mathsf{Val}^\Delta(u,u)
	\end{align*}
	Note that $u$ can be represented by a sequence $u_1,\ldots,u_n$ of values of the states. In this formula, we check that $w$ is a fixed point and that no point point smaller that $w$ is a fixed point. 
\end{proof}

We can proceed to the proof of Theorem~\ref{thm:optimal_state_computable}.
\begin{proof}
	Consider a concurrent reachability game $\langle \Aconc,\top \rangle$ with rational distribution and a state $q \in Q$. Theorem~\ref{thm:positional_optimal} gives that if $q \in \OS_\A$, there is a positional Player $\A$ strategy that is optimal from $q$. Furthermore, as already proved in~\cite{everett57}, Player $\A$ has positional $\varepsilon$-optimal strategy for all $\varepsilon > 0$. Hence:
	\begin{align*}
		q \in \OS_\A & \Leftrightarrow \exists \s_\A \in \SetPosStrat{\Aconc}{\A},\; \val{\Aconc}{\s_\A}(q) = \val{\Aconc}{}(q) \\
		& \Leftrightarrow \exists u \in [0,1],\; \exists \s_\A \in \SetPosStrat{\Aconc}{\A},\; (\val{\Aconc}{\s_\A}(q) \geq u) \wedge \forall \s_\A' \in \SetPosStrat{\Aconc}{\A},\; (u \geq \val{\Aconc}{\s_\A'}(q)) \\
	\end{align*}
	
	Now, Proposition~\ref{prop:value_reach_game_theory_reals} gives that the predicate $\mathsf{ValReachGame}(\Aconc,\s_\A,\s_\B,v)$ can be encoded as an $\mathsf{FO}$-$\R$ formula for all positional strategies $\s_\A$ and $\s_\B$ and valuation $v$ (since the distribution of Nature states is rational). This induces the following $\mathsf{FO}$-$\R$ formula:
	\begin{align*}
		\exists u,\; (\mathsf{Guarantee}(u,q) \;  \wedge \; \mathsf{AtMost}(u,q))
	\end{align*}
	with 
	\begin{displaymath}
		\mathsf{Guarantee}(u,q) := \exists \s_\A \in \SetPosStrat{\Aconc}{\A},\; \forall \s_\B \in \SetPosStrat{\Aconc}{\B},\; \exists v \in [0,1]^Q,\;  \mathsf{ValReachGame}(\Aconc,\s_\A,\s_\B,v) \wedge  u \leq v(q)
	\end{displaymath}
	and 
	\begin{displaymath}
		\mathsf{AtMost}(u,q) := 	\forall \s_\A' \in \SetPosStrat{\Aconc}{\A},\; \exists \s_\B \in \SetPosStrat{\Aconc}{\B},\; \exists v \in [0,1]^Q,\; \mathsf{ValReachGame}(\Aconc,\s_\A,\s_\B,v) \wedge v(q) \leq u
	\end{displaymath}
	We quantify over Player $\A$  and Player $\B$ positional strategies as a shortcut for quantifying over the $|Q| \cdot |A|$ and $|Q| \cdot |B|$ necessary variables to encode them (on all states, we have a set of probabilities on all actions whose sum is equal to $1$).
\end{proof}

\section{Complements on Section~\ref{sec:optimal_game_forms}}
\label{appen:optimal_game_forms}

\subsection{Proof of Lemma~\ref{prop:one_state_game_opt_equiv_correct}}
\label{proof:one_state_game_opt_equiv_correct}
First, let us formally define the three-state reachability game induced by a game form and a partial valuation of the outcomes.
\begin{definition}[One-shot reachability game]
	Let $\formNF = \langle \St_\A,\St_\B,\outComeNF,\outCNF \rangle$ be a game form and $\alpha: \outComeNF \setminus E \rightarrow [0,1]$ be a partial valuation of the outcomes. The three-state reachability game $\langle \Aconc_{(\formNF,\alpha)},\top \rangle$ induced by $\formNF$ and $\alpha$ is such that $\Aconc_{(\formNF,\alpha)} = \AConc$ with:
	\begin{itemize}
		\item $A := \St_\A$ and $B := \St_\B$;
		\item $Q := \{ q_0,\top,\bot \}$;
		\item $\distribSet := \distribSet_{q_0} \cup \{ d_{\mathsf{loop}},\top_{\mathsf{loop}},\bot_{\mathsf{loop}} \}$ with $\distribSet_{q_0} := \{ d_x \mid x \in \outComeNF \setminus E \}$;
		\item for $x \in \outComeNF \setminus E$, we have $\distribFunc(d_x)(\top) := \alpha(x)$ and $\distribFunc(d_x)(\bot) := 1 - \alpha(x)$; $\distribFunc(q_\mathsf{loop})(q_0) := 1$, $\distribFunc(\top_\mathsf{loop})(\top) := 1$ and $\distribFunc(\bot_\mathsf{loop})(\bot) := 1$.
		\item for all $a \in A$ and $b \in B$, we have $\delta(\top,a,b) := \top_\mathsf{loop}$, $\delta(\bot,a,b) := \bot_\mathsf{loop}$. Furthermore, let us define the function $g: \outComeNF \rightarrow \distribSet$ by, for all $x \in \outComeNF$, we have:
		\[ g(x) := \begin{cases} 
		d_\mathsf{loop} & \text{ if } x \in E \\
		d_{o} & \text{ otherwise }
		\end{cases}
		\]
		This function associates to each outcome its corresponding Nature state. For all $a \in A$ and $b \in B$, we set $\delta(q_0,a,b) := g \circ \outCNF(a,b)$.
	\end{itemize}
\end{definition}

Let us now proceed to the proof of Lemma~\ref{prop:one_state_game_opt_equiv_correct}.
\begin{proof}
	Let us consider a game form $\formNF$ and a partial valuation of the outcomes $\alpha: \outComeNF \setminus E \rightarrow [0,1]$. The values, in the one-shot reachability game $\langle \Aconc_{(\formNF,\alpha)},\top \rangle$, of the states are given by the valuation $\lfp \in [0,1]^Q$. First, let us show that the value $
	\lfp(q_0)$ of the state $q_0$ in the one-shot reachability game $\langle \Aconc_{(\formNF,\alpha)},\top \rangle$ is equal to the least fixed point $v_\alpha$ of the function $f_\alpha^\formNF$. For all $u \in [0,1]$, let us denote by $\lfp[u] \in [0,1]^Q$ the valuation of the states such that $\lfp[u](\top) := 1 = \lfp(\top)$, $\lfp[u](\bot) := 0 = \lfp(\bot)$ and $\lfp[u](q_0) := u$. In particular, we have $\lfp[\lfp(q_0)] = \lfp$. Furthermore:
	\begin{itemize}
		\item for all $x \in \outComeNF \setminus E$, $\mu_{\lfp[u]} \circ g(x) = \mu_{\lfp[u]}(d_x) = \alpha(x) \cdot \lfp[u](\top) + (1 - \alpha(x)) \cdot \lfp[u](\bot) = \alpha(x) = \alpha[u](x)$;
		\item for all $x \in E$, $\mu_{\lfp[u]} \circ g(x) = \mu_{\lfp[u]}(d_\mathsf{loop}) = \lfp[u](q_0) = u = \alpha[u](x)$ .
	\end{itemize}
	In fact:
	\begin{displaymath}
		\mu_{\lfp[u]} \circ g = \alpha[u]
	\end{displaymath}
	Now, let $u \in [0,1]$. We have $g \circ \outCNF = \delta(q_0,\cdot,\cdot)$ and $\mu_{\lfp[u]} \circ \delta(q_0,\cdot,\cdot) = \alpha[u] \circ \outCNF$. It follows that $\gameNF{\formNF}{\alpha[u]} = \gameNF{\formNF_{q_0}}{\mu_{\lfp[u]}}$. Hence:
	\begin{displaymath}
		f_\alpha^\formNF(u) = \va_{\gameNF{\formNF}{\alpha[u]}} = \va_{\gameNF{\formNF_{q_0}}{\mu_{\lfp[u]}}} = \Delta(\lfp[u])(q_0)
	\end{displaymath}
	That is, for all fixed point $l \in [0,1]$ of the function $f_\alpha^\formNF$, we have $\Delta(\lfp[l])(q_0) = l = \lfp[l](q_0)$. Thus, $\Delta(\lfp[l]) = \lfp(l)$. That is, the least fixed point of the function $f_\alpha^\formNF$ is equal to the value of the least fixed point of $\Delta$ in $q_0$: $v_\alpha = \lfp(q_0)$.
	
	Second, consider at which condition is the state $q_0$ maximizable (straightforwardly, all other states are maximizable in any case). We consider the first iteration of the construction of the set of bad states. We have $\Bad_0 = \emptyset$ and $\Bad_1 = Q \setminus \Scr(\Bad_0) = Q \setminus  \Scr(\emptyset)$. In fact, $q_0$ is maximizable if and only if $q_0 \in \Scr(\emptyset)$. In other words, $q_0$ is maximizable if and only if either $\lfp(q_0) = v_\alpha = 0$ or $\Eff_q(\top_\distribSet,\emptyset) = \Prg_q(\top_\distribSet) \neq \emptyset$. Assume that $\lfp(q_0) = v_\alpha > 0$. Let us show that progressive strategies are exactly reach-maximizing strategies w.r.t. to the valuation $\alpha$. Consider a strategy $\sigma_\A \in \Opt_\A(\gameNF{\formNF}{\limval{\alpha}}) = \Opt_\A(\gameNF{\formNF_{q_0}}{\mu_\lfp})$ and let $b \in B_{\sigma_\A}$. There are two possiblities:
	\begin{itemize}
		\item either, for all $a \in \Supp(\sigma_\A)$, we have $\limval{\alpha} \circ \outCNF(a,b) > 0$. In that case:
		\begin{align*}
			\outCNF(\Supp(\sigma_\A),b) \not \subseteq E & \Leftrightarrow \outCNF(\Supp(\sigma_\A),b) \cap (\outComeNF \setminus E) \neq \emptyset \\
			& \Leftrightarrow g \circ \outCNF(\Supp(\sigma_\A),b) \cap g[\outComeNF \setminus E] \neq \emptyset \\
			& \Leftrightarrow \delta(q_0,\Supp(\sigma_\A),b) \cap \distribSet_{q_0} \neq \emptyset \\
			& \Leftrightarrow \exists a \in \Supp(\sigma_\A), \delta(q_0,\Supp(\sigma_\A),b) \in \distribSet_{q_0} \\
			& \Leftrightarrow \exists a \in \Supp(\sigma_\A), \distribFunc(\delta(q_0,a,b))(\top) = \limval{\alpha}(\outCNF(a,b)) > 0 \\
			& \Leftrightarrow \delta(q_0,\Supp(\sigma_\A),b) \cap \top_{\distribSet} \neq \emptyset \\
		\end{align*}
		Overall, we have $\outCNF(\Supp(\sigma_\A),b) \not \subseteq E \Leftrightarrow \delta(q_0,\Supp(\sigma_\A),b) \cap \top_{\distribSet} \neq \emptyset$.
		\item or, there is $a_0 \in \Supp(\sigma_\A)$ such that $\limval{\alpha} \circ \outCNF(a_0,b) = 0$. Since we have $\outM_{\gameNF{\formNF}{\limval{\alpha}}}(\sigma_\A,b) = \va_{\gameNF{\formNF}{\limval{\alpha}}} = v_\alpha > 0$, it follows that there exists $a \in \Supp(\sigma_\A)$, such that $\limval{\alpha} \circ \outCNF(a,b) > v_\alpha > 0$. As $\limval{\alpha}[E] = \{ v_\alpha \}$ by definition of $\limval{\alpha}$, it follows that $\outCNF(a,b) \notin E$. Hence, $\outCNF(\Supp(\sigma_\A),b) \not \subseteq E$. Furthermore, $\distribFunc(\delta(q_0,a,b))(\top) = \limval{\alpha}(\outCNF(a,b)) > 0$. That is, $\delta(q_0,\Supp(\sigma_\A),b) \cap \top_{\distribSet} \neq \emptyset$.
	\end{itemize}
	Overall, for all $b \in B_{\sigma_\A}$, we have the equivalence $\outCNF(\Supp(\sigma_\A),b) \not \subseteq E \Leftrightarrow \delta(q_0,\Supp(\sigma_\A),b) \cap \top_{\distribSet} \neq \emptyset$. That is, the strategy $\sigma_\A$ is progressive if and only if it is reach-maximizing w.r.t. the valuation $\alpha$. In fact, we have that the state $q_0$ is maximizable if and only the game form $\formNF$ is RM w.r.t. the partial valuation $\alpha$.
\end{proof}



\subsection{Proof of Lemma~\ref{lem:correct_opt_everywhere}}
\label{proof:lem_correct_opt_everywhere}
Before proving this lemma, we state and prove a sufficient condition for the limit of a partial of the Nature states to be equal to $\mu_\lfp \in [0,1]^\distribSet$.
\begin{proposition}
	\label{prop:lim_partial_val}
	Consider a concurrent reachability game $\langle \Aconc,\top \rangle$ with the values of the states given by the valuation $\lfp \in [0,1]^Q$. Consider a value $x \in [0,1]$, a non-empty set of states $Q_x \subseteq Q$ such that, for all $q \in Q_x$ we have $\lfp(q) = x$ and a set of Nature states $\distribSet_x \subseteq \distribSet$ such that, for all $d \in \distribSet_x$, we have $\Supp(d) \subseteq Q_x$. Then, considering the partial valuation $\alpha: \distribSet \setminus \distribSet_x \rightarrow [0,1]$ of the Nature states such that $\alpha := \restriction{\mu_\lfp}{\distribSet \setminus \distribSet_x}$, there is a state $q \in Q_x$ such that $\limval{\alpha} = \mu_\lfp$ in the game form $\formNF_q$.
\end{proposition}
\begin{proof}
	First, note that for all $d \in \distribSet_x$, we have:
	\begin{displaymath}
		\mu_\lfp(d) = \sum_{q \in Q} \distribFunc(d)(q) \cdot \lfp(q) = \sum_{q \in \Supp(d)} \distribFunc(d)(q) \cdot \lfp(q) = \sum_{q \in \Supp(q)} \distribFunc(d)(q) \cdot x = x = \alpha[x](d)
	\end{displaymath}
	Hence, we have $\alpha[x] = \mu_\lfp$. 
	
	Now, for any state $q \in Q$, we denote by $v_\alpha^q \in [0,1]$ the least fixed point of the function $f_\alpha^{\formNF_q}$, i.e. $\limval{\alpha} = \alpha[v_\alpha^q]$ in the game form $\formNF_q$. 
	Hence, for all $q \in Q_x$, we have $f_\alpha^{\formNF_q}(x) = \va_{\gameNF{\formNF_q}{\mu_\lfp}} = \lfp(q) = x$. That is, $x$ is a fixed point of the function $f_\alpha^{\formNF_q}$ and thus $v_\alpha^q \leq x$. Let us show that there exists a state $q \in Q_x$ such that $v_\alpha^q = x$. Let $v = \max_{q \in Q_x} v_\alpha^q$. For all $q \in Q_x$, by Observation~\ref{obs:smaller_valuation_smaller_outcome}, we have:
	\begin{equation}
	\label{eqn:f_one_LP}
	f_\alpha^{\formNF_q}(v) = \va_{\gameNF{\formNF_q}{\alpha[v]}} \leq \va_{\gameNF{\formNF_q}{\alpha[v_\alpha^q]}} + (v - v_\alpha^q) =  f_\alpha^{\formNF_q}(v_\alpha^q) + (v - v_\alpha^q) = v_\alpha^q + (v - v_\alpha^q) = v
	\end{equation}
	
	Let us show that this implies $x \leq v$. Consider the iterative sequence of valuations $(v_n)_{n \in \N} \in ([0,1]^Q)^\N$ whose limit by Proposition~\ref{prop:lim_vn} is equal to $\lfp$. (Recall that $v_0 \in [0,1]^Q$ is such that $v_0(\top) = 1$ and $v_0(q) = 0$ for all $q \neq \top$. Furthermore, for all $n \geq 0$, $v_{n+1} = \Delta(v_n)$). We show inductively that for all $n \in \N$, we have:
	\begin{displaymath}
	\forall q \in Q_x,\; v_n(q) \leq v
	\end{displaymath} 
	This holds straightforwardly for $n = 0$
	. Now, assume that this holds for some $n \geq 0$. Consider some 
	Nature state $d \in \distribSet_x$. By assumption, we have $\Supp(d) \subseteq Q_x$. Hence:
	\begin{displaymath}
	\mu_{v_{n}}(d) = \sum_{q \in Q_x} \distribFunc(d)(q) \cdot v_n(q) \leq \sum_{q \in Q_x} \distribFunc(d)(q) \cdot v = v = \alpha[v](d)
	\end{displaymath}
	
	This holds for all $d \in \distribSet_x$. In addition, we have $v_n \preceq \lfp$, therefore $\mu_{v_n} \preceq \mu_{\lfp}$. In particular, for all $d \in \distribSet \setminus \distribSet_x$, we have $\mu_{v_{n}}(d) \leq \mu_{\lfp}(d) = \alpha(d)$. Overall, we have: $\mu_{v_{n}} \preceq \alpha[v]$. Now, if we consider some $q \in Q_x$, we have:
	\begin{align*}
		v_{n+1}(q) & = \Delta(v_{n})(q) & \text{ by definition of }v_{n+1} \\
		& = \va_{\langle \formNF_q,\mu_{v_{n}} \rangle} & \text{ by definition of }\Delta\\
		& \leq \va_{\langle \formNF_q,\alpha[v] \rangle} & \text{ since }\mu_{v_{n}} \preceq  \alpha[v]\\ & = f_\alpha^{\formNF_q}(v) 
		\leq v & \text{ by (\ref{eqn:f_one_LP})}
	\end{align*}
	Therefore, for any state $q \in Q_x$, we have for all $n \in \N$, $v_{n}(q) \leq v$. It follows that $x = \lfp(q) = \lim\limits_{n \rightarrow \infty} v_n(q) \leq v$. That is, there exists a state $q \in Q_x$ such that $x \leq v_\alpha^q$. In fact, $v_\alpha^q = x$ and, in the game form $\formNF_q$, we have $\limval{\alpha} = \alpha[v_\alpha^q] = \alpha[x] = \mu_\lfp$.
\end{proof}

We can now proceed to the proof of Lemma~\ref{lem:correct_opt_everywhere}.
\begin{proof}
	Let us consider a concurrent reachability game $\langle \Aconc,\top \rangle$ whose local interactions are all RM game forms. We want to show that $Q = \OS_\A$. By theorem~\ref{thm:positional_optimal}, this is equivalent to having $\Bad = \SubOS_\A = Q \setminus \OS_\A = \emptyset$. Hence, let us show that $\Bad = \emptyset$. To do so, we prove that $\Bad_1 = \Bad_0 = \emptyset$ (see Definition~\ref{def:bad_states}). That is, we assume towards a contradiction that $Q \setminus (\Scr_n(\emptyset) \cup \mu_{\lfp}[0]) \neq \emptyset$ for $n = |Q|$. 
	We want to use the assumption that the local interactions are RM to exhibit a progressive strategy among states in $Q \setminus (\Scr_n(\emptyset) \cup \mu_{\lfp}[0])$. To achieve this, we consider the maximum $x$, over all states in $Q \setminus \Scr_n(\emptyset)$, of the valuation $\lfp$: $x := \max_{q \in Q \setminus \Scr(\emptyset)} \lfp(q) > 0$ (by assumption). We consider also the corresponding set of states $Q_x \subseteq Q \setminus \Scr_n(\emptyset)$ and Nature states $\distribSet_x \subseteq \distribSet \setminus \Scr_n(\emptyset)_\distribSet$ realizing this value w.r.t. $\lfp \in [0,1]^Q$ and $\mu_\lfp \in [0,1]^\distribSet$. (Recall that $\Scr_n(\emptyset)_\distribSet$ refers to the set of Nature states with a non-zero probability to reach a state in $\Scr_n(\emptyset)$: $\Scr(\emptyset)_\distribSet := \{ d \in \distribSet \mid \Supp(d) \cap \Scr_n(\emptyset) \neq \emptyset \}$). That is: 
	\begin{displaymath}
		Q_x := \lfp^{-1}[x] \setminus \Scr(\emptyset) \neq \emptyset
	\end{displaymath}
	and:
	\begin{displaymath}
		\distribSet_x := \mu_{\lfp}^{-1}[x] \setminus \Scr_n(\emptyset)_\distribSet
	\end{displaymath}
	
	For the partial valuation of the Nature states $\alpha: \distribSet \setminus \distribSet_x \rightarrow [0,1]$ ensuring $\alpha := \restriction{\mu_\lfp}{\distribSet \setminus \distribSet_x}$, we want to apply Proposition~\ref{prop:lim_partial_val} to show that there exists a state $q \in Q_x$ such that $\limval{\alpha} = \mu_\lfp$ in the game form $\formNF_q$. Let us show that the support of all Nature states $d \in \distribSet_x$ is included in $Q_x$. Let $d \in \distribSet_x$. We have:
	\begin{displaymath}
		x = \mu_\lfp(d) = \sum_{q \in Q} \distribFunc(d)(q) \cdot \lfp(q) = \sum_{q \in \Supp(d)} \distribFunc(d)(q) \cdot \underbrace{\lfp(q)}_{\leq x} \leq \sum_{q \in \Supp(q)} \distribFunc(d)(q) \cdot x = x
	\end{displaymath}
	Therefore, all the above inequalities are in fact equalities. That is, for all states $q \in \Supp(d)$, we have $\lfp(q) = x$ and $d \notin \Scr(\emptyset)$, since $\Supp(d) \cap \Scr_n(\emptyset) = \emptyset$ (otherwise, we would have $d \in \Scr_n(\emptyset)_\distribSet$). That is, $q \in Q_x$. Overall, we obtain $\Supp(d) \subseteq Q_x$. This holds for all $d \in \distribSet_x$.
	
	We can now apply Proposition~\ref{prop:lim_partial_val} to obtain a state $q \in Q_x$ such that $\limval{\alpha} = \mu_\lfp$ in the game form $\formNF_q$. Let us now exhibit a progressive strategy --  w.r.t. the set $\Gex = \Scr(\Bad)$ -- in the local interaction $\formNF_q$. Let $\sigma_\A \in \Opt_\A(\formNF_q^{\mu_{\lfp}})$ be a local strategy that is reach-maximizing w.r.t. the partial valuation $\alpha$. 
	Let $b \in B_{\sigma_\A}$. 
	That is, since $=v_\alpha = \lfp(q) > 0$, we have $\delta(q,\Supp(\sigma_\A),b) \subsetneq \distribSet_x$. Now, let $d \in \distribSet_x \setminus \delta(q,\Supp(\sigma_\A),b)$. There are two possibilities:
	\begin{itemize}
		\item either $\mu_\lfp(q) \geq x$. In that case, by maximality of $x$ and since $d \notin \distribSet_x$, this implies that $d \in \Scr_n(\emptyset)_\distribSet$.
		\item or $\mu_\lfp(d) < x = \va_{\gameNF{\formNF_q}{\mu_\lfp}} = \outM_{\gameNF{\formNF_q}{\mu_\lfp}}(\sigma_\A,b)$. In that case, this implies that there exists $d' \in \delta(q,\Supp(\sigma_\A),b)$ such that $\mu_\lfp(d') > x$. By maximality of $x$, this implies $d' \in \Scr_n(\emptyset)_\distribSet$.
	\end{itemize}
	In any case, we have $\delta(q,\Supp(\sigma_\A),b) \cap \Scr_n(\emptyset)_\distribSet \neq \emptyset$. As this holds for all optimal actions $b \in B_{\sigma_\A}$, it follows that the strategy $\sigma_\A$ is progressive, and therefore efficient: $\sigma_\A \in \Prg_q(\Scr_n(\emptyset)) = \Eff_q(\Scr_n(\emptyset),\Bad)$. Hence the contradiction with the fact that $q \notin \Scr(\emptyset)$. 
	
	In fact, we have $\Bad_1 = \Bad_0 = \emptyset$, i.e. $\Bad = \emptyset$. Overall, $Q \setminus \OS_\A = \emptyset$ or $\OS_\A = Q$.
\end{proof}


\subsection{Proof of Theorem~\ref{thm:main}}
\label{proof:main}
\begin{proof}
	Consider a set $\mathcal{G}$ of local interactions (or game forms). By Lemma~\ref{prop:one_state_game_opt_equiv_correct}, if it contains a game form $\formNF$ that is not RM, then we can build a three-state reachability game with $\formNF$ as local interaction in the initial state where that initial state is not maximizable. Furthermore, as soon as all local interactions in $\mathcal{G}$ are RM, by Lemma~\ref{lem:correct_opt_everywhere}, all concurrent reachability game built from local interactions in $\mathcal{G}$ have only maximizable states. This proves the equivalence.
\end{proof}

\subsection{Decidability of the fact that game forms are RM (Proposition~\ref{prop:decidability})}
\label{proof:decidability}

\begin{proof}
	Consider a game form $\formNF$ and assume that $\St_\A = \llbracket 1,n \rrbracket$, $\St_\B = \llbracket 1,k \rrbracket$ and $\outComeNF = \llbracket 1,l \rrbracket$. A partial valuation $\alpha: \outComeNF \setminus E \rightarrow [0,1]$ for some subset of outcomes $E \subseteq \outComeNF$ is encoded with a sequence $\alpha = \alpha_1,\ldots,\alpha_l$ of values of the outcomes and a sequence $e = e_1,\ldots,e_l$ of binary values encoding the fact that an outcome is in $\outComeNF \setminus E$. For a value $v \in [0,1]$, the valuation $\alpha[v]: \outComeNF \rightarrow [0,1]$ is equal to $\alpha \cdot e + v \cdot (1-e)$. Finally, recall that for a game in normal form $\formNF'$ and a value $v \in [0,1]$ the predicate $\mathsf{VAL}(\formNF,v)$ encodes in $\mathsf{FO}$-$\R$ the fact that $v = \va_{\formNF'}$. Now, the fact that the game form $\formNF$ is RM is expressed by the $\mathsf{FO}$-$\R$ formula:
	\begin{align*}
	\forall \alpha = & \alpha_1,\ldots,\alpha_l,\; \forall e = e_1,\ldots,e_l,\; \\ 
	& (\mathsf{IsPartialVal}(\alpha,e) \wedge \exists v,\;  \mathsf{IsLimPartialVal}(\alpha,e,v) \; \wedge \\
	& \exists p = p_1,\ldots,p_n,\; \mathsf{IsProba}(p) \; \wedge \; 	\mathsf{IsOptimal}(\alpha,e,v,p) \wedge \mathsf{IsReachMaximizing}(e,p))
	\end{align*}
	with
	\begin{displaymath}
		\mathsf{IsPartialVal}(\alpha,e) := \bigwedge_{1 \leq i \leq l} ((0 \leq \alpha_i \leq 1) \wedge (e_i = 1) \lor (e_i = 0)) \;
	\end{displaymath}
	and 
	\begin{align*}
		\mathsf{IsLimPartialVal}(\alpha,e,v) := \; & 0 \leq v \leq 1 \wedge \mathsf{VAL}(\gameNF{\formNF}{\alpha \cdot e + v \cdot (1-e)},v) \; \wedge \; \\ 
		& \forall v',\; 0 \leq v' < v \Rightarrow \lnot \mathsf{VAL}(\gameNF{\formNF}{\alpha \cdot e + v' \cdot (1-e)},v')
	\end{align*}
	and 
	\begin{displaymath}
		\mathsf{IsStrategy}(p) := \bigwedge_{1 \leq i \leq n} (0 \leq p_i \leq 1) \; \wedge \; \sum_{i = 1}^{n} p_i = 1
	\end{displaymath}
	and 
	\begin{displaymath}
		\mathsf{IsOptimal}(\alpha,e,v,p) := \bigwedge_{1 \leq j \leq k} \sum_{i = 1}^{n} p_i \cdot ( \alpha_{\outCNF(i,j)} \cdot e_{\outCNF(i,j)} + v \cdot (1 - e_{\outCNF(i,j)})) \geq v
	\end{displaymath}
	and 
	\begin{displaymath}
		\mathsf{IsReachMaximizing}(e,p) := \bigwedge_{1 \leq j \leq k} (\bigvee_{1 \leq i \leq n} p_i \wedge e_{\outCNF(i,j)})
	\end{displaymath}
	The formula consists in:
	\begin{itemize}
		\item a universal quantification over partial valuations;
		\item the existence of the least fixed point of the function $f_\alpha^\formNF$;
		\item the existence of a Player $\A$ strategy;
		\item that ensures at least $v$ w.r.t. the valuation $\alpha[v]$ (it is therefore optimal by the predicate $\mathsf{VAL}(\gameNF{\formNF}{\alpha \cdot e + v \cdot (1-e)},v)$);
		\item ensuring that for all columns, there is at least one outcome in the support of that strategy that is in $\outComeNF \setminus E$.
	\end{itemize}
\end{proof}

\subsection{Proof of Proposition~\ref{prop:determined_gf}}
\label{appen:determined_game_form}
\begin{proof}
	Assume that the game form $\formNF$ is determined and consider a partial valuation  $\alpha: \outComeNF \setminus E \rightarrow [0,1]$ of the outcomes. Let us prove that $v_\alpha = f^\formNF_\alpha(0)$ and that the game form $\formNF$ is RM w.r.t. $\alpha$. If $v_\alpha = f^\formNF_\alpha(v_\alpha) = 0$, this holds straightforwardly. Assume now that $v_\alpha > 0$, and therefore $f^\formNF_\alpha(0) > 0$. We set $v := f^\formNF_\alpha(0) \in [0,1]$ and $v' := f^\formNF_\alpha(v)$. Since $f^\formNF_\alpha$ is an non-decreasing function, it follows that $v = f^\formNF_\alpha(0) \leq f^\formNF_\alpha(v) = v'$. Assume towards a contradiction that $v < v'$. Let $E_v \subseteq \outComeNF$ denote the subset of outcomes $o$ such that $\alpha[v](o) \leq v < v'$. Note that, in particular, $E \subseteq E_v$. Since $v' = f^\formNF_\alpha(v)$, there is no $b \in \St_\B$ such that $\outCNF(\St_\A,b) \subseteq E_v$ (which would imply $f^\formNF_\alpha(v) < v'$). Hence, by determinacy of the game form $\formNF$, there is some $a \in \St_\A$ such that $\outCNF(a,\St_\B) \subseteq \outComeNF \setminus E_v \subseteq \outComeNF \setminus E$. Hence, the valuations $\alpha[0]$ and $\alpha[v]$ coincide on $\outCNF(a,\St_\B) \subseteq \outComeNF$ and for all $o \in \outCNF(a,\St_\B)$, we have $\alpha[0](o) > v$. That is, $v = f^\formNF_\alpha(0) \leq \va_{\gameNF{\formNF}{\alpha[0]}}(a) 
	> v$. Hence, the contradiction. Thus, $v = v' = f^\formNF_\alpha(v)$. It follows that $v_\alpha \leq v$. Moreover, $v_\alpha = f^\formNF_\alpha(v_\alpha) \geq f^\formNF_\alpha(0) = v$. That is, $v_\alpha = v = f^\formNF_\alpha(0)$. Furthermore, for $E_0$ the set of outcomes $o$ such that $\alpha[0](o) < v$ (which includes $E$), we can show, by determinacy of the game form $\formNF$, that there is some $a \in \St_\A$ such that $\outCNF(a,\St_\B) \subseteq \outComeNF \setminus E_0$. The strategy $a \in \St_\A$ is then reach-maximizing w.r.t. the partial valuation $\alpha$.
	
	Assume now that the game form $\formNF$ is not determined. There exists a subset of outcomes $E \subseteq \outComeNF$ such that:
	\begin{itemize}
		\item for all $a \in \St_\A$, there exists $b_a \in \St_\B$, such that $\outCNF(a,b_a) \in \outComeNF \setminus E$;
		\item for all $b \in \St_\B$, there exists $a_b \in \St_\A$, such that $\outCNF(a_b,b) \in E$.
	\end{itemize}
	Consider the partial valuation $\alpha: E = \outComeNF \setminus (\outComeNF \setminus E) \rightarrow [0,1]$ such that, for all $o \in E$, we have $\alpha(o) := 1$. Straightforwardly, we have $f_\alpha^\formNF(1) = 1$. Furthermore, consider some $v \in [0,1]$ such that $v < 1$. We have:
	\begin{itemize}
		\item let $n := |\St_\A|$. We define the strategy $\sigma_\A \in  \Dist(\St_\A)$ playing uniformly over all lines of the game form: for all $a \in \St_\A$, we set $\sigma_\A(a) := \frac{1}{n}$. Note that, for all $o \in \outComeNF$, we have $\alpha[v](o) \geq v$. Consider now some $b \in \St_\B$. Recall that $\outCNF(a_b,b) \in E$ and $\alpha[v] \circ \outCNF(a_b,b)  = 1$. In fact:
		\begin{displaymath}
			\outM_{\gameNF{\formNF}{\alpha[v]}}(\sigma_\A,b) = \sum_{a \in \St_\A} \frac{1}{n} \cdot \alpha[v] \circ \outCNF(a,b) = \frac{1}{n} \cdot \sum_{a \in \St_\A \setminus \{ a_b \}} \alpha[v] \circ \outCNF(a,b) + \frac{1}{n} \cdot \alpha[v] \circ \outCNF(a_b,b) \geq \frac{n-1}{n} \cdot v + \frac{1}{n} > v
		\end{displaymath}
		This holds for all $b \in \St_\B$. It follows that $f_\alpha^\formNF(v) \geq \va_{\gameNF{\formNF}{\alpha[v]}}(\sigma_\A) > v$.
		\item We proceed similarly to the previous item with a strategy for Player $\B$. Let $k := |\St_\B|$. We define the strategy $\sigma_\B \in  \Dist(\St_\B)$ playing uniformly over all columns of the game form: for all $b \in \St_\B$, we set $\sigma_\B(b) := \frac{1}{n}$. Consider some $a \in \St_\A$. Recall that $\outCNF(a,b_a) \in \outComeNF \setminus E$ and $\alpha[v] \circ \outCNF(a,b_a)  = v$. In fact:
		\begin{displaymath}
		\outM_{\gameNF{\formNF}{\alpha[v]}}(a,\sigma_\B) = \sum_{b \in \St_\B} \frac{1}{k} \cdot \alpha[v] \circ \outCNF(a,b) = \frac{1}{k} \cdot \sum_{b \in \St_\B \setminus \{ b_a \}} \alpha[v] \circ \outCNF(a,b) + \frac{1}{k} \cdot \alpha[v] \circ \outCNF(a,b_a) \leq \frac{n-1}{n} + \frac{1}{n} \cdot v < 1
		\end{displaymath}
		It follows that $f_\alpha^\formNF(v) \leq \va_{\gameNF{\formNF}{\alpha[v]}}(\sigma_\B) < 1$.
	\end{itemize}
	That is, for all $v \in [0,1]$ such that $v < 1$, we have:
	\begin{displaymath}
		v < f_\alpha^\formNF(v) < 1
	\end{displaymath}
	In fact, $v_\alpha = 1$ and $f_\alpha^\formNF(0) < v_\alpha$.
\end{proof}

	
\end{document}